 \newcount \pointsize \pointsize=10
 \newif \ifmycolour \mycolourfalse

 \documentclass{lmcs}

 \usepackage{qsymbols}
 \usepackage{mathpazo}
 \usepackage{times}
 \usepackage{inference}
 \usepackage{color} \definecolor{darkred}{rgb}{.8,0,0}
 \usepackage{amssymb}

 \usepackage{enumerate}
 \usepackage{hyperref}

 \def\ie{\emph{i.e.}}

 \newcommand{\Comment}[1]{}

\usepackage{pdftexcorrection}

\def\Steffen#1{}\def\Ugo#1{}\def\Franco#1{}

\newcommand{\vect}[1]{\vec{#1}}

 \usepackage{url}
 \usepackage{Macros}

 \usepackage{pdftexcorrection}

 \def\LC{the $`l$-calculus}

\begin{document}

 \title[Intersection Types for $`l`m $] {Intersection Types for the $`l`m$-Calculus}

 \author[S. van Bakel et al.]{Steffen van Bakel}
 \address{Department of Computing, Imperial College London, 180 Queen's Gate, London SW7 2BZ, UK}	
 \email{s.vanbakel@imperial.ac.uk} 
 \thanks{Authors have been partially supported by COST Action EUTYPES CA-15123, Project FIR 1B8C1 of the University of Catania (the second author) and Project FORMS 2015 of the University of Turin (the third author).}

 \author[]{Franco Barbanera}	
 \address{Dipartimento di Matematica e Informatica, Universit\`a degli Studi di Catania, Viale A. Doria 6, 95125 Catania, Italy}	
 \email{barba@dmi.unict.it} 

 \author[]{Ugo de'Liguoro}
 \address{Dipartimento di Informatica, Universit\`a degli Studi di Torino, Corso Svizzera 185, 10149 Torino, Italy}
 \email{ugo.deliguoro@unito.it}

 \keywords{$`l`m$-calculus, intersection types, filter semantics, strong normalisation.}
 \subjclass{F.4 MATHEMATICAL LOGIC AND FORMAL LANGUAGES, F.4.1 Lambda calculus and related systems}

 \begin{abstract}
We introduce an intersection type system for the $`l`m$-calculus that is invariant under subject reduction and expansion.
The system is obtained by describing Streicher and Reus's denotational model of continuations in the category of $`w$-algebraic lattices via Abramsky's domain-logic approach.
This provides at the same time an interpretation of the type system and a proof of the completeness of the system with respect to the continuation models by means of a filter model construction.

We then define a restriction of our system, such that a $`l`m$-term is typeable if and only if it is strongly normalising. 
We also show that Parigot's typing of $`l`m$-terms with classically valid propositional formulas can be translated into the restricted system, which then provides an alternative proof of strong normalisability for the typed $`l`m$-calculus.

 \end{abstract}

 \maketitle

 \section*{Introduction} \label{sec:intro}

The $`l`m$-calculus is a type-free calculus introduced by Parigot \cite{Parigot'92} to denote classical proofs and to compute with them. 
It is an extension of the proofs-as-programs paradigm where types can be understood as classical formulas and (closed) terms inhabiting a type as the respective proofs in a variant of Gentzen's natural deduction calculus for classical logic \cite{Gentzen'35}.
The study of the syntactic properties of the $`l`m$-calculus has been challenging, which led to the introduction of variants of term syntax, reduction rules, and typing as, for example, in de Groote's variant of the $`l`m$-calculus \cite{deGroote'94}.
These changes have an impact on the deep nature of the calculus which emerges both in the typed and in the untyped setting \cite{David-Py'01,Saurin'08}.

Types are of great help in understanding the computational properties of terms in an abstract way.
Although in \cite{Barendregt'84} Barendregt treats the theory of the pure $`l$-calculus without a reference to types, most of the fundamental results of the theory can be exposed in a quite elegant way by using the Coppo-Dezani intersection type system \cite{Coppo-Dezani'80}.
This is used by Krivine \cite{Krivine-book'93}, where the treatment of the pure $`l$-calculus relies on intersection typing systems called $ \mathcal D$ and $\Der \Omega$.

The quest for more expressive notions of typing for $`l`m$ is part of an ongoing investigation into calculi for classical logic.
In order to come to a characterisation of strong normalisation for Curien and Herbelin's (untyped) sequent calculus $ \lmmt$ \cite{Curien-Herbelin'00}, Dougherty, Ghilezan and Lescanne presented System $ \MIU$ \cite{DGL-ITRS'04,DGL-CDR'08}, that defines a notion of intersection and union typing for that calculus.
However, in \cite{Bakel-APAL'10} van Bakel showed that this system is not closed under conversion, an essential property of Coppo-Dezani systems; in fact, it is shown that it is \emph{impossible} to define a notion of typing for $ \lmmt$ that satisfies that property.

In \cite{Bakel-ITRS'10} van Bakel brought intersection (and union) types to the context of the (untyped) $`l`m$-calculus, and showed that for $`l`m$-conversion \emph{it is possible} to prove type preservation under conversion.
However, union types are no longer dual to intersection types and play only a marginal role, as was also the intention of \cite{DGL-CDR'08}.
In particular, the normal $(\unI)$ and $(\unE)$ rules as used in \cite{Barbanera-Dezani-Liguoro-IaC'95}, which are known to create a soundness problem in the context of the $`l$-calculus, are not allowed.
In the view of the above mentioned failure noted in \cite{Bakel-APAL'10}, the result of \cite{Bakel-ITRS'10} came as a surprise, and led automatically to the question we answer here: does a \emph{filter semantics} for $`l`m$ exist?

The idea of building a $`l$-model out of a suitable type assignment system appeared first in \cite{BCD'83}. In that system types are an extension of simple types with the binary operator $\inter$ for intersection, and are pre-ordered by an axiomatisable (in fact decidable) relation $\seq$; if types are interpreted by subsets of the domain $D$ (an applicative structure satisfying certain conditions), one can see $\inter$ as set theoretic intersection and $\seq$ as containment. The discovery of \cite{BCD'83} is that, taking a proper relation $\seq$, the set $\Filt_D$ of filters of types (sets of types closed under type intersection and $\seq$) is a $`l$-model, where (closed) terms can be interpreted by the set of types that can be assigned to them. This is what is called a filter semantics.

It emerged in \cite{Coppo-et.al'84} that models constructed as set of filters of intersection types are exactly the $`w$-algebraic lattices, a category of complete lattices, but with Scott-continuous maps as morphisms. 
$`w$-algebraic lattices are posets whose structure is fully determined by a countable subset of elements, called `compact points' for topological reasons.
Now the crucial fact is that given an $`w$-algebraic lattice $D$, the set $\Compact(D)$ of its compact points can be described by putting its elements into a one-to-one correspondence with a suitable set of intersection types, in such a way that the order over $\Compact(D)$ is reflected by the inverse of the $\seq$ pre-order over types. 
Then one can show that the filter structure $\Filt_D$ obtained from the type pre-order is isomorphic with the original $D$. 
In fact, Abramsky proved that this is not true only of $`w$-algebraic lattices, but of quite
larger categories of domains, like $2/3$-SFP domains that can be finitely described by a generalisation of intersection type theories, called the \emph{logics} of the respective domains in \cite{Abramsky'91}.

Here, instead of defining a suitable type system for $`l`m$, and then trying to prove that it actually induces a filter model, we follow the opposite route. 
We start from a model of the $`l`m$-calculus in $\ALG$, the category of $\omega$-algebraic lattices. 
We then distill the type syntax and the corresponding type theory out of the construction of the model, and recover the typing rules from the clauses that define term interpretation in the given model $\Filt_D$ that is by construction isomorphic to the given $D$.

However, things are 
more complex than this. 
First we need a domain theoretic model of $`l`m$; we use for that purpose Streicher and Reus's \emph{models of continuations}. 
Building on Lafont's ideas and the papers \cite{Lafont-Reus-Streicher'93,Ong-Stewart'97}, in \cite{Streicher-Reus'98} Streicher and Reus proposed a model of both typed and untyped $`l$-calculi embodying a concept of continuation, including Felleisen's $`l{\mathcal C}$-calculus \cite{Felleisen-Friedman-Kohlbecker'86,Felleisen-PhD'87} and a version of Parigot's $`l`m$.
The model is based on the solution of the domain equations $D = C\To R$ and $C = D \times C$, where $R$ is an arbitrary domain of `results'. 
The domain $C$ is set of what are called `continuations' in \cite{Streicher-Reus'98}, which are infinite tuples of elements in $D$. 
$D$ is the domain of continuous functions from $C$ to $R$ and is the set of `denotations' of terms. 
We call the triple $(R,D,C)$ a $`l`m$-model, that exists in $\ALG$ by the inverse limit technique, provided that $R\ele \ALG$. 

The next step is to find type languages $\Lang_D$ and $\Lang_C$, and type theories axiomatising the respective pre-orders $\seqD$ and $\seq_C$, such that $D$ and $C$ are isomorphic to $\Filt_D$ and $\Filt_C$, respectively. 
To this aim we may suppose that logical description of $R$ is given via a language of types $\Lang_R$ and a pre-order $\seq_R$.
But then we need a detailed analysis of $\Compact(D)$ and $\Compact(C)$, keeping into account that $D$ and $C$ are both co-limits of certain chains of domains, and that their compact points are into one-to-one correspondence with the union of the compact points of the domains approximating $D$ and $C$.
This leads us to a mutually inductive definition of $\Lang_D$ and $\Lang_C$ and of $\seqD$ and $\seq_C$.
In this way, we obtain an extension of the type theory used in \cite{BCD'83}, which is a \emph{natural equated} intersection type theory in terms of \cite{Alessi-Severi'08} and hence is isomorphic to the inverse limit construction of a $D_ \infty$ $`l$-model (as an aside, we observe that this matches perfectly with Theorem 3.1 in \cite{Streicher-Reus'98}).

Once the filter domains $\Filt_D$ and $\Filt_C$ have been constructed, we can consider the interpretation of terms and of commands (`unnamed' and `named terms' respectively in Parigot's terminology). 
Following \cite{Streicher-Reus'98}, we define the interpretation of expressions of Parigot's $`l`m$-calculus inductively via a set of equations.
Guided by these equations in the particular case of $\Filt_D$ and $\Filt_C$, and considering the correspondence we have established among types and compact points, we are able to reconstruct the inference rules of a type assignment system which forms the main contribution of our work.

The study of the properties of the system produces a series of results that confirm the validity of the construction. 
First we prove that in the filter model the meaning of $M$ in the environment $e$, denoted by $\Sem{ M }{ e }$, coincides with the filter of all types $`d\ele \Lang_D$ such that $ \derLmu `G |- M : `d | `D $ is derivable in the system, for $`G$ and $`D $ such that $e$ satisfies both $`G$ and $`D $, and similarly for $\Sem{\Cmd}{ e }$, where $\Cmd$ is a command. 
Then if two terms or commands are convertible, they must have the same types. 
In fact, we will prove this result twice: first abstractly, making essential use of the filter model construction; then concretely, by studying in depth the structure of the derivations in our system, and establishing that types are preserved under subject reduction and expansion.

We then face the problem of characterising strong normalisation in the case of $`l`m$. 
This is a characteristic property of intersection types, stated the first time by Pottinger \cite{Pottinger'80} for the ordinary $`l$-calculus: strongly normalising terms can be captured by certain `restricted' type systems, ruling out the universal type $`w$.
As will be apparent in the technical treatment, we cannot simply restrict our system by removing $`w$; however, the characterisation can be obtained by distinguishing certain `good' occurrences of $`w$ that cannot be eliminated, and the `bad' ones that must be forbidden. 
This is still guided by the semantics and by the proof theoretic study of the system, and we can establish that there exists a subsystem of our system that is determined just by a restriction on type syntax, plus the elimination of the rule $(`w)$ from our type system.

We conclude by looking at the relation between our type system and the original one proposed by Parigot \cite{Parigot'92} on the basis of the Curry-Howard correspondence between types and formulas and $`l`m$-terms and proofs of classical logic. 
We show that there exists an interpretation of Parigot's first order types into intersection types such that the structure of derivations is preserved; moreover, since the translations are all restricted intersection types, we obtain a new proof that all proof-terms in $`l`m$, \ie~those typeable in Parigot's system (even extended with negation), are strongly normalising. \\

\noindent
This paper is the full version of \cite{Bakel-Barbanera-deLiguoro-TLCA'11}, extended with a revised version of \cite{BakBdL-ITRS12}.

 \subsection*{Outline of this paper}
The paper is organised as follows. 
After recalling the $`l`m$-calculus in Sect.\skp\ref{sec:calculus}, we study the domain theoretic models in Sect.\skp\ref{sec:semantics}. 
In Sect.\skp\ref{sec:filter} we introduce intersection types and type theories and we illustrate the filter model construction.
The main part of the paper is Sect.\skp\ref{sec:types}, where we introduce the type assignment system; we study type invariance under reduction and expansion in Sect.\skp\ref{sec:closureUnderConv}.
Sect.\skp\ref{sec:character} is devoted to the characterisation of strongly normalising terms by means of a subsystem of ours obtained by suitably restricting the type syntax. 
Then, in Sect.\skp\ref{sec:Parigot}, we compare our system with Parigot's, and show that Parigot's types are translatable into our restricted types while preserving type derivability (in the two systems). 
We finish by discussing some related work in Sect.\skp\ref{sec:related} and draw our conclusions.

 \section{The untyped $`l`m$-calculus} \label{sec:calculus}
The $`l`m$-calculus, as introduced in \cite{Parigot'92}, is an extension of the untyped $`l$-calculus obtained by adding \emph{names} and a name-abstraction operator, $`m$. 
It was intended as a proof calculus for a fragment of classical logic. 
Logical formulas of the implicational fragment of the propositional calculus can be assigned as types to $`l`m$-terms much in the formulae-as-types paradigm of the Curry-Howard correspondence between typed $`l$-calculus and intuitionistic logic. 
With $`l`m$ Parigot created a multi-conclusion typing system. 
In the notation of \cite{Saurin'08}, the derivable statements have the shape $\derLmu `G |- M : A | `D $, where $A$ is the main conclusion of the statement, expressed as the \emph{active} conclusion, and $`D $ contains the alternative conclusions, consisting of pairs of names and types; the left-hand context $`G$, as usual, is a mapping from term variables to types, and represents the assumptions about free term variables of $M$.

As with implicative intuitionistic logic, the reduction rules for the terms that represent proofs correspond to proof contractions; the difference is that the reduction rules for {\LC} are the \emph{logical} reductions, \ie~deal with the elimination of a type construct that has been introduced directly above. 
In addition to these, Parigot expressed also the \emph{structural} rules, where elimination takes place for a type constructor that appears in one of the alternative conclusions (the Greek variable is the name given to a subterm): he therefore needed to express that the focus of the derivation (proof) changes, and this is achieved by extending the syntax with two new constructs $[`a]M$ and $`m`a.M$ that act as witness to \emph{deactivation} and \emph{activation}, which together move the focus of the derivation.

$\lmu$ was conceived in the spirit of Felleisen's $`l{\mathcal C}$-calculus, that Griffin showed to be typeable with classical propositional logic in \cite{Griffin'90}. 
%
The $`l`m$-calculus is type free and uses names and $`m$ to model a form of functional programming with control \cite{deGroote'94}.

Here we briefly revise the basic notions of the $`l`m$-calculus, though slightly changing the notation and terminology, and defer the presentation of the typed $`l`m$-calculus to Sect.\skp\ref{sec:Parigot}.

 \begin{defi}[Term Syntax \cite{Parigot'92}] 
The sets $ \Terms$ of \emph{terms} (ranged over by $M$, $N$, $L$) and $ \Commands$ of \emph{commands} (ranged over by $\Cmd$) are defined inductively by the following grammar, where $x \ele \TVar$, the set of \emph{term variables} (ranged over by $x$, $y$, $z$) and $`a \ele \CVar$, the set of \emph{names} (ranged over by $`a$, $`b$, $`g$), both denumerable:
 \[ \begin{array}{rcl@{\dquad}l}
M,N & :: = &x \mid `lx.M \mid MN \mid `m`a.\Cmd & (\textrm{terms}) \\
 \Cmd & :: = & [`a]M & (\textrm{commands})
 \end{array} \]
We let $T$ range over $\Terms \union \Commands$.

As usual, $`l x.M$ binds $x$ in $M$, and $`m `a .\Cmd$ binds $`a $ in $\Cmd$. 
A variable or a name occurrence is free in a term if it occurs and is not bound: we denote the free variables and the free names occurring in $T$ by $\fv(T)$ and $\fn(T)$, respectively.
 \end{defi}

We identify terms and commands obtained by renaming of free variables or names, and we adopt Barendregt's convention that free and bound variables and names are distinct, assuming silent $`a$-conversion during reduction to avoid capture.
We will extend this convention to also consider occurrences of $x$ and $`a$ bound over $M$ in type judgements like $\derLmu `G,x{:}A |- M : B | `a{:}C,`D $ (see Sect.\skp\ref{sec:types}).

\Comment{
Note that we could have defined the set of \emph{pure $`l`m$-terms} via the grammar:
 \[ \begin{array}{rcl}
M,N & :: =& x \mid `lx.M \mid MN \mid `m`a.[`b]M
 \end{array} \]
but it is convenient to have commands as part of the syntax as well.
}

In \cite{Parigot'92} terms and commands are called `terms' and `named terms', respectively; names are called $`m$-variables, but might be better understood as `continuation variables' (see \cite{Streicher-Reus'98}).
Since this would imply a 
commitment to a particular interpretation, we prefer a more neutral terminology.

In the $`l`m$-calculus, substitution takes the following three forms:
 \[ \begin{array}{l@{\dquad}ll}
 \textit{term substitution:} & T[N/x] & (\textrm{$N$ is substituted for $x$ in $T$}) \\
 \textit{renaming:} & T[`a/`b] & (\textrm{$`b$ in $T$ is replaced by $`a$}) \\
 \textit{structural substitution:} & T \strSub[`a<=L] & (\textrm{every subterm $[`a]N$ of $T$ is replaced by $[`a]NL$})
 \end{array} \]
In particular, \emph{structural substitution} is defined by induction over terms and commands as follows:

 \begin{defi}[Structural Substitution] \label{def:substitution}
The key case for the structural substitution is defined as:
 \[ \begin{array}{rcl}
([`a]M) \strSub[`a<=L] & \ByDef & [`a](M \strSub[`a<=L])L 
 \end{array} \]
The other cases are defined as:
 \[ \begin{array}{rcl@{\quad}l}
x \strSub[`a<=L] & \ByDef & x \\
(`lx.M) \strSub[`a<=L] & \ByDef & `lx.M \strSub[`a<=L] \\
(MN) \strSub[`a<=L] & \ByDef & (M \strSub[`a<=L])(N \strSub[`a<=L]) \\
 \end{array} \]
 \[ \begin{array}{rcl@{\quad}l}
(`m`b.\Cmd) \strSub[`a<=L] & \ByDef & `m`b.\Cmd \strSub[`a<=L] \\
([`b]M) \strSub[`a<=L] & \ByDef & [`b]M \strSub[`a<=L] & (\textrm{if }`a \neq`b )
 \end{array} \]
 \end{defi}
Notice that the first case places the argument of the substitution to the right of a term with name $`a$, and the others propagate the substitution towards subterms that are named $`a$.

The reduction relation for the $`l`m$-calculus is defined as follows.

 \begin{defi}[Reduction $\redbm$ \cite{Parigot'92}] \label{def:reduction}
The reduction relation $\redbm$ is the compatible closure of the following rules:
 \[ \begin{array}{r@{\dquad}rcll}
(`b): & (`lx.M)N & \reduces & M[N/x] \\
(`m): & (`m`a.\Cmd)N & \reduces & `m`a.\Cmd\StrSub{`a}{N} \footnotemark \\
(\Rename): & [`a]`m`b.\Cmd & \reduces & \Cmd[`a/`b] .
 \end{array} \]
 \end{defi}
\footnotetext{\def\Cmd{\textrm{\sf C}}%
This is the common notation for this rule, although one could argue that a better formulation would be: $(`m`a.\Cmd)N \reduces `m`g.\Cmd\StrSub{`a}{N{`.}`g}$ where $`g$ is fresh, and let the structural substitution rename the term: $ ([`a]M) \strSub[`a<=L{`.}`g] = [`g](M \strSub[`a<=L{`.}`g])L $; in fact, when making the substitution \emph{explicit} (see \cite{Bakel-Vigliotti-IFIPTCS'12}), this becomes necessary.
This is reflected in Ex.\skp\ref{ex:mu-self-app}, where before the reduction $(`m`a.[`a]x)x \reduces `m `a.[`a]xx$, $`a$ has type $`d\prod`k$, and after it has type $`k$.}

Note that Rule $(`b)$ is the normal $`b$-reduction rule of the $`l$-calculus. 
Rule $(`m)$ is characteristic for $`l`m$; the intuition behind this rule has been explained by de Groote in \cite{deGroote'94}, by arguing on the grounds of the intended typing of the $`m$-terms: `\emph{in a $`l`m$-term $`m`a.M$ of type $A\arrow B$, only the subterms named by $`a$ are \emph{really} of type $A\arrow B$ (\ldots); hence, when such a $`m$-abstraction is applied to an argument, this argument must be passed over to the sub-terms named by $`a$.}'. 
The `renaming' rule $(\Rename)$ is called `structural reduction' in \cite{Parigot'92} and rule $(`r)$ in \cite{Py-PhD'98}; it is an auxiliary notion of reduction, aimed at simplifying proof terms.

 \begin{defi}[The theory $ \LamMuTheory $] \label{def:theories}
The theory $ \LamMuTheory$ is the equational theory determined by the compatible closure of: $M \redbm N \Then \derLM M = N $.
 \end{defi}

Py \cite{Py-PhD'98} has shown that 
$ \redbm$ is confluent.
Therefore the convertibility relation $=_{`b`m}$ determined by $ \redbm$ is consistent in the usual sense that distinct normal forms are not equated.

 \section{$`l`m$-models and term interpretation} \label{sec:semantics}

As is the case for the $`l$-calculus, in order to provide a semantics to the untyped $`l`m$-calculus we need to look for a domain $D$ and a mapping $\Sem{ \cdot}^D$ such that $\Sem{ M }^D{ e } \ele D$ for each term $M$, where $e$ maps variables to terms and names to continuations. 
Since the interpretation of terms depends on the interpretation of names and commands, we need an auxiliary domain $C$ and a mapping $\Sem{ \cdot}^C$ such that $e \psk `a \ele C$ for any name $`a$, and $\Sem{ \Cmd}^C{ e } \ele C$. The term interpretation is a \emph{model of the theory} $`l`m$ if $\Sem{ M }^D = \Sem{ N }^D$ whenever $\derLM M = N$.

The semantics we consider here is due to Streicher and Reus \cite{Streicher-Reus'98}, but for a minor variant explained below.
The idea is to work in the category ${\mathcal N}_R$ of `negated' domains of the shape $A\To R$, where $R$ is a parameter for the domain of results.
In such a category, continuations are directly modelled and treated as the fundamental concept, providing a semantics both to Felleisen's $ `l{\mathcal C}$-calculus and to a variant of $ `l`m$ that has, next to the two sorts of term we consider here (terms and commands) also \emph{continuation} terms.

Here we adapt that semantics to Parigot's original $ `l`m$. 
We rephrase the model definition in the setting of the normal categories of domains, obtaining something similar to Hindley-Longo `syntactical models.' 
Our models are essentially a particular case of the definitions in \cite{Ong'96,Hofmann-Streicher'97}.

 \begin{defi}[$ `l`m$-Model] \label{def:lambdaMuModel} \label{eq:domain}
A triple ${\ModM} = (R,D,C)$ is a $ `l`m$-model in a category of domains $ \mathcal D$ if $R \ele \Der$ is a fixed domain of \emph{results} and $D$ and $C$ (called domains of \emph{denotations} and of \emph{continuations}, respectively) are solutions in $ \mathcal D$ of the equations:
 \[ \begin{cases}{lll}
D &=& C \SemArrow R \\
C &=& D \times C
 \end{cases} \]
In the terminology of \cite{Streicher-Reus'98}, elements of $D$ are \emph{denotations}, while those of $C$ are \emph{continuations}.
We refer to the above equations as the \emph{continuation domain equations}.
We let $k$ range over $C$, and $d$ over $D$.
 \end{defi}

 \begin{rem} \label{rem:ExtLamMod}
If $(R,D,C)$ is a $`l`m$-model then $C$ is (isomorphic to) the infinite product $D \times D \times D \times \cdots$. 
On the other hand, as observed in \cite{Streicher-Reus'98} \S 3.1, we also have:
 \[ \begin{array}{rcl@{\quad}lrcl@{\quad}lc}
D &\Isom& C \To R &\Isom& (D \times C) \To R &\Isom& D \To (C \To R) &\Isom& D \To D. 
 \end{array} \]
since categories of domains are cartesian closed.
Therefore, a $ `l`m$-model as defined in Def.\skp\ref{def:lambdaMuModel} is an extensional $`l$-model.
 \end{rem}

 \begin{defi}[Term Interpretation] \label{def:interpretation} 
Let ${\ModM}=(R,D,C)$ be a $ `l`m$-model.

 \begin{enumerate}

 \item
We define $\Env = (\TVar \SemArrow D) \times (\CVar \SemArrow C)$ and call elements of $\Env$ \emph{environments}; 
We write $e(x) = e_1(x)$ and $e(`a) = e_2(`a)$ for $e = \Pair{e_1}{e_2} \ele \Env$, and $\EnvM$ for the set of environments interpreting variables and names into $\ModM$.

 \item
We define an \emph{environment update} as: 
 \[ \begin{array}{rcl}
e [x\to d] \psk y & = & 
 \begin{cases}{l@{\dquad}l}
d & (x = y) \\
e \psk y & (\textrm{otherwise})
 \end{cases}
 \\ [4mm]
e [`a\to k] \psk `b & = & 
 \begin{cases}{l@{\dquad}l}
k & (`a = `b) \\
e \psk `b & (\textrm{otherwise})
 \end{cases}
 \end{array} \]

 \item
The \emph{interpretation mappings}
$ \Sem{ \cdot}_{\ModM}^{D }: \Terms \SemArrow \Env \SemArrow D$ and
$ \Sem{ \cdot}_{\ModM}^{C}: \Commands \SemArrow \Env \SemArrow C$, written $\Sem{ \cdot}^D$ and $\Sem{ \cdot}^C$ when $\ModM$ is understood, are mutually defined by the equations:
 \[ \def\arraystretch{1.2} \begin{array}{rcl@{\dquad}l}
 \Sem{ x}^D{ e }{ k }
	&=& e ~ x ~ k \\
 \Sem{ `lx.M}^D{ e }{ k }
	&=& \Sem{ M }^D{ e[x \to d] }{ k' } & (k = \Pair<d,k'>) \\
 \Sem{ MN }^D { e }{ k }
	&=& \Sem{ M }^D{ e }{ \Pair<\Sem{ N }^D { e } , k > } \\
 \Sem{ `m`a.\Cmd}^D{ e }{ k }
	&=& d ~ k' & (\Pair<d , k'> = \Sem{ \Cmd}^C{ e[`a \to k]}) \\
 \Sem{ [`a]M}^C{ e }{ \textcolor{white}{ k }}
	&=& \Pair<\Sem{ M }^D { e } , e \psk `a >
 \end{array} \]

 \end{enumerate}
 \end{defi}
This definition has a strong similarity with Bierman's interpretation of $`l`m$ \cite{Bierman'98}; however, he considers a \emph{typed} version.
In the second equation of the definition of $ \Sem{ \cdot}^D$, the assumption $k = \Pair<d,k'>$ is not restrictive: in particular, if $k = \bot_C = \Pair<\bot_D , \bot_C>$, then $d = \bot_D$ and $k' = k = \bot_C$.

The last two equations differ from those in \cite{Bierman'98} and \cite{Streicher-Reus'98} since there the interpretation of a command is a result:
 \[ \begin{array}{rcl}
\Sem{ [`a]M}_*^R \, e &=& \Sem{ M }_*^D \, e \, (e \psk `a) 
 \end{array} \]
and consequently 
 \[ \begin{array}{rcl}
\Sem{ `m`a.\Cmd}_*^D \, e \, k &=& \Sem{ \Cmd}_*^R \, e[`a \to k],
 \end{array} \]
writing $\Sem{ \cdot}_*^A$ for the resulting interpretation maps. This is not an essential difference however:
let $e' = e[`a \to {k}]$, then
 \[ \begin{array}{rcccl}
\Sem{ `m`a.[`b]M}_*^{D }{e}{k} &=& \Sem{ [`b]M}_*^{R}\, e' &=& \Sem{ M}^{D } \, e' (e'\, `b).
 \end{array} \]

On the other hand, by Definition \ref{def:interpretation}:
 \[ \begin{array}{rcl}
\Sem{ [`b]M}^{C}{e}'&=& \Pair{\Sem{ M}^{D}{e'}}{(e' \, `b) }, 
 \end{array} \]
so
 \[ \begin{array}{rcl}
\Sem{ `m`a.[`b]M}^{D }{e}{k} &=& \Sem{ M}^{D }{e'}{(e'\, `b)}.
 \end{array} \]
Therefore we can show $\Sem{ `m`a.[`b]M}_*^{D } = \Sem{ `m`a.[`b]M}^{D }$ by induction.

The motivation for interpreting commands into continuations instead of results is that the latter are elements of the parametric domain $R$; hence
in the system of Sect.\skp\ref{sec:types} results do not have significant types. On the other hand, by choosing our interpretation of commands we get explicit typing of commands with continuation types. 
Conceptually this could be justified by observing that commands are peculiar `evaluation contexts' of the $`l`m$-calculus, and continuations have been the understood as evaluation contexts since Felleisen's work.
Technically, here we have just a variant of the treatment of, for example, \cite{deLiguoro:ApproxLM12}, which system is based on the more standard interpretation.
 
Below, we fix a $`l`m$-model $\ModM$, and we shall write $\Sem{ \cdot}^{\ModM}$ or simply $\Sem{ \cdot}$ by
omitting the superscripts $C$ and $D$ whenever clear from the context.

 \begin{lem} \label{lem:beta-weak} \label{lem:eta-sem}

 \begin{enumerate}

 \firstitem \label{lem:eta-sem i}
If $x \notele \fv(M)$, then $\Sem{ M }{ e } = \Sem{ M }{ e[x \to d] } $, for all $d \ele D$.

 \item \label{lem:eta-sem ii}
If $`a \notele \fn(M)$, then $\Sem{ M }{ e } = \Sem{ M }{ e[`a \to k] } $, for all $k \ele C$.

 \end{enumerate}
 \end{lem}

 \begin{Proof} Easy. 
 \end{Proof}

We now establish the relation between the various kinds of substitution and the interpretation.

 \begin{lem} \label{lem:sem_substitution_renaming}
 $ \begin{array}[t]{rcll}
 \Sem{ M[N/x] }{ e }{ k } &=& \Sem{ M }{ e[x \to \Sem{ N }{ e }] }{ k } \\
 \Sem{ T[`a/`b] }{ e }{ k } &=& \Sem{ T }{ e[`b \to e\psk`a] }{ k } & \textrm{where }T\ele \Terms \cup \Commands \\
 \Sem{ M \StrSub{`a}{N} }{ e }{ k } &=& \Sem{ M }{ e[`a \becomes \Pair< \Sem{ N }{ e } , e\psk`a > ] }{ k }
 \end{array} $

 \end{lem}

 \begin{proof}
By induction on the definition of (structural) substitution.
The only non-trivial case is when $M \equiv `m `b. [`a]L$ with $`b \neq `a$, so that $(`m `b. [`a]L)\StrSub{`a}{N} \equiv `m`b.[`a]L\StrSub{`a}{N} N$.
By unravelling definitions we have:
 \[ \begin{array}{rcl}
 \Sem{ `m`b.[`a]L\StrSub{`a}{N} N }{ e }{ k } 
	&=& 
d'\, k'
 \end{array} \]
where
 \[ \begin{array}{rcl}
 \Pair{d' }{ k' } 
	&=& 
 \Sem{ [`a]L\StrSub{`a}{N} N }{ e[`b \to k]} \\
	&=& 
 \Pair{\Sem{ L\StrSub{`a}{N} N }{ e[`b \to k]} }{ e[`b \to k] \psk `a},
 \end{array} \]
observing that $e[`b \to k] \psk `a = e\psk`a$, since $`b \neq `a$. 
Then:
 \[ \begin{array}{rcl@{\quad}l}
\Sem{ `m`b.[`a]L\StrSub{`a}{N} N }{ e }{ k } & = & 
\Sem{ L\StrSub{`a}{N} N }{ e[`b \to k] }{ (e\psk `a) } \\ & = & 
\Sem{ L\StrSub{`a}{N} }{ e[`b \to k] }{ \Pair<\Sem{ N }{ e [`b \to k] } , e\psk`a > } \\ & = &
\Sem{ L }{ e[`b \to k, `a \to \Pair<\Sem{ N }{ e } , e\psk`a > ] }{ \Pair<\Sem{ N }{ e } , e\psk`a > }
 \end{array} \]
where the last equation follows by induction and the fact that we can assume that $`b\notele \fv(N)$, so that $\Sem{ N }{ e[`b \to k] } = \Sem{ N }{ e } $. 
Let $e' = e[`a \to \Pair<\Sem{ N }{ e } , e\psk`a >]$, then:
 \[ \begin{array}{rcl}
 \Sem{L}{ e[`b \to k, `a \to \Pair{\Sem{ N }{ e } }{ e\psk`a }] }{ \Pair<\Sem{ N }{ e } , e\psk`a > } &=&
 \Sem{ L }{ e'[`b \to k]}{\Pair<\Sem{ N }{ e } , e\psk`a >}
 \end{array} \]
and
 \[ \begin{array}{rcl@{\quad}l}
 \Pair{\Sem{L}{ e'[`b \to k] } }{ \Pair<\Sem{ N }{ e } , e\psk`a > } 
	& = & 
\Pair{\Sem{ L }{ e'[`b \to k] } }{ e'\psk `a} 
	\\ & = & 
\Pair{\Sem{ L }{ e'[`b \to k] } }{ e'[`b \to k]\psk `a} 
	\\ & = & 
\Sem{ [`a]L}{e'[`b \to k]}
 \end{array} \]
which implies
 \[ \begin{array}{@{}rcl@{\quad}lc@{}}
 \Sem{ `m `b. [`a]L }{ e' }{ k } 
	&=& 
 \Sem{ L }{ e'[`b \to k] }{ \Pair<\Sem{ N }{ e } , e\psk`a > } \\
	&=& 
 \Sem{ `m`b.[`a]L\StrSub{`a}{N} N }{ e }{ k } 
 \end{array} \]
\arrayqed[-20pt]
 \end{proof}

Since our interpretation in Def.\skp\ref{def:interpretation} does not coincide exactly with the one of Streicher and Reus, we have to check that it actually models $`l`m$ convertibility.
We begin by stating the key fact about the semantics, \ie~that it satisfies the following `swapping continuations' equation:\footnote{The equation is from \cite{Streicher-Reus'98}, where it is actually:
$ \Sem{ `m`a.[`b]M }{ e }{ k } = \Sem{ M }{ e[`a \becomes k] }{ (e\psk `b)}$, but this is certainly just a typo.}

 \begin{lem} \label{lem:sem-shortcut}
$ \Sem{ `m`a.[`b]M }{ e }{ k } = \Sem{ M }{ e[`a \to k] }{ (e[`a \to k] \psk `b) }$.
 \end{lem}

 \begin{Proof}
Since $\Sem{ `m`a.[`b]M }{ e }{ k } = d\psk k'$, where
$ \Pair{ d }{ k' } = \Sem{ [`b]M}{ e[`a\to k] } = \langle \Sem{ M }{ e[`a\to k] } ,$ \\ $ e[`a\to k]\psk `b \rangle $. 
 \end{Proof}

We are now in place to establish the soundness of the interpretation.

 \begin{thm}[Soundness of $ \Sem{ `. } $ with respect to $ \LamMuTheory $] \label{thm:sem_soundness}
If $ \derLM M = N $ then $ \Sem{ M } = \Sem{ N } $.
 \end{thm}

 \begin{proof} By induction on the definition of $\eqbmu$; it suffices to check the axioms $(`b)$, $(`m)$, and $(\Rename)$:

 \begin{description}

 \item [$ {(`lx.M)N = M[N/x] } $] 
$ \begin{array}[t]{lll@{\quad}l}
 \Sem{ (`lx.M)N }{ e }{ k } 
	&\ByDef& 
 \Sem{ M }{ e }{ \Pair< \Sem{ N }{ e } , k >} \\
	&=& 
 \Sem{ M }{ e[x \to \Sem{ N }{ e }] }{ k } \\
	&=& 
 \Sem{ M[N/x] }{ e }{ k } & (\textrm{Lem.\skp\ref{lem:sem_substitution_renaming} }) 
 \end{array} $

 \item [{$(`m`a.[`b]M)N = `m`a.([`b]M) \StrSub{`a} {N} $}] 
Notice that, by Barendregt's convention, we can assume that $`a \notele \fn(N)$; let $e' = e[`a \to \Pair{\Sem{ N }{ e } }{ k }]$, then:
 \[ \begin{array}[t]{lll@{\quad}l}
 \Sem{ (`m`a.[`b]M)N }{ e }{ k } 
	&\ByDef& \Sem{ `m`a.[`b]M }{ e }{ \Pair{\Sem{ N }{ e } }{ k } } \\
	&=& \Sem{ M }{ e' }{ (e' `b) } & (\textrm{Lem.\skp\ref{lem:sem-shortcut}}) \\
	&=& \Sem{ M\StrSub{`a}{N} }{ e[`a\to k] }{ (e' \psk `b) } 
& (\textrm{Lem.\skp\ref{lem:sem_substitution_renaming}} ) 
 \end{array} \]
Now if $`b = `a$ we have:
 \[ \begin{array}{rcl@{\quad}l}
 \Sem{ M\StrSub{`a}{N} }{ e[`a\to k] }{ (e' \psk `b) } 
	& = & \Sem{ M\StrSub{`a}{N} }{ e[`a\to k] }{ \Pair{\Sem{ N }{ e } }{ k } } & (`b = `a) \\
	& = & \Sem{ M\StrSub{`a}{N} }{ e[`a\to k] }{ \Pair{\Sem{ N }{ e [`a\to k] } }{ k } } & (`a\notele \fn(N)) \\
	& = & \Sem{ M\StrSub{`a}{N}N }{ e[`a\to k] }{ (e[`a\to k]\psk`a) } \\
	& = & \Sem{ `m`a.[`a]M\StrSub{`a}{N} N }{ e }{ k } & (\textrm{Lem.\skp\ref{lem:sem-shortcut}}) \\
	&\ByDef& \Sem{ `m`a.([`a]M) \StrSub{`a} {N} }{ e }{ k }
 \end{array} \]
Otherwise, if $`b \neq `a$ we have:
 \[ \begin{array}{rcl@{\quad}l}
 \Sem{ M\StrSub{`a}{N} }{ e[`a\to k] }{ (e' \psk `b) } 
	& = & \Sem{ M\StrSub{`a}{N} }{ e[`a\to k] }{ (e[`a \to k] \psk `b) } \\
	& = & \Sem{ `m`a.[`b]M \StrSub{`a}{N} }{ e }{ k } & (\textrm{Lem.\skp\ref{lem:sem-shortcut}}) \\
	&\ByDef& \Sem{ `m`a.([`b]M) \StrSub{`a} {N} }{ e }{ k } 
 \end{array} \]

\item [{$`m`j.[`a]`m`b.[`g]M = `m`j. ([`g]M)[`a/`b] $}]
For any $e'$ 
we have 
 \[ \begin{array}{rcl}
\Sem{[`a]`m`b.[`g]M}{e'} 
	&=& \Pair{\Sem{`m`b.[`g]M}{e'}}{(e'\, `a)} \\ 
	&=& \Pair{\Phi_{e'}}{(e'\, `a)} \end{array} \]
where
 \[ \begin{array}{rcl}
\Phi_{e'} 
	&=& `l k \ele C.~\Sem{`m`b.[`g]M}{e'}{k} \\ 
	&=& `l k \ele C.~\Sem{M}{e'[`b \to k]}{(e'[`b \to k]\, `g)}. 
 \end{array} \]
On the other hand, by definition of interpretation and by Lem.\skp\ref{lem:sem_substitution_renaming} we have:
 \[ \begin{array}{rcl}
\Sem{([`g]M)[`a/`b]}{e'} 
	&=& \Sem{[`g]M}{e'[`b \to e\,`a]} \\
	&=& \Pair{\Sem{M}{e'[`b \to e\,`a]}}{(e'[`b \to e\,`a]\,`g)}.
 \end{array} \]
Therefore, taking $e' = e[`j \to k]$, by Lem.\skp\ref{lem:sem-shortcut} we conclude:
 \[ \begin{array}{rcl}
\Sem{`m`j.[`a]`m`b.[`g]M}{ e }{ k } 
	&=& \Phi_{e'}(e'\,`a) \\
	&=& \Sem{`m`j. ([`g]M)[`a/`b]}{e}{k}.
 \end{array} \]
 \arrayqed[-22pt]

 \end{description}
 \end{proof}

 \section{The filter domain} \label{sec:filter}

 In this section we will build a $`l`m$-model in the category of $`w$-algebraic lattices. 
The model is obtained in Sect.\skp\ref{sub:domainTheoretic} by means of standard domain theoretic techniques, following the construction in \cite{Streicher-Reus'98}; we exploit the fact that compact points of any $`w$-algebraic lattice can be described by means of a suitable \emph{intersection type theory} (recalled in Sect.\skp\ref{w-algebraic-lattices}), to get a description of the model as a filter-model in Sect.\skp\ref{filter-model-solution}. 
This provides us with a semantically justified definition of intersection types (actually of three kinds, to describe the domains $R$, $D$ and $C$, respectively, that form the model) and of their pre-orders that we shall use in Sect.\skp\ref{sec:types} for the type assignment system.

The treatment of Sect.\skp\ref{sub:domainTheoretic} is introductory and can be skipped by readers who are familiar with domain theory, but for Prop.\skp\ref{prop:compactXinfty} and \ref{prop:compcats_of_C}, which are referred to in 
the paper. 
A fuller treatment of these topics can be found for example in \cite{Amadio-Curien'98}, Ch.~1-3 and 7. 
The developments in Sect.\skp\ref{w-algebraic-lattices} and Sect.\skp\ref{filter-model-solution} are inspired by \cite{BCD'83,Coppo-et.al'84} and \cite{Abramsky'91}; in particular, we have used \cite{Dezani-Honsell-Alessi'03} in Sect.\skp\ref{filter-model-solution}; we borrow the terminology of `intersection type theory' from \cite{Bar2013}, where intersection type systems and filter models are treated in full detail in Part III.

 \subsection{A domain theoretic solution of continuation domain equations} \label{sub:domainTheoretic}

Complete lattices are partial orders $(X,\po)$, closed under meet $\bigsqcap Z$ 
(greatest lower bound) and join $\bigsqcup Z$ (smallest upper bound) of arbitrary subsets $Z\subseteq X$. 
Observing that $\bigsqcap Z = \bigsqcup \Set{x \ele X \mid \forall z \ele Z \Pred[x \po z]}$ and $\bigsqcup Z = \bigsqcap \Set{x \ele X \mid \forall z \ele Z \Pred[ z \po x]}$, we have that if $X$ is closed under arbitrary meets (joins) it is likewise under arbitrary joins (meets). 
Furthermore, in $X$ there exist $\bot = \bigsqcup \emptyset$ and $\top = \bigsqcap \emptyset$, which are the bottom and top elements of $X$ with respect to $\po$, respectively.

A subset $Z\subseteq X$ is \emph{directed} if for any finite subset $V\subseteq Z$ there exists $z\ele Z$ which is an upper bound of $V$. 
In particular, directed subsets are always non-empty. An element $e\ele X$ is \emph{compact} if, whenever $e \po \bigsqcup Z$ for some directed $Z\subseteq X$, there exists $z\ele Z$ such that $e \po z$; we write $\Compact(X)$ for the set of compact elements of $X$. 
For $x\ele X$ we define $\Compact(x) = \Set{e \ele \Compact(X) \mid e \po x}$;
since directed sets are non-empty, $\bot \ele \Compact(X)$ and hence $\bot\ele \Compact(x)$, for all $x\ele X$. 
A complete lattice $X$ is \emph{algebraic} if $\Compact(x)$ is directed for any $x\ele X$ and $x = \bigsqcup \Compact(x)$; $X$ is \emph{$`w$-algebraic} if it is algebraic and the subset $\Compact(X)$ is countable.

A function $f:X\To Y$ of $`w$-algebraic lattices is \emph{Scott-continuous} if and only if it preserves directed sups, namely $f(\bigsqcup Z) = \bigsqcup_{z\ele Z}f(z)$ whenever $Z\subseteq X$ is directed. 
By algebraicity, any continuous function with domain $X$ is fully determined by its restriction to $\Compact(X)$, that is, given a monotonic function $g:\Compact(X)\To Y$ there exists a unique continuous function $\widehat{g}:X\To Y$ that coincides with $g$ over $\Compact(X)$, namely $\widehat{g}(x) = \bigsqcup g(\Compact(x))$; $\widehat{g}$ is called the \emph{continuous extension} of $g$.
The category $\ALG$ has $`w$-algebraic lattices as objects and Scott-continuous maps as morphisms. 
As such $\ALG$ is a full subcategory of the category of domains, but not of the category of (complete) lattices, since morphisms do not preserve arbitrary joins. 
In this paper we use the word \emph{domain} as synonym of $`w$-algebraic lattice.

If $X,Y$ are domains, then the component-wise ordered cartesian product $X\prod Y$ and the point-wise ordered set $[X \To Y]$ of Scott-continuous functions from $X$ to $Y$ are both domains.
In particular, if $f,g \ele [X\To Y]$ then the function $(f\join g)(x) = f(x) \join g(x)$ is the join of $f$ and $g$, that are always \emph{compatible} since they have an upper bound.
If $Z$ is an $`w$-algebraic lattice then $[X\prod Y \To Z] \simeq [X \To [Y \To Z]]$ is a natural isomorphism, and therefore the category $\ALG$ is cartesian closed. 

$\Compact(X\prod Y) = \Compact(X)\prod \Compact(Y)$; 
$\Compact[X\To Y]$ is the set of finite joins of \emph{step functions} $\StepFun{a }{ b}$ where $a \ele \Compact(X)$, $b \ele \Compact(Y)$, defined by 
 \[ \begin{array}{rcl}
 \StepFun{a }{ b}(x) &=& 
 \begin{cases}{l@{\quad}l} b & (\textrm{if }a \po x) \\
 \bot & (\textrm{otherwise}) 
 \end{cases}
 \end{array} \]

An infinite sequence $(X_n)_{n\ele \Nat}$ of domains is \emph{projective} if for all $n$ the continuous functions $e_n:X_n \To X_{n+1}$ and $p_n:X_{n+1}\To X_n$ exist, called \emph{embedding-projection pairs}, that satisfy $p_n \circ e_n = \textsf{id}_{X_n}$ and $e_n \circ p_n \seq \textsf{id}_{X_{n+1}}$, where $\seq$ is the pointwise ordering and $ \textsf{id}_V $ is the identity function on $V$. 
The \emph{inverse limit} of the projective chain $(X_n)_{n\ele \Nat}$ is the set $X_\infty = \lim_\leftarrow X_n$ which is defined as the set of all vectors $\vect{x} \ele \Pi_n X_n$ such that $x_i = p_i(x_{i+1})$ for all $i$, ordered component wise. 
Moreover, for all $n$ there exists an embedding-projection pair $e_{n,\infty}:X_n\To X_\infty$ and $p_{n,\infty}:X_\infty\To X_n$ such that $p_{n,\infty}(\vect{x}) = x_n$ for all $\vect{x}\ele X_\infty$: for details see for example \cite{Amadio-Curien'98}, Ch.~7. 

 \begin{prop} \label{prop:compactXinfty}
The inverse limit $X_\infty = \lim_\leftarrow X_n$ of a sequence $(X_n)_{n\ele \Nat}$ of domains is itself a domain such that
 \[ \begin{array}{rcl}
\Compact(X_\infty) 
	&=& 
\bigcup_n \Set{e_{n,\infty}(x) \mid x \ele \Compact(X_n)}.
 \end{array} \]
 \end{prop}

 \begin{Proof} That $X_\infty$ is a domain follows by the fact that each $X_n$ is, for all $n$, and that for all $x\ele X_n$ there exists $\vect{x} = e_{n,\infty}(x)$ such that $x_n = p_{n,\infty}(\vect{x}) = x$.

Writing $\bigcup_n e_{n,\infty} \Compact(X_n)$ for the right-hand side of the above equation, consider a directed subset $Z\subseteq X_\infty$.
For any $n$, if $x \ele \Compact(X_n)$ then $e_{n,\infty}(x) \po \bigsqcup Z$ implies
 \[ \begin{array}{rcl@{\quad}lccc}
x 
	&=& 
p_{n,\infty} \circ e_{n,\infty}(x) 
	&\po& 
p_{n,\infty}(\bigsqcup Z) 
	&=& 
\bigsqcup p_{n,\infty}(Z)
 \end{array} \]
by the fact that $(e_{n,\infty},p_{n,\infty})$ is an embedding-projection pair, and the continuity of $p_{n,\infty}$. By assumption there exists $z\ele Z$ such that $x \po p_{n,\infty}(z)$, and therefore 
 \[ \begin{array}{rcl@{\quad}lccc}
e_{n,\infty}(x) &\po& (e_{n,\infty}\circ p_{n,\infty})(z) &\po& z \ele Z ,
 \end{array} \] 
hence $e_{n,\infty}(x) \ele \Compact(X_\infty)$, by the arbitrary choice of $Z$. 
This proves $\Compact(X_\infty) \supseteq\bigcup_n e_{n,\infty} \Compact(X_n)$.

To see the converse inclusion, take $\vect{x} \ele \Compact(X_\infty)$. 
We claim that $x_n = p_{n,\infty}(\vect{x}) \ele \Compact(X_n)$, for any $n$. 
Indeed, if $U\subseteq X_n$ is directed and such that $x_n \po \join U$, consider the set 
 \[ \begin{array}{rcl}
V &=& \Set{ \vect{y} \ele X_\infty \mid \forall m\neq n \Pred[{ y_m = x_m \And\ \exists u \ele U \Pred[ \ y_n = u ] }] }
 \end{array} \]
Then $V$ is directed and $\vect{x} \po \join V$. 
From the hypothesis $\vect{x} \ele \Compact(X_\infty)$ we know that there exists $\vect{y} \ele V$ such that $\vect{x} \po\vect{y}$, and so there exists $u = y_n \ele U$ such that $x_n \po u$, establishing the claim. 
From this and the first part of this proof it follows that 
 \[ \begin{array}{rclcl}
\Set{e_{n,\infty}(p_{n,\infty}(\vect{x}))}_{n\ele \Nat} &=& \Set{e_{n,\infty}(x_n)}_{n\ele \Nat} &\subseteq& \bigcup_n e_{n,\infty} \Compact(X_n)
 \end{array} \] 
is a chain of elements in $\Compact(X_\infty)$, and by construction
$\vect{x} = \bigsqcup_n e_{n,\infty}(x_n)$; since $\vect{x} \ele \Compact(X_\infty)$ we can conclude that $\vect{x} = e_{n_0,\infty}(x_{n_0})$ for some $n_0$, so that $\Compact(X_\infty) \subseteq\bigcup_n e_{n,\infty} \Compact(X_n)$ as desired.
 \end{Proof}

We now consider the construction in \cite{Streicher-Reus'98} for the particular case of $\ALG$. Let $R$ be some fixed domain, dubbed the domain of \emph{results} (for the sake of solving the continuation domain equations in a non-trivial way it suffices to take $R = \Set{\bot,\top}$ with $\bot \spo \top $, the two-point lattice). Now define the following sequences of domains:
 \[ \begin{array}{lll}
C_0 &=& \Set{\bot} \\
D_n &=& [C_n \To R] \\
C_{n+1} &=& D_n\prod C_n 
 \end{array} \]
where $ \Set{\bot}$ is the trivial lattice such that $\bot = \top$. 
Observe that $D_0 = [C_0\To R] \Isom R $ and $ D_0 \Isom D_0 \times \Set{\bot} = C_1$ and so $D_1= [C_1 \To R] \Isom [R \To R]$. By unravelling the definition of $C_n$ and $D_n$ we obtain:
 \[ \begin{array}{rcl}
C_n	&=& 
[C_{n-1} \To R]\prod [C_{n-2} \To R]\prod \cdots\prod [C_0\To R]\prod C_0.
 \end{array} \]
In \cite{Streicher-Reus'98} Thm.~3.1, it is proved that these sequences are projective, so that $D = \lim_{\leftarrow} D_n$ and $C = \lim_{\leftarrow} C_n$ are the initial/final solution of the continuation domain equations such that $R \Isom D_0$.
By Prop.\skp\ref{prop:compactXinfty} we know that, up to the embeddings of each $D_n$ into $D$ and of each $C_n$ into $C$, the compact points of $D$ and $C$ are the union of the compacts of the $D_n$ and $C_n$, respectively:
 \begin{eqnarray} \label{eq:compactUnion}
 \Compact(D) = \bigcup_n \Compact(D_n), &\quad& \Compact(C) = \bigcup_n \Compact(C_n).
 \end{eqnarray}
In particular, $\Compact(R) = \Compact(D_0) \subseteq \Compact(D)$.
Since $C \simeq D\prod C$, $C$ can be seen as the infinite product $\Pi_n D = D\prod D \cdots \,$; also, $\Compact(C)$ is a proper subset of the product $\Pi_n \, \Compact(D) = \Compact(D)\prod \Compact(D)\prod \cdots$. 

 \begin{prop} \label{prop:compcats_of_C}
The compact points in $C = \lim_{\leftarrow} C_n$ are those infinite tuples in $\Pi_n \, \Compact(D)$ whose components are all equal to $\bot$ but for a finite number of cases:
 \[ \begin{array}{rcl}
\Compact(C) &=& \Set{\Tuple <d_1, d_2, \ldots > \ele \Pi_n \, \Compact(D) \mid \exists \,i \ \forall j\geq i \Pred[ d_j = \bot] }.
 \end{array} \]
 \end{prop}

 \begin{Proof}
Let 
 \[ \begin{array}{rcl}
K &=& \Set{\Tuple <d_1, d_2, \ldots > \ele \Pi_n \, \Compact(D) \mid \exists \,i \ \forall j\geq i \Pred[ d_j = \bot ] }
 \end{array} \]
Since $C = \Pi_n D$ is pointwise ordered, we have 
$K \subseteq \Compact(C)$, so 
it suffices to show the inverse inclusion.
Let $k = \Tuple <d_1, d_2, \ldots > \ele \Compact(C)$ and assume, towards a contradiction, that there exist infinitely many components $\Set{d_i \mid i \ele \Nat}$ of $k$ that are different from $\bot$. 
Set $k_j = \Tuple <d_1, d_2, \ldots, d_{j}, \bot, \ldots > \ele K$ (the tuple definitely equal to $\bot$ after $d_j$, while previous components are the same as in $k$), then the set $\Set{k_j \mid j\ele \Nat}$ is directed (actually a chain) and $k = \bigsqcup_j k_j$; but $k \notpo k_j$ for any $j$, hence $k\notele \Compact(C)$, a contradiction. 
 \end{Proof}

Another way to see this proposition is to observe that the $C_n$ are finite products of the shape $[C_{n-1} \To R]\prod [C_{n-2} \To R]\prod \cdots\prod [C_0\To R]\prod C_0$, where $C_0 = \Set{\bot}$.
Hence any tuple in $C_n$, and therefore in $\Compact(C_n)$, has the form $\Tuple <d_1,\ldots,d_{n-2},\bot>$. 
Now the embedding of such a tuple into $C = \lim_{\leftarrow} C_n$ is the infinite tuple
$\Tuple <d_1,\ldots,d_{n-2},\bot,\ldots,\bot,\ldots>$ that is definitely $\bot$ after $d_{n-2}$, and we know that the images of compact points in the $C_n$'s are exactly the elements of $\Compact(C)$.

 \subsection{Intersection type theories and the filter construction} \label{w-algebraic-lattices}

Intersection types form the `domain logic' of $\omega$-algebraic lattices in the sense of \cite{Abramsky'91}. 
This means that for each domain $X$ in $\ALG$ there exists a countable language $\Lang_X$ of intersection types together with an appropriate pre-order $\seq_X$ such that $(\Lang_X,\seq_X)$ is the Lindenbaum algebra of the compact points $\Compact(X)$ of $X$, 
\ie~an axiomatic, and hence finitary presentation of the structure $\CompOp(X) = (\Compact(X), \po^{\SuperOp})$ where $\po^{\SuperOp}$ is just the inverse of the partial order $\po$ of $X$.

To understand this, first observe that $\Compact(X)$ is closed under binary joins. 
Indeed, for any $e_1,e_2 \ele \Compact(X)$ if $e_1 \join e_2 = \bigsqcup \Set{e_1, e_2} \po \bigsqcup Z$ for some directed $Z\subseteq X$ then $e_1,e_2 \po \bigsqcup Z$, which implies that there exist $z_1,z_2\ele Z$ such that $e_1 \po z_1$ and $e_2 \po z_2$. 
By directness of $Z$ there exists some $z_3 \ele Z$ such that $z_1,z_2 \po z_3$, hence
$e_1 \join e_2 \po z_1 \join z_2 \po z_3$. 
Thereby the structure $(\Compact(X), \po)$ is a sup-semilattice (a poset closed under finite joins), hence its dual $\CompOp(X)$ is an inf-semilattice (a poset closed under finite meets), whose meet operator $\sqcap^{\SuperOp}$ coincides with the join $\join$ over $\Compact(X)$.

By algebraicity $X$, is generated by $\Compact(X)$ in the sense that $(X,\po)$ is isomorphic to the poset $(\Idl(\Compact(X)),\subseteq)$, where
$\Idl(\Compact(X))$, the set of ideals over $\Compact(X)$, consists of directed and downward closed subsets of $\Compact(X)$. It turns out that the compact elements of $\Idl(\Compact(X))$ 
are just the images ${\downarrow} e = \Compact(e)$ of the elements $e \ele \Compact(X)$. 
Dually, $(X,\po)$ is isomorphic to the poset $(\Filter(\CompOp(X)),\subseteq)$, where $\Filter(\CompOp(X))$ is the set of \emph{filters} over $\CompOp(X)$, that are non-empty subsets of $\Compact(X)$ which are upward closed with respect to $\po^{\SuperOp}$ and closed under $\sqcap^{\SuperOp}$. 

Therefore filters over $\CompOp(X)$ give rise to the algebraic lattice $(\Filter(\CompOp(X)),\subseteq)$, whose
compact elements are $\filt^{\SuperOp} e = \Set{e' \ele \Compact(X) \mid e\po^{\SuperOp} e'}$, called the \emph{principal filters}. 
In summary, we have the isomorphisms in the category $\ALG$:
 \[ \begin{array}{rcccl}
X &\Isom& \Idl(\Compact(X)) &\Isom& \Filter(\CompOp(X)) .
 \end{array} \]

The fact that $\Compact(X)$ is a countable set allows for a finitary (syntactic) presentation
of $X\Isom \Filter(\CompOp(X))$ itself by introducing a language of types denoting the elements of $\Compact(X)$ and axioms and rules defining a pre-order over types whose intended meaning is $\po^{\SuperOp}$.

An \emph{intersection type language} $\Lang$ is a set of expressions closed under the binary operation $\inter$ and including the constant $`w$. 
A pre-order $\seq$ is defined over this set, making $\inter$ into the meet and $`w$ into the top element, as formally stated in the next definition.

 \begin{defi}[Intersection Type Language and Theory] \label{def:interLangTheory}

 \begin{enumerate}

 \firstitem 
A denumerable set of type expressions $\Lang$ is called an \emph{intersection type language} if there exists a constant $`w \ele \Lang$ and $\Lang$ is closed under the binary operator $`s\inter`t $, called \emph{type intersection}.

 \item 
An \emph{intersection type theory $\TypeTheor$ over $\Lang$} (where $\Lang$ is an intersection type language) is an axiomatic presentation of a pre-order $\seqT$ over 
types in $\Lang$ validating the following axioms and rules:
 \[ \begin{array}{c@{\dquad}c@{\dquad}c@{\dquad}c}
	\Inf {`s \inter `t \seqT `s } &
	\Inf {`s \inter `t \seqT `t } &
	\Inf {`s \seqT `w } &
	\Inf {`r \seqT `s \quad `r \seqT `t }{ `r \seqT `s \inter `t }
	\end{array} \]

 \item
We abbreviate $`s\seqT`t\seqT`s$ by $`s \simT`t$ and write $[`s]_{\TypeTheor}$ for the equivalence class of $`s$ with respect to $\simT$. The subscript ${\TypeTheor}$ will be omitted when no ambiguity is possible.

 \end{enumerate}
 \end{defi}

The type $`s\inter`t$ is called an intersection type in the literature. 
The reason for this is that as a type of $`l$-terms it is interpreted as the intersection of the interpretations of $`s$ and $`t$ in set theoretic models of the $`l$-calculus. 
This is rather unfortunate in the present setting, where we shall speak of filters and of their intersections. 
To avoid confusion, we speak of `type intersections' when we refer to expressions of the shape $`s\inter`t$, reserving the word `intersection' to the set theoretic operation.

Given an intersection type theory $\TypeTheor$ over a language $\Lang$ that axiomatises the pre-order $\seq_{\TypeTheor}$, the quotient $\LangQseqT$ is an inf-semilattice. 
We will now establish a sufficient condition for $\seq_{\TypeTheor}$ to be isomorphic to $\CompOp(X)$ for some $X \ele \ALG$.

 \begin{lem} \label{lem:compactRep}
Let $ \TypeTheor$ be an intersection type theory over $\Lang$ and $\seq_{\TypeTheor}$ the relative pre-order. 
Let $(X,\po)$ be a domain and $ \TypeToComp: \Lang \To \Compact(X) $ an order-reversing surjective mapping, \ie~such that for all $`s,`t \ele \Lang$:
 $ \begin{array}{rcl}
`s \seq_{\TypeTheor} `t &\Then& \TypeToComp(`t) \po \TypeToComp(`s). 
 \end{array} $
Then $\LangQseqT \Isom \CompOp(X)$ as inf-semi-lattices.
 \end{lem}

 \begin{Proof}
Let $ \TypeToComp': \LangQseqT \To \CompOp(X)$ be defined by $ \TypeToComp'([`s]) = \TypeToComp(`s)$.
If $[`s] = [`t]$, then $`s \simT `t$, so by assumption $ \TypeToComp(`s) \sqsupseteq \TypeToComp(`t) \sqsupseteq \TypeToComp(`s)$, which implies $ \TypeToComp(`s) = \TypeToComp(`t)$. 
This implies that $ \TypeToComp'$ is well defined and that $\TypeToComp$ preserves and reflects $\seqT$ with respect to $ \po^{\SuperOp}$, and that $\TypeToComp'$ is a bijection, since $ \TypeToComp$ is surjective. 
Finally, $ \TypeToComp'([`s \inter `t]) = \TypeToComp(`s \inter `t) = \TypeToComp(`s) \join \TypeToComp(`t) = \TypeToComp(`s) \meet^{\SuperOp} \TypeToComp(`t)$. 
In particular, $ \TypeToComp'([`w]) = \TypeToComp(`w) = \bot$.
 \end{Proof}

Under the hypotheses of the last lemma we have $X \Isom \Filter(\CompOp(X)) \Isom \Filter(\LangQseqT)$; nonetheless, we consider the more convenient isomorphism of $X$ with the set of filters over the pre-order $(\Lang,\seq_{\TypeTheor})$, which we call \emph{formal filters}.

 \begin{defi}[Formal Filters] \label{def:formal-filt}
 \begin{enumerate}

 \firstitem
A \emph{formal filter} with respect to an intersection type theory $\TypeTheor$ over $\Lang$ is a subset $f\subseteq\Lang$ such that:
 \[ \begin{array}{c@{\dquad}c@{\dquad}c} 
\Inf{`w \ele f} &
\Inf{`s \ele f \quad `s \seqT`t }{`t \ele f} &
\Inf{`s \ele f \quad `t \ele f }{`s \inter `t \ele f}
 \end{array} \]
%
$\Filt(\TypeTheor)$ is the set of formal filters induced by the theory and we let $f,g$ range over $\Filt(\TypeTheor)$.

 \item 
The filter $ \filtThr `s = \Set{`t \ele \Lang \mid `s \le_{\TypeTheor} `t}$ is called \emph{principal} and we write $\Filt_p(\TypeTheor)$ for the set of principal filters.

 \end{enumerate}
 \end{defi}

We recall some properties of formal filters and of the poset $(\Filt(\TypeTheor),\subseteq)$. 
Since these are easily established or well known from the literature, we just state them or provide short arguments. 
The following is a list of some useful facts that follow immediately by definition. 

 \begin{fact} \label{filter lemma}
 \begin{enumerate}

 \firstitem \label{filter lemma leq} 
$`s\seqT `t$ if and only if $\filtThr`t \subseteq \filtThr`s$.

 \item \label{filter lemma join-meet} 
 $\filtThr `s \join \filtThr`t = \filtThr`s\inter`t$.

 \item \label{filter lemma ele} 
If $`s \ele f$, then $\filtThr `s \subseteq f$.

 \item \label{filter lemma filtersup}
$f = \bigcup_{`s \ele f} \filtThr `s $.

 \item \label{filter lemma arbitraryInter}
For any ${\FiltMod[G]}\subseteq \Filt(\TypeTheor)$, $\bigcap{\FiltMod[G]} \ele \Filt(\TypeTheor)$.

 \end{enumerate}
 \end{fact}

From Fact \ref{filter lemma}\,(\ref{filter lemma arbitraryInter}) follows that the poset $(\Filt(\TypeTheor),\subseteq)$ is a complete lattice, with set-theoretic intersection as (arbitrary) meet. 
Notice that, for ${\FiltMod[G]}\subseteq \Filt(\TypeTheor)$, the join 
 \[ \begin{array}{rcl}
\bigsqcup {\FiltMod[G]} &=& \bigcap \Set{f \ele \Filt(\TypeTheor) \mid \forall g \ele {\FiltMod[G]} \Pred[ g \subseteq f] }
 \end{array} \]
includes $\bigcup {\FiltMod[G]}$ but does not coincide with it in general, since $\bigcup {\FiltMod[G]}$ is not necessarily closed under $\inter$. 
However, since $\bigcup {\FiltMod[G]}$ is upper-closed with respect to $\seqT$, to get an explicit characterisation of $\bigsqcup {\FiltMod[G]}$ it is enough to close $\bigcup {\FiltMod[G]}$ under finite type intersections:
 \[ \begin{array}{rcl}
\bigsqcup {\FiltMod[G]} &=&
 \Set{`s \mid \Exists n, `s_i ~ (i \ele \n\footnotemark) \Pred[{ \forall i \seq n \Pred[`s_i \ele \bigcup{\FiltMod[G]} ] \And `s \simT `s_1\inter\cdots\inter`s_n }] } 
 \end{array} \] 
\footnotetext{We write $\n$ for the set $\Set{1,\ldots,n}$ and $a_i \ele V ~ (i \ele \n)$ for $a_1 \ele V, \ldots, a_n \ele V $.}%
On the other hand, if ${\FiltMod[G]}$ is directed with respect to $\subseteq$, then $\join{\FiltMod[G]} =\bigcup {\FiltMod[G]}$. 
In fact, if $`s_i \ele \bigcup{\FiltMod[G]} ~ (i \ele \n)$ we have that
$\filtThr `s_i \subseteq g_i$, for certain $g_i\ele {\FiltMod[G]} ~ (i \ele \n) $, by Fact \ref{filter lemma}\,(\ref{filter lemma ele}). 
By directness there exists $g' \ele {\FiltMod[G]}$ such that $g_1 \cup \cdots \cup g_n \subseteq g'$, and hence the same holds for $\filtThr `s_1 \cup \cdots \cup \filtThr `s_n$. 
Then $`s_i \ele g' ~ (i \ele \n)$ and so $`s_1\inter\cdots\inter`s_n \ele g' \subseteq \bigcup {\FiltMod[G]}$, since $g'$ is a formal filter.

 \Steffen{The next lemma is not referenced in the paper.
Moreover, it seems that you are using some knowledge here that is unspecified.
For example, it seems you show that $\Compact(\Filt(\TypeTheor)) = \Filt_p(\TypeTheor)$; we have defined $\Filt_p(\TypeTheor)$, but not $\Compact(\Filt(\TypeTheor))$.
So how can this proof check the equality?.}

The next lemma is not referenced explicitly in the paper, but is used throughout the rest of this section; it can be considered folklore in the theory of filter $`l$-models.

 \begin{lem} \label{prop:formal-filt}
$(\Filt(\TypeTheor),\subseteq)$ is an $`w $-algebraic lattice with top $\filt_{\TypeTheor} `w$ and compacts $\Compact(\Filt(\TypeTheor)) = \Filt_p(\TypeTheor)$.
 \end{lem}

 \begin{Proof}
From the discussion above we know that $(\Filt(\TypeTheor),\subseteq)$ is a complete lattice, so it remains to show that it is $`w$-algebraic.

Let $\filtThr `s \subseteq \bigsqcup{\FiltMod[G]}$ for some directed ${\FiltMod[G]} \subseteq \Filt(\TypeTheor)$. 
Then $\bigsqcup{\FiltMod[G]} = \bigcup{\FiltMod[G]}$ so that $`s\ele \filtThr `s \subseteq \bigcup{\FiltMod[G]}$. Therefore there exists $g\in{\FiltMod[G]}$ such that $`s\ele g$ which implies $\filtThr `s \subseteq g$ by \ref{filter lemma}\,(\ref{filter lemma ele}). 
Hence $\Filt_p(\TypeTheor) \subseteq \Compact(\Filt(\TypeTheor))$.
By \ref{filter lemma}\,(\ref{filter lemma filtersup}) we have $f = \bigcup_{`s \ele f} \filtThr `s $ for any $f \ele \Filt(\TypeTheor)$. 
Let $\filtThr `s_i \subseteq f$ for certain $`s_i \ele f ~ (i \ele \n)$; using \ref{filter lemma}\,(\ref{filter lemma join-meet}) repeatedly we have $\filtThr `s_1 \join\cdots\join \filtThr`s_n = \filtThr(`s_1\inter\cdots\inter `s_n)$. 
On the other hand, $`s_1\inter\cdots\inter `s_n \ele f$ since $f$ is a formal filter, and $`s_1\inter\cdots\inter `s_n \seqT `s_i$ for all $i \ele \n$, which by \ref{filter lemma}\,(\ref{filter lemma leq}) implies that $ \filtThr `s_i \subseteq \filtThr(`s_1\inter\cdots\inter `s_n)$, \ie~$\Set{\filtThr `s \mid `s\ele f}$ is directed. 

Now if $f\ele \Compact(\Filt(\TypeTheor))$, then $f\subseteq \filtThr `s$ for some $`s \ele f$; by \ref{filter lemma}\,(\ref{filter lemma ele}) we conclude $f = \filtThr `s$. 
Therefore $\Compact(\Filt(\TypeTheor))\subseteq \Filt_p(\TypeTheor)$ and hence $\Compact(\Filt(\TypeTheor)) = \Filt_p(\TypeTheor)$ by the above. 
By this and \ref{filter lemma}\,(\ref{filter lemma filtersup}) we conclude that $(\Filt(\TypeTheor),\subseteq)$ is algebraic, and in fact is $`w$-algebraic because $\Lang$ is countable and the map $`s \mapsto \filtThr `s$ from $\Lang$ to $ \Filt_p(\TypeTheor)$ is obviously onto.
 \end{Proof}

 \begin{prop} \label{prop:filter-iso} 
Let $\TypeTheor$ be an intersection type theory over $\Lang$, $X$ a domain and $ \TypeToComp: \Lang \To \Compact(X)$ a mapping that satisfies the hypotheses of Lem.\skp\ref{lem:compactRep}. 
Then $\Filt(\TypeTheor) \Isom \Filter(\CompOp(X)) \Isom X$.
 \end{prop}

 \begin{proof} 
If $f \ele \Filt(\TypeTheor)$ then by Lem.\skp\ref{filter lemma}, $\Set{[`s] \mid `s \ele f}$ is a filter over $\LangQseqT$; vice versa, if $f$ is a filter over $\LangQseqT$ then $\bigcup f = \Set{`s \mid [`s] \ele f}$ is a formal filter; therefore $\Filt(\TypeTheor) \Isom \Filter(\LangQseqT)$. 
On the other hand, by Lem.\skp\ref{lem:compactRep} and the hypothesis, we have $\LangQseqT \Isom \CompOp(X)$ via the mapping $ \TypeToComp'([`s]) = \TypeToComp(`s)$, so that the desired isomorphism $\Filt(\TypeTheor) \Isom\Filter(\CompOp(X))$ is given by 
 \[ \begin{array}{rcl} 
 \Semlambda f . \bigsqcup \Set{\TypeToComp'([`s]) \mid `s \ele f} 
	&=& 
 \Semlambda f . \bigsqcup \Set{\TypeToComp(`s) \mid `s \ele f}. 
 \end{array} \]
\arrayqed[-28pt]
 \end{proof}

Because of Prop.\skp\ref{prop:filter-iso} and essentially following \cite{BCD'83}, we ignore the distinction between formal filters over a pre-order and filters over the ordered quotient and we shall work with the simpler formal filters, henceforth just called filters.

 \subsection{A filter domain solution to the continuation domain equations} \label{filter-model-solution}

Given an arbitrary domain $R$ as in Def.\skp\ref{eq:domain}, we fix the initial and final solution $D$, $C$ of the continuation equations in the category $\ALG$. 
For $A = R,D,C$, we will now define the languages $\Lang_A$ and the theories $\TypeTheor_A$, inducing the pre-orders $\seq_{\TypeTheor_A}$, which we write as $\seqA$.

 \begin{defi} \label{def:typeTHeorFromDomain}
Take $R \ele \ALG$, ordered by $\po_R$, with bottom $\bot$ and join $\join$.

 \begin{enumerate}

 \item 
The intersection type language $\Lang_R$ is defined by the grammar:
 \[ \begin{array}{rcl@{\quad}l}
`r & :: = & `y_a \mid `w \mid `r \inter `r & (a \ele \Compact(R))
 \end{array} \]
We let $`y$ range over the set $\Set {`y_a \mid a \ele \Compact(R)}$.

 \item $\TypeTheor_R$ is the smallest intersection type theory axiomatising the pre-order $\seq_R$ such that (where $\simR = \seq_R \cap \seq^{\SuperOp}_R$):
 \[ \begin{array}{c@{\dquad}c}
\Inf	{`y_ \bot \simR `w}
	&
\Inf	{`y_{a \join b} \simR `y_a \inter `y_b}
 \end{array} \]

 \item 
The mapping $ \TypeToComp_R: \Lang_R\To\Compact(R)$ is defined by:
 \[ \begin{array}{rcl}
 \TypeToComp_R(`y_a) &=& a \\ 
 \TypeToComp_R(`w) &=& \bot \\ 
 \TypeToComp_R(`r_1 \inter `r_2) &=& \TypeToComp_R(`r_1) \join \TypeToComp_R(`r_2) 
 \end{array} \]

 \end{enumerate}
Observe that $\inter$ is the meet with respect to $\seqR$.
 \end{defi}

The following property, that states the relation between $\po_R$ and $\seqR$, holds naturally:

 \begin{lem} \label{lem:KRtheory}
 \begin{enumerate}

 \firstitem \label{lem:KRtheory-1}
	$`y_a \seqR `y_b \Iff b \po_R a$.

 \item \label{lem:KRtheory-2}
	$`r \seqR `r' \Iff \TypeToComp_R(`r') \po_R \TypeToComp_R(`r)$.

 \end{enumerate}
 \end{lem}

 \begin{proof} In the following we remove all the subscripts $R$ for notational simplicity.

 \begin{enumerate}

 \item 
If $`y_a \seq `y_b$ then either $b = \bot $, so that $\bot \po a$, or
$`y_a \sim `y_a \inter `y_b$, since $\inter$ is the meet with respect to $\seq$, which is an intersection type theory. 
By definition we have $`y_{a \sqcup b} \sim `y_a \inter `y_b$, so it follows that $`y_a \sim `y_{a \sqcup b}$, which implies $a = a \join b $, so $b \po a$.

Vice versa, if $b \po a$ then $a = a \join b$ and we have $`y_a = `y_{a \sqcup b} \sim `y_a \inter `y_b$ from which we conclude $`y_a \seq `y_b$.
	
 \item 
When we, consistently with the theory $\seq$, identify $`y_\bot$ and $`w$, then $`r = `y_ {a_1}\inter \cdots \inter `y_{a_h}$, for some $a_i \ele \Compact(R) ~ (i \ele \h)$; notice that $`y_ {a_1}\inter \cdots \inter `y_{a_h} \sim `y_{a_1\sqcup\cdots\sqcup a_h}$. 
Likewise, $`r' = `y_ {b_1}\inter \cdots \inter `y_{b_k } $, then by part (\ref{lem:KRtheory-1}) we have:
 \[ \begin{array}{rcl@{\quad}lccc@{\quad}c@{\quad}ccc}
`r 
	& \sim & 
`y_{a_1\sqcup\cdots\sqcup a_h} 
	& \seq & 
`y_{b_1\sqcup\cdots\sqcup b_k } 
	& \sim & 
`r' 
	&\Iff&
b_1\join\cdots\join b_k 
	&\po&
a_1\join\cdots\join a_h
 \end{array} \]
The result follows from 
$\TypeToComp_R(`r') = b_1\join\cdots\join b_k$ and $\TypeToComp_R(`r) = a_1\join\cdots\join a_h$.%
\qed
 \end{enumerate}
 \end{proof}

The following corollary is the converse of Prop.\skp\ref{prop:filter-iso}.

 \begin{cor} \label{cor:fromDomToFilter}
There exists an intersection type theory $\TypeTheor_R$ such that $\Filt_R \Isom R$.
 \end{cor}

 \begin{Proof}
Let $\TypeTheor_R$ and $\TypeToComp_R$ be defined as in Def.\skp\ref{def:typeTHeorFromDomain}.
Now $\TypeToComp_R$ is surjective since $\TypeToComp_R(`y_a) = a$ for all $a\ele \Compact(R)$ and, by Lem.\skp\ref{lem:KRtheory}\,(\ref{lem:KRtheory-2}), it satisfies the hypotheses of Lem.\skp\ref{lem:compactRep}.
We conclude that $\Filt_R \Isom R$ by Prop.\skp\ref{prop:filter-iso}. 
 \end{Proof}

 \begin{rem} \label{rem:Filt_R}
By Prop.\skp\ref{prop:filter-iso}, the isomorphism $\Filt_R \Isom R$ is given by the map $r \mapsto \bigsqcup \Set{\TypeToComp_R(`r) \mid `r \ele r}$; as observed in the proof of Lem.\skp\ref{lem:KRtheory}\,(\ref{lem:KRtheory-2}), for any $`r\ele \Lang_R$ there exists $a\ele \Compact(R)$ such that $`r \simR `y_a$, therefore the filter $r$ is mapped isomorphically to $\bigsqcup \Set{\TypeToComp_R(`y_a) \mid `y_a \ele r} = \bigsqcup \Set{a \mid `y_a \ele r}$ by Def.\skp\ref{def:typeTHeorFromDomain}. 
In case $r \ele \Compact(\Filt_R)$ then by Lem.\skp\ref{prop:formal-filt} and the previous remarks $r = \filt_R`r = \filt_R`y_a$ for some $`r$ and $a$, and its image in $R$ is just $a$. 
 \end{rem}

 \begin{defi} [Type theories $\TypeTheor_D$ and $\TypeTheor_C $] \label{def:typeTheories}

 \begin{enumerate}

 \firstitem
$\Lang_D$ and $\Lang_C$ are the intersection type languages defined by the grammar:
 \[ \begin{array}{r@{\quad}rcll} 
 \Lang_D: & `d & :: =& `r \mid `k\arrow `r \mid `w \mid `d\inter `d & (`r \ele \Lang_R) \\
 \Lang_C: & `k & :: =& `d\prod `k \mid `w \mid `k \inter `k 
 \end{array} \]
We let $`d$ range over $\Lang_D$, and $`k$ over $\Lang_C$, and $`s,`t$ over $\Lang_D \union \Lang_C$.

 \item
We define $\Inter_{i \ele I}`s_i$ through:
 \[ \begin{array}{rcl@{\quad}l}
\Inter_{i \ele \emptyset}`s_i &\ByDef& `w \\
\Inter_{i \ele I}`s_i &\ByDef& `s_p \inter (\Inter_{i \ele I\except p}`s_i) & (p \ele I) 
 \end{array} \]

 \item
The theories $\TypeTheor_D$ and $ \TypeTheor_C$ are the least intersection type theories closed under the following axioms and rules, inducing the pre-orders $\seqD$ and $ \seqC $ over $\Lang_D$ and $\Lang_C$ respectively: 
 \[ \begin{array}{c@{\dquad}c@{\dquad}c@{\dquad}c@{\dquad}c}
\Inf	{ `r_1 \seqR `r_2 }
	{ `r_1 \seqD `r_2 }
	&
\Inf	{ `w \seqD `w\arrow `w }
	&
\Inf	{ `y \seqD `w\arrow `y }
	&
\Inf	{ `w\arrow `y \seqD `y }
	&
\Inf	{ `w \seqC `w\prod `w }
 \end{array} \]
 \[ \begin{array}{c@{\dquad}c}
\Inf	{(`k\arrow `d_1) \inter (`k\arrow `d_2) \seqD `k\arrow(`d_1 \inter`d_2) }
	&
\Inf {(`d_1\prod `k_1) \inter (`d_2\prod `k_2) \seqC (`d_1 \inter `d_2)\prod (`k_1 \inter `k_2)}
 \end{array} \]

 \[ \begin{array}{c@{\dquad\dquad}c}
\Inf	{`k_2 \seqC `k_1 \quad `r_1 \seq_R `r_2
	}{ `k_1\arrow `r_1 \seqD `k_2\arrow `r_2 }
	&
\Inf	{`d_1 \seqD`d_2 \quad `k_1 \seqC `k_2
	}{ `d_1\prod `k_1 \seqC `d_2\prod `k_2 }
 \end{array} \]
As usual, we define $`s \simA `t$ if and only if $`s \seqA `t \seqA `s$, for $A = C,D$.

 \end{enumerate}
 \end{defi}
It is straightforward to show that both $(`s \inter `t) \inter `r \simA `s \inter (`t \inter `r)$ and $`s \inter `t \simA `t \inter `s$, so the type constructor $\inter$ is associative and commutative, and we will write $`s \inter `t \inter `r$ rather than $(`s \inter `t) \inter `r$.
Thereby the definition of $ \Inter_{i \ele I}`s_i $ does not depend on the order in which $p$ is chosen from $I$.

The pre-order $\seqD$ is the usual one on arrow types in that the arrow is contra-variant in the first argument and co-variant in the second one. 
The pre-order $\seqC$ on product types is co-variant in both arguments, and is the component-wise pre-order. 
As immediate consequence of Def.\skp\ref{def:typeTheories} we have that $`w \simD `w\arrow `w $, $`w \simC `w\prod `w $, 
$ (`k\arrow `d_1) \inter (`k\arrow `d_2) $ 
	$\simD$
$`k\arrow(`d_1 \inter`d_2) $,
and
$ (`d_1\prod `k_1) \inter (`d_2\prod `k_2) $
	$\simC$ 
$(`d_1 \inter `d_2)\prod (`k_1 \inter `k_2) $.

The equation $`w \simD `w\arrow `w $ together with $`y \simD `w\arrow `y $ are typical of filter models that are extensional $`l$-models.
The equation $`w \simC `w\prod `w $ allows for a finite representation of compact elements in $C$, that otherwise should be described by infinite expressions of the form $`d_1\prod \cdots\prod `d_k\prod `w\prod `w\prod \cdots$ (see points (\ref{lem:SubtypeProp-1}) and (\ref{lem:SubtypeProp-2}) of Lem.\skp\ref{lem:SubtypeProp} below). 

We have also:
 \[ \begin{array}{rcl@{\quad}lccc}
`w &\simD& `w\arrow `w &\seqD& `k\arrow `w &\seqD& `w ,
 \end{array} \]
which implies that $`k\arrow `r \simD `w $ if and only if $`r\simD`w $.

 \begin{lem} \label{lem:SubtypeProp}

 \begin{enumerate}

 \firstitem \label{lem:SubtypeProp-1}
$\Forall `k \ele \Lang_C \ \Exists `d_i \ele \Lang_D ~ (i \ele \n) \Pred[ `k \simC `d_1\prod \cdots\prod `d_n\prod `w ] $.

 \item \label{lem:SubtypeProp-2}
$`d_1\prod \cdots\prod `d_h\prod `w \seqC `d'_1\prod \cdots\prod `d'_k\prod `w \Iff k \seq h \And \Forall i \seq h \Pred [ `d_i \seqD `d'_i ] $.

 \item \label{lem:SubtypeProp-4}
$\Forall `d \ele \Lang_C \ \Exists n>0,`k_i \ele \Lang_C ,\ `r_i \ele \Lang_R ~ (i \ele \n) \Pred[ `d \simD \Inter_{\n} (`k _i\arrow `r_i) ] $.

 \item \label{lem:SubtypeProp-3}
If $I,J$ are finite and non-empty sets of indices and $`r_i\not\!\!\simD \,`w $ for all $i \ele I$ then:
 \[ \begin{array}{rcl}
 \Inter_{j \ele J} (`k '_j\arrow `r'_j) \seqD \Inter_{i \ele I} (`k _i\arrow `r_i) 
	&\Iff& \\
\multicolumn{3}{l}{\dquad\dquad
 \Forall i \ele I \, \Exists J_i\subseteq J \Pred[ J_i \neq\emptyset \And `k _i\seqC \Inter_{j \ele J_i}`k _j \And \Inter_{j \ele J_j}`r_j \seq_R `r_i ]}
 \end{array} \]

 \end{enumerate}
 \end{lem}

 \begin{Proof} By induction on the structure of types and derivations in the theories $\TypeTheor_C$ and $\TypeTheor_D$. 
 \end{Proof}

\noindent
Note that the equivalence $`y \simD `w\arrow `y$ is necessary to show Lem.\skp\ref{lem:SubtypeProp}\,(\ref{lem:SubtypeProp-4}). 

The proofs of parts (\ref{lem:SubtypeProp-1}) and (\ref{lem:SubtypeProp-2}) of the last lemma are straightforward; they should be compared however with Prop.\skp\ref{prop:compcats_of_C}, to understand how compact points in $C$ are represented by types in $\Lang_C$.
Parts (\ref{lem:SubtypeProp-4}) and (\ref{lem:SubtypeProp-3}) are characteristic of \emph{extended abstract type structures} (\EATS; see for example \cite{Amadio-Curien'98} Sect.~3.3); in particular the latter implies:
 \[ \begin{array}{rcl}
\Inter_{j \ele J} (`k '_j\arrow `r'_j) \seqD `k\arrow `r 
	&\Then& 
\Inter\Set{`r'_j\mid `k\seqC `k'_j} \seqR `r.
 \end{array} \]
See \cite{BCD'83} Lem.~2.4 (ii) or \cite{Amadio-Curien'98} Prop.~3.3.18.
 \Ugo{The previous remark, that was removed, has been restored because the above is the form of the statement in the BCD paper, and is more widely known than the formulation in the lemma. 
The latter is however more general and useful to us.}

The next step is to define the mappings $\TypeToComp_D:\Lang_D \To \Compact(D)$ and $\TypeToComp_C:\Lang_C \To \Compact(C)$ such that both satisfy the hypotheses of Lem.\skp\ref{lem:compactRep}. 
In doing that we shall exploit Equation \ref{eq:compactUnion} in Sect.\skp\ref{sub:domainTheoretic}, by introducing a stratification of $\Lang_D$ and $\Lang_C$, thereby extending \cite{Dezani-Honsell-Alessi'03}. 
First we define the \emph{rank} of a type in $\Lang_D$ and $\Lang_C$ inductively as follows:
 \[ \begin{array}{rcl}
 \Rank(`r) = \Rank(`w) &=& 0 \\
 \Rank(`s \inter `t) &=& \max \Set{\Rank(`s), \Rank(`t)} \\
 \Rank(`d\prod `k) &=& \max \Set{\Rank(`d), \Rank(`k)} + 1 \\
 \Rank(`k\arrow `r) &=&\Rank(`k) + 1.
 \end{array} \]
Then we define $\Lang_{A_n} = \Set{`s \ele \Lang_A \mid \Rank(`s) \seq n}$ for $A=D,C$. By this we have that if $n \seq m$ then $\Lang_{A_n} \subseteq \Lang_{A_m}$ and that $\Lang_A = \bigcup_n \Lang_{A_n}$.

 \begin{defi} \label{def:typesToComp}
The mappings $ \TypeToComp_{C_n}: \Lang_{C_n} \To \Compact(C_n)$ and
$ \TypeToComp_{D_n}: \Lang_{D_n} \To \Compact(D_n)$ are defined by mutual induction through:
 \[ \begin{array}{lcl}
 \TypeToComp_{C_0}(`k) &=& \bot \\
 \TypeToComp_{D_n}(`y) &=& \StepFun{\bot }{ \TypeToComp_R(`y)} = \Semlambda \_ \, . \, \TypeToComp_R(`y) \\
 \TypeToComp_{D_n}(`k\arrow `r) &=& \StepFun{\TypeToComp_{C_n}(`k) }{ \TypeToComp_R(`r)} \\
 \TypeToComp_{C_{n+1}}(`d\prod `k) &=& \Pair{\TypeToComp_{D_n}(`d) }{ \TypeToComp_{C_n}(`k)}
 \end{array} \]
And, for $A_n = C_n, D_n$:
 \[ \begin{array}{lll}
 \TypeToComp_{A_n}(`w) &=& \bot \\
 \TypeToComp_{A_n}(`s \inter`t) &=& \TypeToComp_{A_n}(`s) \join \TypeToComp_{A_n}(`t)
 \end{array} \]
 \end{defi}

The following lemma states that the mappings $\TypeToComp_{C_n}$ and $\TypeToComp_{D_n}$ are well defined, which is necessary since $\Lang_{C_n} \subseteq \Lang_{C_m}$ and $\Lang_{D_n} \subseteq \Lang_{D_m}$ when $n\seq m$.

 \begin{lem} \label{lem:rank}
For all $`k \ele \Lang_C$ and $`d \ele \Lang_D$, if $\Rank(`k) \seq m$ and $\Rank(`d) \seq n$ then
 \[ \begin{array}{rcl}
\TypeToComp_{C_m}(`k) = \TypeToComp_{C_{\Rank(`k)}}(`k) 
	&\textrm{and}&
\TypeToComp_{D_n}(`d) = \TypeToComp_{D_{\Rank(`d)}}(`d).
 \end{array} \]
 \end{lem}

 \begin{proof} By easy induction over $m - \Rank(`k)$ and $n - \Rank(`d)$, respectively. 
Consider the case of $`d\prod`k\arrow `r$; then $\Rank(`d\prod`k\arrow `r) = p + 2$ where $p = \max \Set{\Rank(`d),\Rank(`k)}$ and
 \[ \begin{array}{rcl}
\Theta_{D_n}(`d\prod`k\arrow `r) &=& 
\StepFun{\Theta_{C_n}(`d\prod`k) }{ \Theta_R(`r)} \\ &=&
\StepFun{\Pair{\Theta_{D_{n-1}}(`d) }{ \Theta_{C_{n-1}}(`k)}}{ \Theta_R(`r)} . 
 \end{array} \]
If $p+2 \seq n$ then $\Rank(`d) \seq p < p+1 \seq n-1$ and $\Rank(`k) \seq p < p+1 \seq n-1$, so that by induction:
$\TypeToComp_{D_{n-1}}(`d) = \TypeToComp_{D_{\Rank(`d)}}(`d)$ and $\TypeToComp_{C_{n-1}}(`k) = \TypeToComp_{C_{\Rank(`k)}}(`k)$. Then:
 \[ \begin{array}{rcl} 
\StepFun{\Pair{\TypeToComp_{D_{n-1}}(`d) }{ \TypeToComp_{C_{n-1}}(`k)}}{ \TypeToComp_R(`r)} 
	&=& 
 \StepFun{\Pair{\TypeToComp_{D_{\Rank(`d)}}(`d) }{ \TypeToComp_{C_{\Rank(`k)}}(`k)}}{ \TypeToComp_R(`r)} \\
	&=& 
	\TypeToComp_{D_{p+2}}(`d\prod`k\arrow `r).
 \end{array} \]
\arrayqed[-22pt] 
 \end{proof}

For $A=D,C$, let $\seq_{A_n}$ be the pre-order $\seqA$ restricted to $\Lang_{A_n}$.

 \begin{lem} \label{lem:Theta_nSurRev}
For every $n$, the mappings $\TypeToComp_{C_n}$ and $ \TypeToComp_{D_n}$ are surjective and order reversing with respect to $\seq_{C_n}$ and $\seq_{D_n}$ respectively, \ie~they satisfy the hypotheses of Lem.\skp\ref{lem:compactRep}.
 \end{lem}

 \begin{proof}
By induction on the definition of $\TypeToComp_{C_n}$ and $ \TypeToComp_{D_n}$. 

 \begin{description}

 \item[$n=0$]
The language $\Lang_{C_0}$ is generated by the constant $`w$ and the connectives $\prod $ and $\inter$, so all types in $\Lang_{C_0}$ are equated by $\sim{C_0}$. Then the thesis holds for $\TypeToComp_{C_0}$, since $C_0=\Set{\bot}$. 
On the other hand, since the only way for a type in $\Lang_D$ to be of rank greater than $0$ is to include an arrow, $\Lang_{D_0}$ is generated by the constants $`w$ and $`y_a$ for $a \ele \Compact(R)$ and the connective $\inter$ so that $\Lang_{D_0} = \Lang_R$. 
Besides, the isomorphism $D_0 \Isom R$ is given by the mapping $\Semlambda \_ \, . \, r \to r$, which is the continuous extension of the mapping $\StepFun{\bot }{ a} \to a$ from $\Compact(D_0)$ to $\Compact(R)$. 
Then $\TypeToComp_{D_0}(a) = \StepFun{\bot }{ a} { \to} a = \TypeToComp_R(`y_a)$, 
where the last mapping is an isomorphism of ordered sets, hence order preserving and respecting.
We conclude that $\TypeToComp_{D_0} = \TypeToComp_{R}$ up to the isomorphism $\Compact(D_0) \Isom \Compact(R)$, hence satisfies the hypotheses of Lem.\skp\ref{lem:compactRep} by Lem.\skp\ref{lem:KRtheory}\,(\ref{lem:KRtheory-2}).

 \item[$n>0$]
If $`k \ele \Lang_{C_n}$ then:
 \[ \begin{array}{rcl}
 \TypeToComp_{C_n}(`k) 
	&=& 
\left \{
 \begin{array}{ll}
	\Pair{\TypeToComp_{D_{n{-}1}}(`d) }{ \TypeToComp_{C_{n{-}1}}(`k)} & (\textrm{if }`k = `d'\prod `k') \\
	\TypeToComp_{C_n}(`k_1)\join \TypeToComp_{C_n}(`k_2) & (\textrm{if }`k = `k_1\inter`k_2) \\
	\bot & (\textrm{if }`k = `w).
 \end{array} \right.
 \end{array} \]
For $`k = `d'\prod `k'$ we have $\Rank(`k) \seq n$ implies $\Rank(`d'),\Rank(`k') \seq n{-}1$ by definition of $\Rank$; hence 
$`d' \ele \Lang_{D_{n{-}1}}$ and $`k' \ele \Lang_{C_{n{-}1}}$. 
By induction, $\TypeToComp_{D_{n{-}1}}: \Lang_{D_{n{-}1}} \To \Compact(D_{n{-}1})$ and
$\TypeToComp_{C_{n{-}1}}: \Lang_{C_{n{-}1}} \To \Compact(C_{n{-}1})$ are onto and order reversing. 
Since $\Compact(C_n) = \Compact(D_{n{-}1}) \times \Compact(C_{n{-}1})$, 
by induction $\TypeToComp_{C_n}$ is onto and order reversing. 
If $`k = `k_1\inter`k_2$ then for any $`k_3 \ele \Lang_{C_n}$:
 \[ \begin{array}{rcl@{\quad}l}
\TypeToComp_{C_n}(`k_1)\join \TypeToComp_{C_n}(`k_2) \po \TypeToComp_{C_n}(`k_3) 
& \Iff & \TypeToComp_{C_n}(`k_i) \po \TypeToComp_{C_n}(`k_3) & (i =1,2) \\ 
& \Iff & `k_3 \seq_{C_{n}} `k_i & (\textrm{by a subordinate induction on $`k$}) \\
& \Iff & `k_3 \seq_{C_{n}} `k_1\inter `k_2.
 \end{array} \]
Finally, the case $`k = `w$ is obvious as $\bot = \Pair{\bot }{ \bot}$ is the bottom in $\Compact(C_n)$, while $`w \sim_{C_n} `w\prod `w$ is the top
in $(\Lang_{C_n}, \seq_{C_{n}})$.

If $`d \ele \Lang_{D_n}$ then:
 \[ \begin{array}{rcl}
\TypeToComp_{D_n}(`d) 
	&=& 
\left \{
 \begin{array}{ll}
\StepFun{\bot }{ \TypeToComp_R(`r)} & (\textrm{if }`d = `r \ele \Lang_R) \\
\StepFun{\TypeToComp_{C_n}(`k) }{ \TypeToComp_R(`r)} & (\textrm{if }`d = `k\arrow `r ) \\ 
\TypeToComp_{D_n}(`d_1)\join \TypeToComp_{D_n}(`d_2) & (\textrm{if }`d = `d_1\inter `d_2) \\ 
\bot & (\textrm{if }`d = `w)
 \end{array} \right.
 \end{array} \]
By construction 
 \[ \begin{array}{rcl}
\Compact(D_{n}) 
	&=& 
\Compact([C_n \To R]) \\
	&=& \Set{\join_{i \ele I} \StepFun{k_i }{ r_i} \mid I \textrm{ is finite} \And \Forall i \ele I \Pred[ k_i \ele \Compact(C_n) \And r_i \ele \Compact(R) ] }.
 \end{array} \] 
We know from the above that $\TypeToComp_{C_n}$ is surjective (since both $\TypeToComp_{D_{n{-}1}}$ and $\TypeToComp_{C_{n{-}1}}$ are), while $\TypeToComp_R$ is surjective by definition. 
Let $k_i = \TypeToComp_{C_n}(`k_i)$ and $r_i = \TypeToComp_R(`r_i)$, then
$\join_{i \ele I} \StepFun{k_i }{ r_i} = \TypeToComp_{D_n}(\Inter_{i \ele I} `k_i\arrow `r_i)$; since also $ \TypeToComp_{D_n}(`w) = \bot = \StepFun{\bot }{ \bot} = \TypeToComp_{D_n}(`w\arrow `w)$, we conclude that $\TypeToComp_{D_n}$ is surjective.

To see that $\TypeToComp_{D_n}$ is order reversing, note that $\StepFun{k }{ r} \po f$ for $f \ele [C_n \To R]$ if and only if $r \po f(k)$, which is trivially the case if $r = \bot$. Since $\bot = \TypeToComp_R(`w)$ and also $\bot = \StepFun{k }{ \bot} = \TypeToComp_{D_n}(`k\arrow `w)$ for any $k$ and $`k$, while $`k\arrow `w \sim{D_n} `w \geq_{D_n}`d$ for any $`d \ele \Lang_{D_n}$, the thesis trivially holds if $r = \bot$.

Suppose that $r \neq \bot$. 
Since
 $ \begin{array}{rcl}
 \left(\join_{i \ele I} \StepFun{k_i }{ r_i}\right)(x) &=& \join_{j \ele J}r_j 
 \end{array} $
for $J = \Set{j \ele I \mid k_j\po x}$, we have 
 \[ \begin{array}{lll}
 \StepFun{k }{ r} \po \join_{i \ele I} \StepFun{k_i }{ r_i} & \Iff &
r \po \join_{i \ele I} \StepFun{k_i }{ r_i}(k) \\ 
& \Iff & \Exists J \subseteq I \Pred[ r \po \join_{j \ele J}r_j \And \join_{j \ele J}k_j \po k ]
 \end{array} \]
By subjectivity of $\TypeToComp_{C_n}$ and $\TypeToComp_R$ we know that
there exist $`k,`r,$ such that $\TypeToComp_{C_n}(`k) = k$ and $\TypeToComp_{C_n}(`k_i) = k_i$,
and $`k_i,`r_i$ such that $\TypeToComp_R(`r) = r$,$\TypeToComp_{R}(`r_i) = r_i$, for every $i\ele I$. Therefore,
 \[ \begin{array}{lcll}
 \TypeToComp_{D_n}(`k\arrow `r) 
	& = & 
 \StepFun{\TypeToComp_{C_n}(`k) }{ \TypeToComp_R(`r)} \po 
 \join_{i \ele I} \StepFun{\TypeToComp_{C_n}(`k_i) }{ \TypeToComp_R(`r_i)} 
	\\ & \Iff & 
 \Exists J \subseteq I \Pred[ \TypeToComp_R(`r) \po \join_{j \ele J}\TypeToComp_R(`r_j) \And 
 \join_{j \ele J}\TypeToComp_{C_n}(`k_j) \po \TypeToComp_{C_n}(`k) ]
	\\ & \Iff & 
 \Exists J \subseteq I \Pred[ \TypeToComp_R(`r) \po \TypeToComp_R(\Inter_{j \ele J}`r_j) \And 
 \TypeToComp_{C_n}(\Inter_{j \ele J}`k_j) \po \TypeToComp_{C_n}(`k) ]
	\\ & \Iff & 
 \Exists J \subseteq I \Pred[ \Inter_{j \ele J}`r_j \seqR `r \And `k \seq_{C_n} \Inter_{j \ele J}`r_j ]
	\\ & \Iff & 
 \Inter_{i \ele I} (`k_i\arrow `r_i) \seq_{D_n} `k\arrow `r 
	\hfill \textrm{(by\skp\ref{lem:SubtypeProp}\,(\ref{lem:SubtypeProp-3}))}
 \end{array} \]
where we use that both $\TypeToComp_{C_n}$ (as proved above) and $\TypeToComp_R$ (by Lem.\skp\ref{lem:KRtheory}\,(\ref{lem:KRtheory-2})) are order reversing, and that\skp\ref{lem:SubtypeProp}\,(\ref{lem:SubtypeProp-3}) applies because $\TypeToComp_R(`r) \neq \bot$ if and only if $`r \neq_R `w$ by\skp\ref{lem:KRtheory}\,(\ref{lem:KRtheory-2}).
The general case 
 \[ \begin{array}{rcl@{\quad}lc}
 \TypeToComp_{D_n}(\Inter_{i \ele I}`k_i\arrow `r_i) 
	&=& 
 \join_{i \ele I}\TypeToComp_{D_n}(`k_i\arrow `r_i)
	&\po& 
 \TypeToComp_{D_n}(\Inter_{j \ele J}`k'_j\arrow `r'_j) 
 \end{array} \] 
now follows, since this is equivalent to
 \[ \begin{array}{rcl@{\quad}lccc}
 \TypeToComp_{D_n}(`k_i\arrow `r_i) 
	&=& 
 \StepFun{\TypeToComp_{C_n}(`k_i) }{ \TypeToComp_R(`r_i)} \\
	&\po& 
 \TypeToComp_{D_n}(\Inter_{j \ele J}`k'_j\arrow `r'_j) \\
	&=& 
\join_{j \ele J} \StepFun{\TypeToComp_{C_n}(`k'_j) }{ \TypeToComp_R(`r'_j)}
 \end{array} \]
for all $i \ele I$.%
\qed

 \end{description}
 \end{proof}

 \begin{defi} \label{def:typeInterpDandC}
The mappings $\TypeToComp_D:\Lang_D\arrow \Compact(D)$ and $\TypeToComp_C:\Lang_C\arrow \Compact(C)$ are defined by
 \[ \begin{array}{rcl}
\TypeToComp_D(`d) &=& \TypeToComp_{D_{\Rank(`d)}}(`d) \\
\TypeToComp_C(`k) &=& \TypeToComp_{C_{\Rank(`k)}}(`k)
 \end{array} \]
 \end{defi}

 \begin{rem} \label{rem:ThetaRank}
By Lem.\skp\ref{lem:rank} and of the definition of $\Rank$, $\TypeToComp_D(`k\arrow `r) = \StepFun{\TypeToComp_C(`k)}{\TypeToComp_R(`r)}$ and similarly $\TypeToComp_C(`d\prod`k) = \Pair{\TypeToComp_D(`d)}{\TypeToComp_C(`k)}$. In general all the equations
in Def.\skp\ref{def:typesToComp} concerning the mappings $\TypeToComp_{A_n}$ do hold for the respective maps
$\TypeToComp_{A}$.
 \end{rem}

 \begin{lem} \label{lem:Theta_AsurjOderRev}
The mappings $\TypeToComp_D$ and $\TypeToComp_C$ are surjective and order reversing, \ie~satisfy the hypotheses of Lem.\skp\ref{lem:compactRep}.
 \end{lem}

 \begin{Proof}
First observe that if $n \geq \Rank(`d)$ then $\TypeToComp_D(`d) = \TypeToComp_{D_{\Rank(`d)}}(`d) = \TypeToComp_{D_n}(`d)$ 
by Lem.\skp\ref{lem:rank}, and similarly for $\TypeToComp_C$.
Now if $d\ele \Compact(D)$ then by Equation (\ref{eq:compactUnion}) we have $\Compact(D) = \bigcup_n \Compact(D_n)$, so that there exists $n$ such that $d \ele \Compact(D_n)$. By Lem.\skp\ref{lem:Theta_nSurRev} $\TypeToComp_{D_n}$ is surjective, hence there exists $`d \ele \Lang_{D_n}$ such that $\TypeToComp_{D_n}(`d) = d$. Then $\Rank(`d) \seq n$ and
$\TypeToComp_D(`d) = \TypeToComp_{D_n}(`d)= d$, by the above remark. Hence $\TypeToComp_D$ is surjective.

On the other hand, if $`d_1 \seqD `d_2$ then $`d_1 \seq_{D_n} `d_2$ for any $n \geq \max\Set{\Rank(`d_1),\Rank(`d_2)}$; by 
Lem.\skp\ref{lem:Theta_nSurRev} and the above remark we conclude that
 \[ \begin{array}{rcl@{\quad}lccc}
\TypeToComp_D(`d_1) &=& \TypeToComp_{D_n}(`d_1) &\sqsupseteq& \TypeToComp_{D_n}(`d_2) &=& \TypeToComp_D(`d_1),
 \end{array} \]
which establishes that $\TypeToComp_D$ is order reversing. 
The proof concerning $\TypeToComp_C$ is similar. ~
 \end{Proof}

 \begin{thm} \label{thm:iso-theorem}
For $A = R,D,C$, the filter domain $\Filt_A$ is isomorphic to $A$.
 \end{thm}

 \begin{Proof}
That $\Filt_R \Isom R$ is stated in Cor.\skp\ref{cor:fromDomToFilter}. By Lem.\skp\ref{lem:Theta_AsurjOderRev}
 $\TypeToComp_D$ and $\TypeToComp_C$ satisfy Lem.\skp\ref{lem:compactRep}, hence we conclude by Prop.\skp\ref{prop:filter-iso}. ~
 \end{Proof}

Thm.\skp\ref{thm:iso-theorem} implies that $(\Filt_R,\Filt_D,\Filt_C)$ is a $`l`m$-model. However, it is a rather implicit description of the model on which we base the construction of the intersection type assignment system in the next section. 
To get a better picture relating term and type interpretation, below we will show how functional application and the operation of adding an element of $\Filt_D$ in front of a continuation in $\Filt_C$ are defined in this model; this provides us with a more explicit description of the isomorphisms relating $\Filt_D$ and $\Filt_C$.

In the following, we let $d$ and $k$ range over filters in $\Filt_D$ and $\Filt_C$, respectively; notice that above they were used for elements of $C$ and $D$. 
Since no confusion is possible, and since a clear link exists between these concepts, we permit ourselves a little overloading in notation.

 \begin{defi}
For $d \ele \Filt_D$ and $k \ele \Filt_C$ we define:
 \[ \begin{array}{rcl}
 d`.k &\ByDef& \filt_D\Set{`r \ele \Lang_R \mid \Exists `k\arrow `r \ele d \Pred[ `k \ele k ] } 
 	\\ 
 d\cons k &\ByDef& \filt_C\Set{\Inter_{i \ele I}`d_i\prod `k_i \ele \Lang_C \mid \Forall i \ele I \Pred[ `d_i \ele d \And `k_i \ele k ] }
 \end{array} \]
 \end{defi}

The upward closure $ \filt_D$ in the definition of $d`.k$ is redundant, since we can show that the set $\Set{\Inter_{i \ele I}`d_i\prod `k_i \ele \Lang_C \mid \Forall i \ele I \Pred[ `d_i \ele d \And `k_i \ele k ] }$ is a filter.
We have added $ \filt_D$ to simplify proofs; in fact any set of types $\filt A$ is clearly closed under $\sim$.
A similar remark holds for $\filt_C$ in the definition of $d\cons k$, where we have to include $`w$. Alternatively one could stipulate the usual convention that $\Inter_{i \ele I}`d_i\prod `k_i$ is syntactically the same as $`w$ when $I=\emptyset$.

 \begin{lem} \label{lem:appConsCont}
$d`.k \ele \Filt_R$ and $d\cons k \ele \Filt_C$, for any $d \ele \Filt_D$ and $k \ele \Filt_C$. 
Moreover, the mappings
$\_ `.\_$ and $\_ \cons \_$ are continuous in both their arguments.
 \end{lem}

 \begin{Proof} The proof that $d`.k$ is well defined and continuous is essentially the same as that with {\EATS} (see for example \cite{Amadio-Curien'98} Sect.~3.3).
The set $d\cons k$ is a filter by definition. By definition unfolding we have that 
 \[ \begin{array}{rcl@{\quad}lc}
d\cons k 
	&=& 
(\bigcup_{`d \ele d} \filt_D `d) \cons (\bigcup_{`k \ele k } \filt_C `k) 
	&=& 
\bigcup_{`d \ele d,`k \ele k } \filt_D `d \cons \filt_C `k,
 \end{array} \]
hence $\_ \cons \_$ is continuous.
 \end{Proof}


In the particular case of $[\Filt_C\To\Filt_R]$, step functions (see Sect.\skp\ref{sub:domainTheoretic}) take the form $\StepFun{\filt_C `k }{ \filt_R \psk `r}$. 
Indeed for $k \ele \Filt_C$ we have that $\filt_C `k \subseteq k$ if and only if
$`k \ele k$, so that
we have:
 \[ \begin{array}{lll}
 \StepFun{\filt_C `k }{ \filt_R \psk `r} (k) & = &
 \left\{
 \begin{array}{ll}
 \filt_R\psk `r & (\textrm{if }`k \ele k ) \\
 \filt_R\psk `w & (\textrm{ otherwise})
 \end{array} \right. \\ [4mm]
& = &
 \filt_D(`k\arrow `r)`.k
 \end{array} \]

Thus arrow types represent step functions. Similarly, the product of domains $X\prod Y$ ordered component-wise is a domain such that $\Compact(X\prod Y) = \Compact(X)\prod \Compact(Y)$. 
In case of $\Filt_D \times \Filt_C$ compact points are of the shape $\Pair{\filt_D`d }{ \filt_C`k }$, which corresponds to 
the filter $\filt_C `d\prod `k \ele \Filt_C$.
This justifies the following definition:

 \begin{defi} \label{def:iso-maps}
We define the following maps:
 \[ \begin{array}{r@{~}l@{\qquad}lcl}
F:& \Filt_D \To [ \Filt_C\To\Filt_R] &
F \, d \, k &=& d `. k 
	\\
G:&[ \Filt_C\To\Filt_R] \To \Filt_D &
G \, f &=& \filt_D\,\Set{\Inter_{i \ele I} `k_i\arrow `r_i \ele \Lang_D \mid \Forall i \ele I \Pred[ `r_i \ele f(\filt `k_i) ] } 
	\\
H:& \Filt_C \To (\Filt_D \times \Filt_C) &
H \, k &=& \Pair< \Set{`d \ele \Lang_D \mid `d\prod `k \ele k } , \Set{`k \ele \Lang_C \mid `d\prod `k \ele k } >
	\\
K:&(\Filt_D \times \Filt_C) \To \Filt_C &
K \Pair<d,k> &=& d \cons k 
 \end{array} \]
 \end{defi}

 \begin{rem} \label{rem:GK}
As expected form the claim that step functions in $[ \Filt_C\To\Filt_R]$ are represented by arrow types in $\Lang_D$, for any $`k\arrow `r \ele \Lang_D$ we have $G\StepFun{\filt_C`k}{\filt_R`r} = \filt_D(`k\arrow `r )$. 
Indeed, $`k\arrow `r \seqD `k'\arrow `r'$ if and only if $`k'\seq_C`k$ and $`r\seq_R`r'$, \ie~$\filt_C`k\subseteq\filt_C`k'$ and $\filt_R`r'\subseteq\filt_R`r$, 
if and only if $`r' \ele \filt_R`r = \StepFun{\filt_C`k}{\filt_R`r}(\filt_C`k')$.
Similarly, the type $`d\prod`k\ele \Lang_C$ represents pairs in $\Filt_D \times \Filt_C$ via $K$, \ie~$K\Pair{\filt_D`d}{\filt_C`k} = \filt_D`d \cons \filt_C`k = \filt_C(`d\prod`k)$.
 \end{rem}

When no ambiguity is possible, we will write $\filt `r$ for $\filt_R `r$, and similarly for $\filt_D$ and $\filt_C$.

 \begin{lem} \label{lem:filterMappingWellDef}
The functions $F$, $G$ and $H$, $K$ are well defined and monotonic with respect to subset inclusion.
 \end{lem}

 \begin{Proof} By Lem.\skp\ref{lem:appConsCont}, $F$ and $K$ are well defined and continuous, hence monotonic. 

For all $f \ele [ \Filt_C\To\Filt_R]$, by definition the set $G\,f$ is a filter over $(\Lang_D, \seqD)$; we check that $G$ is monotonic.
Observe that $\Inter_{i \ele I} `k_i\arrow `r_i \ele G\,f$ if and only if $\join_{i \ele I} \StepFun{\filt `k_i }{ \filt `r_i} \po f$; on the other hand,
if $f \po g$ then $\join_{i \ele I} \StepFun{\filt `k_i }{ \filt `r_i} \po f$ implies $\join_{i \ele I} \StepFun{\filt `k_i }{ \filt `r_i} \po g$, so $\Inter_{i \ele I} `k_i\arrow `r_i \ele G\,f$ implies $\Inter_{i \ele I} `k_i\arrow `r_i \ele G\,g$.

The function $H$ is evidently monotonic with respect to $\subseteq$.
We check that it is well defined, \ie~that both 
 \[ \begin{array}{rcll}
d' &=& \Set{`d \ele \Lang_D \mid `d\prod `k \ele k } & \textrm{and} \\
k' &=& \Set{`k \ele \Lang_C \mid `d\prod `k \ele k }
 \end{array} \] 
are filters whenever $k$ is one. 
Let $`d_1$, $`d_2 \ele d'$, then there exist $`k_1$, $`k_2$ such that $`d_1\prod `k_1$, $`d_2\prod `k_2 \ele k$ (and hence $`k_1$, $`k_2 \ele k'$).
Since $k$ is a filter, we have $`d_1\prod `k_1\inter `d_2\prod `k_2 \ele k$; also, $`d_1\prod `k_1\inter `d_2\prod `k_2 \simC (`d_1\inter`d_2)\prod (`k_1\inter`k_2)$ implies $(`d_1\inter`d_2)\prod (`k_1\inter`k_2) \ele k$, as $k$, being a filter, is closed under meets and $\simC$. 
We conclude that $`d_1\inter`d_2 \ele d'$; similarly, we can reason that $`k_1\inter`k_2 \ele k'$. 

The same reasoning shows that both $d'$ and $k'$ are upward closed sets with respect to $\seqD$ and $\seqC$, respectively.
 \end{Proof}

We can show that the following isomorphisms exist: 

 \begin{thm} \label{iso-theorem}
$\Filt_D \Isom [\Filt_C\To\Filt_R]$ via $F$ with inverse $G$, and $\Filt_C \Isom \Filt_D \times \Filt_C$ via $H$ with inverse $K$.
 \end{thm}

 \begin{proof} Since any monotonic function of posets that is invertible is an isomorphism,
by Lem.\skp\ref{lem:filterMappingWellDef} it suffices to show that $G = F^{-1}$ and $K = H^{-1}$.

 \begin{enumerate} \itemsep 2pt

 \item
 $ \begin{array}[t]{@{}lll@{\kern-2cm}l}
(F \circ G)\,f\,k 
	&=& 
F (\filt\,\Set{\Inter_{i \ele I} `k_i\arrow `r_i \ele \Lang_D \mid \Forall i \ele I \Pred[ `r_i \ele f(\filt `k_i) }) \, k ] \\
	&=& 
\filt\,\Set{\Inter_{i \ele I} `k_i\arrow `r_i \ele \Lang_D \mid \Forall i \ele I \Pred [ `r_i \ele f(\filt `k_i) } `.k ] \\ 
	&=& 
\filt\Set{`r \mid \Exists `k \ele k \Pred [ `r \ele f(\filt `k)} = \bigcup_{`k \ele k } f(\filt`k) ] \\ 
	&=& 
\join_{\filt`k \subseteq k } f(\filt`k) & (\textrm{since }\Set{\filt`k \mid `k \ele k } \textrm{ is directed}) \\ 
	&=& 
f(k) & (\textrm{by continuity of }f)
 \end{array} $

\noindent
hence $(F \circ G)\,f = f$.

 \item
 $ \begin{array}[t]{@{}lllll}
(G \circ F) \, d 
	& = & 
G(\Semlambda k \ele \Filt_C \, . \, d`.k) \\ 
	& = & 
\filt \Set{\Inter_{i \ele I} `k_i\arrow `r_i \mid \Forall i \ele I \Pred [ `r_i \ele d `.\filt `k_i ] } \\ 
	& = & 
\filt \Set{\Inter_{i \ele I} `k_i\arrow `r_i \mid \Forall i \ele I ~ \Exists `k'_i \Pred [`k_i \seqC `k'_i \And `k'_i\arrow `r_i \ele d ] } 
	& = & 
d
 \end{array} $

\noindent
where $\Semlambda$ represents semantic abstraction.
In the last equation, the inclusion $\supseteq$ is obvious, while the inclusion $\subseteq$ follows by the fact that if $`k_i \seqC `k'_i$ then $`k'_i\arrow `r_i \seqD `k_i\arrow `r_i$, hence $`k'_i\arrow `r_i \ele d$ implies $`k_i\arrow `r_i \ele d$ for all $i \ele I$, which in turn implies that $\Inter_{i \ele I} `k_i\arrow `r_i \ele d$.

 \item
 $ \begin{array}[t]{@{}lllll}
	(H \circ K) \Pair{d }{ k } & = & H (d\cons k) \\ 
	& = & \Pair< \Set{`d \ele \Lang_D \mid `d\prod `k \ele d\cons k } , \Set{`k \ele \Lang_C \mid `d\prod `k \ele d\cons k } > 
	& = & \Pair{d }{ k }
 \end{array} $

\noindent
by observing that $`d\prod `k \ele d\cons k$ if and only if $`d \ele d$ and $`k \ele k$, and that if $`k' \ele \filt`w \subseteq d\cons k$ then $`k' \simC `w\prod `w$ and obviously $`w \ele d$ and $`w \ele k$.

 \item
 $ \begin{array}[t]{@{}llll}
	(K \circ H)\, k & = &
	\Set{`d \ele \Lang_D \mid `d\prod `k \ele k } \cons \Set{`k \ele \Lang_C \mid `d\prod `k \ele k } \\ 
	& = & \filt \Set{\Inter_{i \ele I}`d_i\prod `k_i \mid \Forall i \ele I \, \Exists `d'_i, `k'_i \Pred [ d_i\prod `k'_i, `d'_i\prod `k_i \ele k ] } \\
	& = & k & (\textrm{since }k \ele \Filt_C )
 \end{array} $
\arrayqed
 \end{enumerate}
 \end{proof}

 \begin{rem} \label{rem:naturalEquated}
As observed in Remark\skp\ref{rem:ExtLamMod}, a $`l`m$-model is an extensional $`l$-model. 
Thm.~3.1 in \cite{Streicher-Reus'98} states that the initial/final solution of the continuation domain equations is isomorphic to the domain $R_ \infty \Isom [R_ \infty\To R_\infty]$, \ie~Scott's $D_\infty$ $`l$-model obtained as inverse limit of a chain where $D_0 = R$.

To see this from the point of view of the intersection type theory, consider the extension $\Lang_{`l} = \dots \mid `d\arrow `d $ of $\Lang_{D }$. Let $ \TypeTheor_{`l}$ be the theory obtained by adding to $ \TypeTheor_D$ the equation $ `d\prod `k\arrow `r = `d\arrow `k\arrow `r $.
Then in the intersection type theory $ \TypeTheor_{`l}$, the following rules are derivable:
 \[ \begin{array}{c@{\dquad}c}
\Inf	{(`d\arrow `d_1) \inter (`d\arrow `d_2) \seq_{`l} `d\arrow(`d_1 \inter`d_2) }
&
\Inf	{`d'_1 \seq_{`l} `d_1 \quad `d_2 \seq_{`l} `d'_2
	}{ `d_1\arrow `d_2 \seq_{`l} `d'_1\arrow `d'_2 }
 \end{array} \]
 
By this, $ \TypeTheor_{`l}$ is a \emph{natural equated} intersection type theory in terms of \cite{Alessi-Severi'08},
and hence $ \Filt^{`l} \Isom [ \Filt^{`l}\To\Filt^{`l}]$ where $ \Filt^{`l}$ is the set of filters generated by the pre-order $\seq_{`l}$ (see \cite{Alessi-Severi'08}, Cor.~28\,(4)). 
 \end{rem}

 \section{An intersection type system} \label{sec:types}

Let ${\ModM} = (R,D,C)$ be a $`l`m$-model, where $D,C$ are initial solutions of the continuation domain equations (we say then that $\ModM$ is \emph{initial}).
In this section, using the fact that ${\ModM}$ is isomorphic to the filter model $\Filt = (\Filt_R,\Filt_D,\Filt_C)$, as established by Thm.\skp\ref{thm:iso-theorem} and \ref{iso-theorem}, we will define a type assignment system such that the statement $M : `d$ (or $\Cmd: `k$) is derivable, under appropriate assumptions about the variables and names in it, if and only if $\SemF { M }{ e } \ele \filt_D `d$ (or $\SemF { \Cmd}{ e } \ele \filt_C `k$) for all environments $e$ respecting those assumptions.

Thereby an interpretation of types can be defined such that $\SemF { `s } = \filt_A\psk `s$ for $A=D,C$. Since filters are upward closed sets of types, we have that $\SemF { T}{ e } \ele \filt_A\psk `s$ if and only if $`s \ele \SemF { T}{ e }$,
and we obtain that the denotation of a term/command is just the set of types that can be inferred for it in the assignment system.

 \subsection{Type assignment} \label{subsec:Type assignment}
We now give some preliminary definitions for our type system.

 \begin{defi}[Bases, Name Contexts, and Judgements] 

 \begin{enumerate}

 \firstitem
A \emph{basis} is a finite mapping from term variables to types in $\TypeTheor_D$, written as a finite set $`G = \Set{x_1{:}`d_1, \ldots,x_n{:}`d_n}$ where the term variables $x_i$ are pairwise distinct.

 \item
A \emph{name context} (or \emph{context}) is a finite mapping from names to types in $\TypeTheor_C$, written as a finite set $`D = \Set{`a_1{:}`k_1, \ldots, `a_m{:}`k_m}$ where the continuation variables $`a_i$ are pairwise distinct.

 \item
We write $`G,x{:}`d$ for the basis $`G \union \Set{x{:}`d}$, and assume that either $x$ does not occur in $`G$ or $x{:}`d \ele `G$, and similarly for $`a{:}`k,`D $.

 \item
We write $`G \except x $ for $`G \except \Set{x{:}`G(x)}$ and $`D \except `a$ for $`D \except \Set{`a{:}`D(`a)}$.

 \item \label{def:ext-leq-i}
Let $`G$ be a basis and $`D $ a name context. We define 
 \[ \begin{array}{rcl}
 \dom(`G) &\ByDef& \Set{x \mid \Exists `d  \Pred[ x{:}`d \ele `G] }
 \\ 
 \dom(`D) &\ByDef& \Set{`a \mid \Exists `k  \Pred[ `a{:}`k \ele `D] }
 \end{array} \]
and write $x \notele `G$ ($`a \notele `D $) if $x \notele \dom(`G)$ ($`a \notele \dom(`D)$).

 \item
A \emph{judgement} is an expression of the form $ \derLmu `G |- M : `d | `D $ or $ \derLmu `G |- {\Cmd} : `k | `D $ where $`G$ is a basis and $`D $ is a name context.
$M$ and $\Cmd$ are the \emph{subjects} and $`d$ and $`k$ the \emph{predicates}.

 \end{enumerate}
 \end{defi}
\noindent
We will occasionally allow ourselves some freedom when writing basis and contexts, and also consider $`G,x{:}`w$ a basis and $`a{:}`w,`D $ a context.

Judgements are in appearance very similar to Parigot's (see Sect.\skp\ref{sec:Parigot}), apart from the obvious difference in the language of types; in fact, there exists a relation between Parigot's system and the one presented here, which will be treated in detail in Sect.\skp\ref{sec:Parigot}.
Since bases and contexts are sets, the order in which variable and name assumptions are listed is immaterial.

We will occasionally treat basis and contexts as total functions by using the following notation:
 \[ \begin{array}{rcl}
 \fun{`G}(x) &=& 
 \begin{cases}{l@{\quad}l}
	`d & (\textrm{if }x{:}`d \ele `G) \\
	`w & (\textrm{otherwise})
 \end{cases}
\\ [4mm]
 \fun{`D }(`a) &=& 
 \begin{cases}{l@{\quad}l}
	`k & (\textrm{if }`a{:}`k \ele `D) \\
	`w & (\textrm{otherwise})
 \end{cases}
 \end{array} \]
Notice that then $`G \except x $ corresponds to the function update $`G[x:= `w]$ and $`D \except `a$ to $`D[`a:= `w]$.

 \begin{defi}[Intersection type system for $`l`m $] \label{def:intersTypeAss}
We define intersection type assignment for $`l`m$ through the following sets of inference rules: 
 \begin{description} 

 \item [\emph{Type rules}]
 $ \kern-2cm \begin{array}[t]{rl@{\dquad}rl}
 \multicolumn{4}{c}{
 \begin{array}{rl}
(\Axiom) : & \Inf	{\derLmu `G, x{:}`d |- x : `d | `D }
 \end{array} }
 \\ [4mm]
(\Abs) : &
\Inf	[`G(x) = `d]
	{\derLmu `G |- M : `k\arrow `r | `D }
	{\derLmu `G \except x |- `lx.M : `d\prod `k\arrow `r | `D }
&
(\App) : &
\Inf	{\derLmu `G |- M : `d\prod `k\arrow `r | `D
	\quad
	\derLmu `G |- N : `d | `D
	}{ \derLmu `G |- MN : `k\arrow `r | `D } 
 \\ [5mm]
(\TCmd) : &
\Inf	[\fun{`D }(`a)=`k]	
	{\derLmu `G |- M : `d | `D }
	{\derLmu `G |- [`a]M : `d\prod `k | `D }
&
(\MAbs) : &
\Inf	[\fun{`D }(`a)=`k]
	{\derLmu `G |- {\Cmd} : (`k'\arrow `r)\prod `k' | `D }
	{\derLmu `G |- `m`a.{\Cmd} : `k\arrow `r | `D \except `a }
\\ ~
 \end{array} $ 

 \item [\emph{Logical rules}]
~
 \[ \begin{array}{rl@{\dquad}rl@{\dquad}rl}
(\inter) : &
\Inf	{\derLmu `G |- T : `s | `D
	\quad
	\derLmu `G |- T : `t | `D
	}{ \derLmu `G |- T : `s\inter`t | `D }
&
(`w) : &
	\Inf	{\derLmu `G |- T : `w | `D }
&
(\seq) : &
\Inf	{\derLmu `G |- T : `s | `D \quad `s \seq `t }
	{\derLmu `G |- T : `t | `D }
 \end{array} \]

 \end{description}

We will write $ \derLmu `G |- T : `s | `D $ if there exists a derivation built using the above rules that has this judgement in the bottom line, and $\Der :: \derLmu `G |- T : `s | `D $ if we want to name that derivation.

 \end{defi}

As mentioned above, we extend Barendregt's convention to judgements $ \derLmu `G |- T : `s | `D $ by seeing the variables that occur in $`G$ and names in $`D $ as binding occurrences over $T$ as well; in particular, we will assume that no variable in $`G$ and no name in $`D $ is bound in $T$.

To understand these rules we can think of types as properties of term denotations in the 
initial model ${\ModM} = (R,D,C)$. In particular, if $`s \ele \Lang_A$ 
then $`s$ denotes a subset $\Sem{ `s }^{\ModM} \subseteq A$, for $A = R,D,C$. 
The judgement $ \derLmu `G |- T : `s | `D $ is then interpreted as the claim
that $\Sem{ T }^{\ModM}{ e } \ele \Sem{ `s }^{\ModM}$ whenever $e\psk x \ele \Sem{`G(x)}$ and $e\psk `a \ele \Sem{`D(`a)}$ for all $x$ and $`a$ (the formal definitions will be given in Sect.\skp\ref{subsec:typeInterpretation}).

The \emph{logical} rules, which are familiar from intersection type systems for the standard $`l$-calculus, just state that types are sets: $`w$ is the largest set which coincides with the domain of interpretation itself, the pre-order is subset inclusion, and $`s\inter`t$ is the set theoretic intersection of $`s$ and $`t$. 
Note that the subject in the conclusion of a logical rule is the same as in the premises. 
Moreover, remark that we use here the term `logical' in the sense of Abramsky's domain logic, not in the sense of (propositional) logic of any kind. 
In particular, intersection is \emph{not} conjunction, both in systems for the $`l$-calculus and in the present one.

The \emph{type rules} are syntax directed
; they have been obtained from the equations in Def.\skp\ref{def:interpretation} by representing the left-hand side of the equation in the conclusion and the right-hand side in the premises of the corresponding rule:

 \begin{description}

 \item [$\LAbs$] 
This rule corresponds to the equation $\Sem[`lx.M]^D{ e }{ \Pair<d,k> } = \Sem[M]^D{ e[x \to d] }{ k }$, where $\Sem[\cdot]^D$ is short for $\Sem[\cdot]_\Der^{\ModM}$. 
It states that $`lx.M$ is a function of continuations $\Pair<d,k>$, whose values are those of $M$ where $x$ is interpreted by $d$, and applied to continuation $k$. 
On the other hand, the arrow types from $\Lang_D$ represent properties of functions:
a property of $`lx.M$ is then a type $`d\prod `k\arrow `r$ (the conclusion of the rule) so that whenever $`d\prod `k$ is a property of $\Pair<d,k>$, \ie~$d \ele \Sem[`d]^{\ModM}$ and $k \ele \Sem[`k]^{\ModM}$, $`r$ is a property of the result. 
But since the result is $\Sem[M]^D{ e[x \to d] }{ k }$, it suffices to prove that $M$ has the property $`k\arrow `r$ whenever $x$ is interpreted by $d$, which is represented by the assumption $x{:}`d$ in the premise of $(\LAbs)$.

 \item [$\App$] 
Dually, this rule comes from the equation $\Sem[MN]^D{ e }{ k } = \Sem[M]^D{ e }{ \Pair<\Sem[N]^D{ e } , k > }$. 
For the application $MN$ to have the property $`k\arrow `r$ (as in the conclusion)
it suffices that if applied to a continuation $k \ele \Sem{ `k }^{\ModM}$ it yields a value with property $`r$. By the equation, such a value is
computed by putting $\Sem{ N }^D{ e }$ before $k$ in the continuation passed to $\Sem{ M }^D{ e }$. 
Therefore, for the conclusion to hold it suffices to prove that $N$ has type $`d$ and $M$ type $`d\prod `k\arrow `r$. 

 \item [$\TCmd$] 
This rule is based on the equation $\Sem{ [`a]M}^C{ e } = \Pair<\Sem{ M }^D { e } , e\psk `a >$, which states that the meaning of a command $[`a]M$ is a continuation $\Pair<d,k>$ where $d$ is the meaning of $M$ and $k = e \, `a $. 
For $\Pair<d,k>$ to have the property $`d\prod `k$ (as in the conclusion) we have to check that $M$ has the property $`d$ whenever $`a$ denotes the continuation $k$ with property $`k$. 
Since the assumptions about the environment are in the contexts $`D $ in case of names, this is represented by the side condition $\fun{`D }(`a)=`k$ of the rule.

 \item [$\MAbs$] 
This rule is the more involved case, which corresponds to the equation
$\Sem{ `m`a.\Cmd}^D{ e }{ k } = d \, k' $, where $\Pair<d , k'> = \Sem{ \Cmd}^C{ e[`a \to k] }$.
This states that $\Sem{ `m`a.\Cmd}^D{ e }$ is the function that, when applied to a continuation $k$ yields the value of the application of the first component $d$ to the second component $k'$ of a different continuation $\Pair<d , k'>$, which however depends on $k$, because it is computed by $\Cmd$ whenever $`a$ is sent to $k$. Now the result $ d \, k'$
will have the property $`r$ if for some $`k'$ both $k' \ele \Sem{ `k' }^{\ModM}$ and $d \ele \Sem{ `k'\arrow `r}^{\ModM}$.
Therefore, to type $`m`a.\Cmd$ by $`k\arrow `r$ (as in the conclusion) we have to ensure that the continuation represented by $\Cmd$ has the property $(`k'\arrow `r)\prod `k'$, whenever $`a{:}`k$ occurs in the context (as in the premise).

 \end{description}

 \begin{rem}
Note how rules $(\App)$ and $(\LAbs)$ are actually instances of the familiar rules for application and $`l$-abstraction in the simply typed $ `l$-calculus. In fact, $`d\prod `k\arrow `r \ele \Lang_D$ is equivalent to $`d\arrow(`k\arrow `r) \ele \Lang_{`l}$ so that, if we admitted types of $\Lang_{`l}$, the following rules would be admissible:
 \[ \begin{array}{rl@{\dquad}rl}
(\arrE): & 
\Inf	{\derLmu `G |- M : `d\arrow (`k\arrow `r) | `D
	\quad
	\derLmu `G |- N : `d | `D
	}{ \derLmu `G |- MN : `k\arrow `r | `D }
&
(\arrI): & 
\Inf	{\derLmu `G,x{:}`d |- M : `k\arrow `r | `D
	}{ \derLmu `G |- `lx.M : `d\arrow (`k\arrow `r) | `D }
 \end{array} \]
We will use the following variant of rule $(\inter)$:
 \[ \begin{array}{rl}
(\Inter) : &
\Inf	{\derLmu `G |- T : `s_i | `D
	\quad
	(\Forall i \ele I)
	}{ \derLmu `G |- T : \Inter_{i \ele I}`s_i | `D }
 \end{array} \]

 \end{rem}

 \begin{rem} \label{rem:TCmd-rule}
Rule $(\TCmd)$ is equivalent to the following two:
 \[ \begin{array}[t]{rl@{\dquad}rl}
(\TCmd_1 ) : &
\Inf	{\derLmu `G |- M : `d | `D }
	{\derLmu `G |- [`a]M : `d\prod `k | `a{:}`k,`D } &
(\TCmd_2) : &
\Inf	[`D(`a) = `w]
	{\derLmu `G |- M : `d | `D }
	{\derLmu `G |- [`a]M : `d\prod `w | `D }
 \end{array} \]
By definition, the context $`a{:}`k,`D $ in the conclusion of $(\TCmd_1)$ is only legal when either $`a\notele \dom(`D)$ or $`a{:}`k\in`D $.

The need of $(\TCmd_2)$ will become apparent when proving the admissibility of the strengthening rule (Lem.\skp\ref{lem:strengthening}) and the completeness of the type assignment. 
For the moment we observe that with $(\TCmd_1 )$ the conclusion of the shape $ \derLmu `G |- {[`a]}M : {`d\prod `w} | `D $ would be derivable from $ \derLmu `G |- M : `d | `D $ only if
$`a{:}`w\ele `D $; note that $\fun{`D }(`a)=`w$ does not require that $`a\ele \dom(`D)$. 
On the other hand, $(\TCmd_2)$ allows the implicit typing of $`a$ by $`w$ even if $`a\notele \dom(`D)$.
Not having rule $(\TCmd_2)$ (that is a particular case of $(\TCmd)$) would introduce an asymmetry with respect to the typing with $`w$ of the term variable $x$ in a basis $`G$, since we can conclude $ \derLmu `G |- x : `w | `D $ either by rule $(\Ax)$ (in which case $x{:}`w \ele `G$ is required), or by rule $(`w)$, where $x\notele `G$ is allowed.

With the above proviso, in the proofs we shall often consider the rules $(\TCmd_1 )$ and $(\TCmd_2 )$ instances of $(\TCmd)$ without explicit mention.
 \end{rem}

The admissibility of the following rules will be useful:

 \begin{lem}[Admissibility of Weakening and Thinning] \label{lem:weakening}
The following rules are admissible:
 \[ \begin{array}{rl}
(\Weak): &
\Inf	[ `G \subseteq `G' \And `D \subseteq `D' ]
	{ \derLmu `G |- T : `s | `D 
	}{ \derLmu `G' |- T : `s | `D' }
 \\ [5mm]
(\Thin): &
\Inf	[`G' \supseteq \Set {x{:}`d\ele`G \mid x \ele \fv(T) }, ~ `D' \supseteq \Set {`a{:}`k\ele`D \mid `a \ele \fn(T) }]
	{\derLmu `G |- T : `s | `D
	}{\derLmu `G' |- T : `s | `D' }
 \end{array} \] 
 \end{lem}

 \begin{Proof} Easy. 
 \end{Proof}

Notice that, by our interpretation of Barendregt's convention, the variables in $`G'$ and names in $`D'$ are not bound in $T$.

\Comment{
 \begin{rem}
Notice that weakening is only sound when extending Barendregt's convention to judgements.
If we would not do so, then, for example, we could derive
 \[ \begin{array}{c}
\Inf	[\Weak]
	{\Inf	[\LAbs]
{\Inf	[\Ax]
	{ \derLmu x{:}`k\arrow `r |- x : `k\arrow `r | {} }
}{ \derLmu {} |- `lx.x : (`k\arrow `r)\prod `k\arrow `r | {} }
	}{ \derLmu x{:}`d |- `lx.x : (`k\arrow `r)\prod `k\arrow `r | {} }
 \end{array} \]

Note that we cannot apply rule $(\LAbs)$ to this derivation, since that would create the term $`lx.`lx.x$, which is not a valid term.
We can solve this through $`a$-conversion, replacing $`lx.x$ by $`ly.y$; 
 \[ \begin{array}{c}
\Inf	[\Weak]
	{\Inf	[\LAbs]
{\Inf	[\Ax]
	{ \derLmu y{:}`k\arrow `r |- y : `k\arrow `r | {} }
}{ \derLmu {} |- `ly.y : (`k\arrow `r)\prod `k\arrow `r | {} }
	}{ \derLmu x{:}`d |- `ly.y : (`k\arrow `r)\prod `k\arrow `r | {} }
 \end{array} \]
notice that then the application of rule $(\Weak)$ is valid according to our criteria.
 \end{rem}
}

In presence of the subtyping $(\seq)$ we can have a further form of weakening, namely by weakening the types in the assumptions. 
We first extend the operator $\inter$ and the pre-orders $\seqD$ and $\seqC$ to bases and contexts.

 \begin{defi} \label{def:ext-leq}

 \begin{enumerate}

 \firstitem \label{def:ext-leq-ii}
For bases $`G_1, `G_2$ we define the basis $`G_1\inter\,`G_2$ by:
 \[ \begin{array}{rcl}
`G_1\inter\,`G_2 &\ByDef& 
	\Set{x{:}`G_1(x)\inter\, `G_2(x) \mid x \ele \dom(`G_1) \cap \dom(`G_2) } \\ 
	&& \quad \Union \Set {x{:}`d \ele `G_1 \mid x \notele \Dom(`G_2) } \\ 
	&& \quad \Union \Set {x{:}`d \ele `G_2 \mid x \notele \Dom(`G_1) }
 \end{array} \]
For contexts $`D_1,`D _2$, we define the context $`D_1\inter\,`D _2$ similarly.
 
 \item \label{def:ext-leq-iii}
We extend the relations $\seqD$ and $\seqC$ to bases and contexts respectively by:
 \[ \begin{array}[t]{rcl}
`G_1 \seqD `G_2 & \ByDef & 
	\Forall x\ele \TVar \Pred[ `G_1(x) \seqD `G_2(x) ]
	\\ 
`D_1 \seqC `D_2& \ByDef & 
	 \Forall `a\ele \CVar \Pred[ `D_1(`a) \seq_C `D_2(`a) ]
 \end{array} \]

 \end{enumerate}
 \end{defi}

Note that, if $`G_1$, $`G_2$ are well-formed bases then so is $`G_1\inter\,`G_2$, and if $`D_1$, $`D _2$ are well-formed contexts, then so is $`D_1\inter\,`D _2$. 
Also, $\dom(`G_1\inter\,`G_2) = \dom(`G_1) \union \dom(`G_2)$ and
$\dom(`D_1\inter\,`D _2) = \dom(`D_1) \union \dom(`D_2)$. 
Therefore $`G_1\inter\,`G_2$ is often called `union of bases' in the literature.

The relations $`G_1 \seqD `G_2$ and $`D_1 \seqC `D_2$ are the pointwise extensions of the relations $\seqD$ and $\seqC$ over types; note that the quantifications are not restricted to the domains of the bases nor of the contexts.

Another immediate consequence of the definition is that $`G_1 \inter\,`G_2 \seqD `G_i$ and $`D_1 \inter\,`D _2 \seqC `D_i$, for $i = 1,2$. 
However if $`G_1 \seqD `G_2$ then $\dom(`G_1)$ and $\dom(`G_2)$ are unrelated in general, since we have for example $\Set{x{:}`w,y{:}`d_1\inter`d_2} \seqD \Set{z{:}`w,y{:}`d_1}$. 
Therefore $`G_1 \seqD `G_2$ does not imply that $`G_1$ and $`G_1 \inter\,`G_2$ are the equal, as one perhaps would expect; this is however without consequence since in this case $`G_1 \seqD `G_1 \inter\,`G_2$, so that using the admissibility of strengthening to be shown below, one can prove that all the typings obtainable by means of either basis, can be obtained by the other one.
A similar remark holds for contexts.

Now we are in place to prove the admissibility of strengthening:

 \begin{lem}[Admissibility of Strengthening] \label{lem:strengthening}
The following rule is admissible:
 \[ \begin{array}{rl}
(\Strength): &
\Inf	[ `G' \seqD `G \And `D' \seqC `D ]
	{ \derLmu `G |- T : `s | `D 
	}{ \derLmu `G' |- T : `s | `D' }
 \end{array} \] 
 \end{lem}

 \begin{proof}
By straightforward induction over the structure of derivations.\qed
{
The proof is by straightforward induction over the structure of derivations; we only deal with some interesting cases. 

 \begin{description}
 \item [$ \Ax $] 
Then $T\equiv x$ and $`s = `G(x)$. 
We can construct: 
 \[ \begin{array}{c}
\Inf	[ `G'(x) \seqD `G(x)]
	{\Inf	[\Ax]{}
{ \derLmu `G' |- x : `G'(x) | `D' } 
	}{ \derLmu `G' |- x : `G(x) | `D' }
 \end{array} \]
where $`G'(x) \seqD `G(x)$ follows by the assumption that $`G' \seqD `G$.

 \item [$ \Abs $] 
Then $T \equiv `lx.M$, $`s = `d\prod`k\arrow `r$ and the inference ends with:
 \[ \begin{array}{c}
\Inf	[\Abs]
	{\InfBox{ \derLmu `G,x{:}`d |- M : `k\arrow `r | `D } 
	}{\derLmu `G |- `lx.M : `d\prod `k\arrow `r | `D }
 \end{array} \]
where $x \notele \dom(`G)$. 
By our extension to Barendregt's convention over judgements, we can assume that also $x \notele \dom(`G')$. 
Thereby $`G', x{:}`d$ is a well-formed basis, and clearly $`G' \seqD `G$ implies $`G', x{:}`d \seqD `G, x{:}`d$. 
By induction we have $ \derLmu `G',x{:}`d |- M : `k\arrow `r | `D' $, from which we conclude $ \derLmu `G' |- `lx.M : `d\prod `k\arrow `r | `D' $ by rule $(\Abs)$.


 \item [$ \TCmd $] 
Then $T\equiv [`a]M$ and $`s = `d\prod `k$, where $`k=`D(`a)$, and the derivation ends by
 \[ \begin{array}{c}
\Inf	[\TCmd]
	{\InfBox{ \derLmu `G |- M : `d | `D }
	}{ \derLmu `G |- [`a]M : `d\prod `k | `D }
 \end{array} \]
Since $`D'(`a)=`k'\seqC`k$ implies $`d\prod `k' \seqC `d\prod `k$, by induction $ \derLmu `G' |- M : `d | `D' $, we can construct:
 \[ \begin{array}{c}
\Inf	[\seq]
	{\Inf	[\TCmd]
{\InfBox{ \derLmu `G' |- M : `d | `D' }
}{ \derLmu `G' |- [`a]M : `d\prod `k' | `D' } 
	 \quad
	 \InfBox{`d\prod `k' \seqC `d\prod `k}
	}{ \derLmu `G' |- [`a]M : `d\prod `k | `D' }
 \end{array} \]
Note that in case $`a{:}`w\ele `D $ but $`a\notele \dom(`D')$ we still would have
that $`D'\seqC `D $, but the instance $(\TCmd_1)$ of rule $(\TCmd)$ would not be applicable.
 \qed

 \end{description}
} 
 \end{proof}
It is straightforward to show that $`G \subseteq `G'$ implies $`G' \seq `G$ (and $`D \subseteq `D'$ implies $`D' \seq `D$), so rule $(\Strength)$ contains rule $(\Weak)$.

The following lemma describes the set of types that can be assigned to a term or a command. 

 \begin{lem} \label{lem:logicalRules}
If $\Der :: \derLmu `G |- T : `s | `D $, then either $`s \simA `w$ or there exist sub-derivations $\Der_i :: \derLmu `G |- T : `s_i | `D $ of $\Der$ $(i \ele \n)$, such that $\Inter_{i=1}^n`s_i \seqA `s$ and the last rule of each $\Der_i$ is a type rule.
 \end{lem}

 \begin{Proof}
By straightforward induction over the structure of derivations. 
 \end{Proof}

This particular property will be of use in many of the proofs below, where we reason by induction over the structure of derivations and allows us to always assume that a type rule was applied last, and not treat the logical rules.

Another way to state the above result is the following:

\newpage

\def\derLmuInt{\derLmu}

 \begin{lem}[Generation Lemma] \label{mu gen Lemma} \label{gen lemma}
Let $`d \not\sim_D `w$ and $`k \not\sim_C`w$:
 \[ \begin{array}{rcl}
\derLmuInt `G |- x : `d | `D &\Iff& \Exists x{:}`d \ele `G \Pred[ `d' \seqD `d ]
	\\
\derLmuInt `G |- `lx.M : `d | `D &\Iff& 
\Exists I \, \Forall i \ele I ~ \Exists `d_i,`k_i,`r_i \Pred[ \derLmuInt `G, x{:}`d_i |- M : `k_i\arrow `r_i | `D \And \Inter_I `d_i\prod `k_i\arrow `r_i \seqD `d ] 
	\\
\derLmuInt `G |- MN : `d | `D &\Iff& 
\\ \multicolumn{3}{r} {
\Exists I \, \Forall i \ele I ~ \Exists `d_i,`k_i,`r_i \Pred[ \derLmuInt `G |- M : `d_i\prod `k_i\arrow `r_i | `D \And \derLmuInt `G |- N : `d_i | `D \And \Inter_I`k_i\arrow `r_i \seqD `d ] 
}
\quad
	\\
\derLmuInt `G |- `m`a.{\Cmd} : `d | `D &\Iff& 
\Exists I \, \Forall i \ele I ~ \Exists `k_i,`r_i, `k'_i \Pred[ \derLmuInt `G |- {\Cmd} : (`k_i\arrow `r_i)\prod `k_i | `a{:}`k'_i, `D \And \Inter_I`k'_i\arrow `r_i \seqD `d ] 
	\\
\derLmuInt `G |- [`a]M : `k | `D &\Iff& 
\Exists `k',I \, \Forall i \ele I ~ \Exists `d_i \Pred[ \derLmuInt `G |- M : `d_i | `D \And `D(`a) = `k'  \And \Inter_I `d_i\prod `k' \seqC `k ] 
 \end{array} \]
 \end{lem}

\begin{Proof}
The proof is standard. For illustrate we just develop in detail one of the cases. 

\begin{itemize}
\item [$\Leftarrow$:] assume that $\derLmuInt `G |- M : `d_i | `D $,  $`D(`a) =`k' \seqC `k_i$ and $\Inter_{i \ele I} `d_i\prod `k_i \seqC `k$.
	Then by $(\TCmd)$ we get $\derLmuInt `G |- [`a]M : `d_i \prod `k' | `D $ for all $i$; therefore by (possibly multiple inferences of) $(\inter)$ we obtain
	$\derLmuInt `G |- [`a]M : \Inter_i `d_i \prod `k' | `D$ and the thesis follows by $(\leq)$.

\item [$\Rightarrow$:] 
	Let $ \derLmuInt `G |- [`a]M : `k | `D $; to show: $ \Exists `k',I \, \Forall i \ele I ~ \Exists `d_i \Pred[ \derLmuInt `G |- M : `d_i | `D \And `a{:}`k' \ele `D \And \Inter_I `d_i\prod `k' \seqC `k ]  $.
The derivation can finish with the rules: $(\TCmd)$, $(\inter)$ or $(\seq)$.

 \begin{description}
 \item[$\TCmd$] 
So the derivation is shaped like:
 \[ \begin{array}{c}
\Inf	[\fun{`D }(`a)=`k'']	
	{ \InfBox{ \derLmu `G |- M : `d | `D }
	}{ \derLmu `G |- [`a]M : `d\prod `k'' | `D }
 \end{array} \]
with $`k = `d\prod `k''$; take $`k' = `k''$, $I = \Set{1}$, $`d_1 = `d$.

 \item[$\inter$] 
So the derivation is shaped like:
 \[ \begin{array}{c}
\Inf	{\InfBox{ \derLmu `G |- T : `k_1 | `D }
	\quad
	\InfBox{ \derLmu `G |- T : `k_2 | `D }
	}{ \derLmu `G |- T : `k_1\inter`k_2 | `D }
 \end{array} \]
By induction, there exist $`k'_1$, $`k'_2$, $I_1$, $I_2$, such that
 \[ \begin{array}{c}
\Forall i \ele I_1 ~ \Exists `d^1_i \Pred[ \derLmuInt `G |- M : `d^1_i | `D \And `D(`a) = `k'_1 \And \Inter_{I_1} `d_i\prod `k'_1 \seqC `k_1 ] ~ \And
	\\
\Forall i \ele I_2 ~ \Exists `d^2_i \Pred[ \derLmuInt `G |- M : `d^2_i | `D \And `D(`a) = `k'_2 \And \Inter_{I_2} `d^2_i\prod `k'_2 \seqC `k_2 ] 
 \end{array} \]
Then necessarily $`k'_1 = `k'_2$.
Take $`k' = `k'_1 = `k'_2$, $I = I_1 \union I_2$, then for all $`d_i$ with $i \ele I$ we have $ \derLmuInt `G |- M : `d_i | `D $, and $ \Inter_I `d_i\prod `k' = ( \Inter_{I_1} `d^1_i\prod `k'_1 ) \inter ( \Inter_{I_2} `d^2_i\prod `k'_2 ) \seqC `k_1 \inter `k_2 $.

 \item[$\seq$]
So there 
 \[ \begin{array}{c}
\Inf	{\InfBox{ \derLmu `G |- T : `k_1 | `D } \dquad \InfBox{ `k_1 \seq `k_2 }
	}{\derLmu `G |- T : `k_2 | `D }
 \end{array} \]
By induction, there exist $`k'$, $I$, such that
 \[ \begin{array}{c}
\Forall i \ele I ~ \Exists `d_i \Pred[ \derLmuInt `G |- M : `d_i | `D \And `D(`a) = `k' \And \Inter_I `d_i\prod `k' \seqC `k_1 ] 
 \end{array} \]
Notice that then also $ \Inter_I `d_i\prod `k' \seqC `k_2 $. 

 \end{description}

\end{itemize}
 \end{Proof}

%


 \subsection{Type interpretation and soundness} \label{subsec:typeInterpretation}

In this section we will formally define the type interpretation and thereby the interpretation of typing judgements. 
As anticipated above in the informal discussion of the system, the meaning of a type will be a subset of the domain of interpretation. 

In definitions and statements below we relate types to a $`l`m$-model ${\ModM}=(R,D,C)$, silently assuming that the language $\Lang_R$ includes a constant $`y_a$ for every $a\ele \Compact(R)$.

 \begin{defi}[Type interpretation] \label{def:typesInterpretation}
Let ${\ModM}=(R,D,C)$ be a $`l`m$-model.
For $A = R,D,C$ we define the interpretation $ \Sem{ \cdot}^{{\ModM},A} : \Lang_{A} \SemArrow \Power{A}$ (written $\Sem{ \cdot}^A$ when $\ModM$ is understood) as follows:

 \[ \begin{array}[t]{rcl}
\Sem{ `y_a}^R &=& \filt_R a = \Set{r \ele R \mid a \po r} \\ 
\Sem{ `d\prod `k }^C &=& \Sem{`d}^D \times \Sem{ `k }^C \\ 
\Sem{ `k\arrow `r}^D &=& \Set{d \ele D \mid \Forall k \ele \Sem{ `k }^C \Pred [ d \, k \ele \Sem{ `r}^R ]} \\ 
\Sem{ `y_a}^D = \Sem{ `w\arrow `y_a}^D &=& \Set{d \ele D \mid \Forall k \ele C \Pred [ d \, k \ele \Sem{ `y_a}^R ] }
 \end{array} \]
and
 \[ \begin{array}[t]{rcl}
\Sem{ `w}^A &=& A \\ 
\Sem{ `s_1 \inter`s_2}^A &=& \Sem{ `s_1}^A \cap \Sem{ `s_2}^A 
 \end{array} \]
 \end{defi}

 \begin{rem} \label{rem:typeconstinterp}
The last definition is a special case with respect to the natural adaptation of the intersection type interpretation as subsets of a $`l$-model, in that we fix the interpretation of the type constants $`y_a$. 
This is consistent with the approach of constructing types from the solution of the continuation domain equations, and is the intended interpretation throughout this paper. 
In particular, it implies that the language $\Lang_R$ depends on the chosen domain of results $R$, and that the interpretation of a type is always a principal filter of either $R$, $D$, or $C$, according to its kind, as is proven in the next lemma. 

This choice poses no limitations. 
If we postulate that there exist denumerably many constants $`y_0, `y_1, \ldots$ in $\Lang_R$, then we can generalise the definition of type interpretation in a straightforward way to $\Sem{`s}^A_{`h}$ (relative to the \emph{type environment} $`h$, a mapping of type constants such that $`h(`y_i) \subseteq R$ for all $i$) by defining $\Sem{ `y_i}^R_{`h} = `h(`y_i)$ as the base case of the inductive definition. 
Then the above definition is recovered by considering an arbitrary exhaustive enumeration of the compacts $a_0, a_1, \ldots = \Compact(R)$ (possibly with repetitions; this enumeration exists since $R$ is $`w$-algebraic) and defining the interpretation of type constants through $`h_0(`y_i) = \filt_R a_i$.
 \end{rem}

There exists a close relation between the interpretation of types and the maps $\TypeToComp_A$ (see Def.\skp\ref{def:typeInterpDandC}), that is made explicit in the following lemma.

 \begin{lem} \label{lem:typeInterpCompactCone}
 For $A=R,D,C$ and any $`s \ele \Lang_A$, we have that $\Sem{ `s}^A = \filt_A \TypeToComp_A(`s)$.
 \end{lem}

 \begin{proof}
 By induction over the structure of types and by cases on $A$. 

 \begin{description} \itemsep 2pt
 \item [$ `s \equiv `w $] 
Then $\TypeToComp_A(`w) = \bot$ and $\filt_A \bot = A = \Sem{ `w}^A$.

 \item [$ `s \equiv `s_1 \inter`s_2 $] 
By induction, $\Sem{ `s_i}^A = \filt_A \TypeToComp_A(`s_i)$ and $\TypeToComp_A(`s_1 \inter`s_2) = \TypeToComp_A(`s_1) \join \TypeToComp_A(`s_2)$ by Rem.\skp\ref{rem:ThetaRank}. 
We have $\filt_A(\TypeToComp_A(`s_1) \join \TypeToComp_A(`s_2)) = \filt_A\TypeToComp_A(`s_1) \cap \filt_A\TypeToComp_A(`s_2)$, and therefore $\Sem{ `s_1 \inter`s_2}^A$ $= \Sem{ `s_1}^A \cap \Sem{ `s_2}^A = \filt_A \TypeToComp_A(`s_1 \inter`s_2) $.
	
 \item [$ `s \equiv `y_a $] ~

 \begin{description} 
 \item [$A = R$]
This follows immediately from $\TypeToComp_R(`y_a) = a$ and $\Sem{ `y_a}^R = \filt_R a$.
	
 \item[$A = D$]
Then $\TypeToComp_D(`y_a) = \StepFun{\bot}{a}$; but for any $d \ele D = [C \To R]$, by definition of step functions, $\StepFun{\bot}{a} \po d$ if and only if $a \po d \ k$ for all $k \ele C$, that is if and only if $d \ k \ele \Sem{ `y_a}^R$ for the above; it follows that $\TypeToComp_D(`y_a) \po d$ if and only if $d \ele \Sem{`y_a}^D$ as desired.

 \end{description} 

 \item [$ `s \equiv `d\prod `k $] 
Then $\TypeToComp_C(`d\prod `k) = \Pair{\TypeToComp_D(`d)}{\TypeToComp_C(`k)}$ by Rem.\skp\ref{rem:ThetaRank}, and for any
$ \Pair{d}{k} \ele C = D\prod C$ we have:
 \[ \begin{array}{lcll}
\Pair{\TypeToComp_D(`d)}{\TypeToComp_C(`k)} \po \Pair{d}{k} 
	& \Iff &
\TypeToComp_D(`d) \po d \And \TypeToComp_C(`k) \po k 
		& (\textrm{by definition of order over } D\prod C ) \\
	& \Iff & 
d \ele \Sem{`d}^D \And k \ele \Sem{`k}^C 
		& (\textrm{by ind.}) \\
	& \Iff & 
\Pair{d}{k} \ele \Sem{`d}^D \times \Sem{ `k }^C 
		& (\textrm{by Def. \ref{def:typesInterpretation}}) \\
	&=& 
\Sem{ `d\prod `k }^C 

 \end{array} \] 
	
 \item [$ `s \equiv `k\arrow `r $] 
Then $\TypeToComp_D(`k\arrow `r) = \StepFun{\TypeToComp_C(`k)}{\TypeToComp_R(`r)}$ by Rem.\skp\ref{rem:ThetaRank}; for any $d \ele D = [C\To R]$ we have: 
 \[ \begin{array}{lcll}
\StepFun{\TypeToComp_C(`k)}{\TypeToComp_R(`r)} \po d 
	& \Iff &
\forall k \ele C \Pred[ \TypeToComp_C(`k) \po k \Then \TypeToComp_R(`r) \po d \ k ]
		& (\textrm{by definition of } \StepFun{`.}{`.} ) \\ 
	& \Iff & 
\forall k \ele C \Pred[ k \ele \Sem{`k}^C \Then d \ k \ele \Sem{`r}^R ]
		& (\textrm{by induction}) \\
	& \Iff & 
d \ele \Sem{ `k\arrow `r}^D 
	& (\textrm{by Def. \ref{def:typesInterpretation}}) 
 \end{array} \]
\arrayqed[-22pt] 
 \end{description}
 \end{proof}

 \begin{cor} \label{cor:leq} 
For $A = R,D,C$, if $`s, `t \ele \Lang_A$ then $`s \seqA `t \Iff \Sem{ `s }^A \subseteq \Sem{ `t }^A$.
 \end{cor}

 \begin{proof}
 $ \begin{array}[t]{rcl@{\quad}l}
`s \seqA `t 
	& \Iff & 
\TypeToComp_A(`s) \sqsupseteq \TypeToComp_A(`t) 
	& (\textrm{by Lem.\skp\ref{lem:KRtheory}\,(\ref{lem:KRtheory-2}) and 
	\ref{lem:Theta_AsurjOderRev}}) \\ 
	& \Iff & 
\filter \TypeToComp_A(`s) \subseteq \filter \TypeToComp_A(`t) 
	& (\textrm{by Lem.\skp\ref{lem:typeInterpCompactCone}}) \\ 
	& \Iff & 
\Sem{ `s }^A \subseteq \Sem{ `t }^A 
 \end{array} $ 
\arrayqed 
 \end{proof}

We will now define satisfiability for typing judgements with respect to a $`l`m$-model.

 \begin{defi}[Satisfiability] \label{def:judegmentSat}
Let ${\ModM} = (R,D,C)$ be a $`l`m$-model.
We define semantic satisfiability through:
 \[ \def\arraystretch{1.2} \begin{array}{lcl}
 \modelsM {\env} |= `G;`D 
	& \Iff& 
 \Forall x \Pred[ \env \psk x \ele \Sem{`G(x)}_{\ModM}^{\mathcal D } ] \And	
 \Forall `a \Pred[ \env \psk `a \ele \Sem{ `D(`a) }_{\ModM}^{\mathcal C} ] 
 \\ 
 \modelsM `G |= M : `d | `D 
	& \Iff &	
 \Forall \env \Pred[ \modelsM {\env} |= `G;`D 
	\Then \Sem{ M }_{\ModM}^{\mathcal D }{ \env } \ele \Sem{`d}_{\ModM}^{\mathcal D } ] 
 \\ 
 \modelsM `G |= {\Cmd} : `k | `D 
	& \Iff &
 \Forall \env \Pred[ \modelsM {\env} |= `G;`D 
	\Then \Sem{ \Cmd}_{\ModM}^{\mathcal C }{ \env } \ele \Sem{ `k }_{\ModM}^{\mathcal C} ] 
 \end{array} \]
We will write $\models$ for $\models_{\ModM}$ when $\ModM$ is understood.
 \end{defi}

 \begin{rem} \label{rem:validity}
Continuing the discussion in Rem.\skp\ref{rem:typeconstinterp}, we note that we do not consider here the concept of \emph{validity}, namely satisfiability with respect to any $`l`m$-model $\ModM$, since we model both the language $\Lang_R$ and the pre-order $\seq_R$ \emph{after} $R$, which is the particular domain of results of $\ModM$. 

As a matter of fact, we could define validity as follows: first we fix the type theory $\TypeTheor_R$ for the language $\Lang_R$ with denumerably many type constants $`y_i$; then the satisfiability notion should be relativised to \emph{both} term \emph{and} type environments $e$ and $`h$, asking that the latter is a model of the theory $\TypeTheor_R$, in the sense that whenever $`r \seq_R `r'$, it holds that $`h(`r) \subseteq `h(`r')$.

However, we will not consider such a general formulation, as it would involve an unnecessary complication of the theory developed here.
 \end{rem}

The next result states the soundness of the typing system. 
Note that, although the construction of the system has been made by having an initial model in mind, the soundness theorem holds for any model.

 \begin{thm}[Soundness of type assignment] \label{thm:typeAssSound}
Let $\ModM$ be a $`l`m$-model. 
If $ \derLmu `G |- T : `s | `D $, then $ \modelsLmu `G |= T : `s | `D $.
 \end{thm}

 \begin{proof}
By induction on the structure of derivations.
(We will drop the super and subscripts on the interpretation function.)

 \begin{description} \itemsep 3pt

 \item [$ \Axiom $] 
Then $T \equiv x$ and $`s = `d$; let $e \models `G,x{:}`d;`D $, then $e\psk x \ele \Sem{`d}$. 
Hence, by Def.\skp\ref{def:interpretation}, we get $ \Sem{ x }{ e } \ele \Sem{`d}$.

 \item [$ \LAbs $] 
Then $T \equiv `lx.N$ and there exist $`d, `k$ and $`r$ such that $ `s = `d\prod `k\arrow `r$ and $ \derLmu `G,x{:}`d |- N : `k\arrow `r | `D $.
By definition, $ \Sem{ `lx.N }{ e }{ k } = \Sem{ N }{ e[x\to d] }{ k' } $, where $k = \Pair<d,k'>$; also, $e \models `G;`D $ and $ d \ele \Sem{`d} $ if and only if $ e[x \to d] \models`G,x{:}`d;`D $.
By induction, for any $e \models `G;`D $ and $d \ele \Sem{`d}$, $ \modelsLmu `G,x{:}`d |= N : `k\arrow `r | `D $, so $ \Sem{ N }{ e[x\to d] } \ele \Sem{ `k\arrow `r }$, so $ \Sem{ N }{ e[x\to d] }{ k } \ele \Sem{ `r }$ for any $k \ele \Sem{ `k }$. 
So $ \Sem{ `lx.N }{ e }{ \Pair<d,k> } \ele \Sem{ `r }$, so $ \Sem{ `lx.N }{ e } \ele \Sem{ `d\prod `k\arrow `r}$, and we conclude $ \modelsLmu `G |= `lx.N : `d\prod `k\arrow `r | `D $.

 \item [$ \App $] 
Then $T \equiv PQ$ and there exist $`d, `k$ and $`r$ such that $ `s = `k\arrow `r$, $ \derLmu `G |- P : `d\prod `k\arrow `r | `D $ and $ \derLmu `G |- Q : `d | `D $.
By definition, $ \Sem{ PQ }{ e }{ k } = \Sem{ P }{ e }{ \Pair< \Sem{ Q }{ e } , k > } $.
Let $e \models `G;`D $.
By induction, $ \modelsLmu `G |= P : `d\prod `k\arrow `r | `D $ and $ \modelsLmu `G |= Q : `d | `D $, so $ \Sem{ P }{ e } \ele \Sem{ `d\prod `k\arrow `r }$ and $ \Sem{ Q }{ e } \ele \Sem{`d}$; in particular, for any $ \Pair<d,k'> \ele \Sem{`d}\prod \Sem{ `k } = \Sem{ `d\prod `k }$, we have $ \Sem{ P }{ e }{ \Pair<d,k'> } \ele \Sem{ `r }^R$, so $ \Sem{ PQ }{ e } \ele \Sem{ `k\arrow `r}$, and thereby $ \modelsLmu `G |= PQ : `k\arrow `r | `D $.

 \item [$ \TCmd $] 
Then $T \equiv [`a]N$ and there exist $`d$ and $`k = `D(`a)$ such that $`s = `d\prod `k$ and $ \derLmu `G |- N : `d | `D $.
By induction we have that $ \modelsLmu `G |= N : `d | `D $, so that for any $e \models `G;`D $ we have that $\Sem{N}\,e \ele \Sem{`d}$. 
But $e \models `G;`D $ implies that $e\,`a \ele \Sem{`D(`a)} = \Sem{`k}$, hence 
$\Sem{[`a]N}\, e = \Pair{\Sem{N}\,e}{e\,`a} \ele \Sem{`d}\prod\Sem{`k} =\Sem{`d\prod`k}$ as desired. 
Then $ \modelsLmu `G |= [`a]N : `d\prod`k | `D $ by the arbitrary choice of $e$.

 \item [$ \MAbs $] 
Then $T \equiv `m`a.\Cmd$, and there exist $`k, `k'$ and $`r$ such that $`s = `k\arrow `r $ and $ \derLmu `G |- {\Cmd} : (`k'\arrow `r)\prod `k' | `a{:}`k,`D $.
By definition, $ \Sem{ `m`a.\Cmd }{ e }{ k } = d\,k'$, where $\Pair<d,k'> = \Sem{ \Cmd }{ e[`a\to k] } $.
Let $e \models `G;`D $, and $k \ele \Sem{ `k }$, then $e[`a \to k] \models `G;`a{:}`k,`D $. 
By induction, $ \modelsLmu `G |= {\Cmd} : (`k'\arrow `r)\prod `k' | `a{:}`k,`D $, so $ \Sem{ \Cmd }{ e[`a\to k] } \ele \Sem{ (`k'\arrow `r)\prod `k' }$. 
Let $ \Sem{ \Cmd }{ e[`a\to k] } = p $, then $ \pi_1\, p \ele \Sem{ `k'\arrow `r }$ and $ \pi_2 \, p \ele \Sem{ `k' }$, and $ \pi_1 \, p \, (\pi_2 \, P ) \ele \Sem{ `r }$, so $ \Sem{ `m`a.\Cmd }{ e }{ k } \ele \Sem{ `r }$, for any $k \ele \Sem{ `k }$, so $ \Sem{ `m `a.\Cmd }{ e } \ele \Sem{ `k\arrow `r}$, and therefore $ \modelsLmu `G |= `m`a.{\Cmd} : `k\arrow `r | `a{:}`k,`D $.

 \item [$ \inter $] 
By induction and the interpretation of an intersection type.

 \item [$ `w $] 
Immediate by the definition of interpretation of $`w$.

 \item [$ \seq $] 
By induction and Cor.\skp\ref{cor:leq}.
 \qed

 \end{description}
 \end{proof}

We will now show that we can interpret a term or command by the set of types that can be given to it (Thm.\skp\ref{thm:filter-term interpretation}).
Towards that result, we first make the denotation of terms and commands and the interpretations of types in the filter model $\Filt$ explicit.

 \begin{lem} [] \label{lem:filtSem}
The following equations hold:
 \[ \begin{array}{r@{\Skip}c@{\Skip}l}
\SemF { `l x.M}{ e } 
	&=& 
	\filt_D\, \Set{\Inter_{i \ele I}(`d_i\prod `k_i\arrow `r_i) \ele \Lang_D \mid \Forall i \ele I \Pred [{ `k_i\arrow `r_i \ele \SemF { M }{ e[x\to (\filt_D `d_i)] } }] } 
	\\ [1mm]
\SemF { MN}{ e } &=& 
	\filt_D \Set{\Inter_{i \ele I}`k_i\arrow `r_i \ele \Lang_D\mid	\Forall i \ele I \ \Exists `d_i \ele \SemF { N }{ e }, `k_i \ele \Lang_C \Pred[ `d_i\prod `k_i\arrow `r_i \ele \SemF { M }{ e } ] }
	\\ [1mm]
\SemF { `m`a.\Cmd}{ e } &=& 
	\filt_D \Set{\Inter_{i \ele I}`k_i\arrow `r_i \ele \Lang_D\mid \Forall i \ele I \ \Exists `k'_i \ele \Lang_C \Pred[{ (`k'_i\arrow `r_i)\prod `k'_i \ele \SemF { \Cmd}{ e [`a\to (\filt_C `k_i)] } }] }
	\\ [1mm]
\SemF { [`a]M}{ e } &=& 
	\filt_C \Set{\Inter_{i \ele I}`d_i\prod `k_i \ele \Lang_C \mid \Forall i \ele I \Pred[ `d_i \ele \SemF { M }{ e } \And `k_i \ele e(`a) ] }
 \end{array} \]
 \end{lem}

 \begin{Proof}
By unravelling definitions. For example:
 \[ \begin{array}[t]{rcl@{\quad}l}
\Sem{ `l x.M}^{\Filt^D }{ e }
	& = & \Semlambda d\cons k \ele \Filt_C \, . \, \SemF { M }{ e[x\to d] }{ k } \\
	& = & \bigcup \Set{ \join_{i \ele I}\StepFun{\filt_C `d_i\prod `k_i }{ \filt_R\psk `r_i} \mid \filt_R\psk `r_i \subseteq \SemF { M }{ e[x\to \filt_D `d_i] }{ (\filt_C `k_i) } } \\
	& = & \filt_D \Set{ \Inter_{i \ele I}(`d_i\prod `k_i\arrow `r_i) \ele \Lang_D \mid \Forall i \ele I \Pred[{ `k_i\arrow `r_i \ele \SemF { M }{ e[x\to (\filt_D `d_i)] } }] }
	\end{array} \]
using the fact that $ \Semlambda d\cons k \ele \Filt_C \, . \, \SemF { M }{ e[x\to d] }{ k }$ is continuous, hence it is the \emph{sup} of finite joins of the step functions 
$\StepFun{\filt_C `d_i\prod `k_i }{ \filt_R\psk `r_i} $.
Observe that $d\cons k = \filt_C \Set{\Inter_{i \ele I}`d_i\prod `k_i \mid \Forall i \ele I \Pred[ `d_i \ele d \And `k_i \ele k ] }$, that the set $\Set{\join_{i \ele I}\StepFun{\filt_C `d_i\prod `k_i }{ \filt_R\psk `r_i} \mid \ldots}$ is directed and hence its join is its union, and finally that $\filt_R\psk `r_i \subseteq \SemF { M }{ e[x\to \filt_D `d_i] }{ (\filt_C `k_i) }$ if and only if $`k_i\arrow `r_i \ele \SemF { M }{ e[x\to (\filt_D `d_i) } ]$.
 \end{Proof}

 \begin{lem} \label{lem:filtSemTypes}
For $A = R,D,C$: if $`s\ele \Lang_A$ then, up to the isomorphisms $\Filt_R\Isom R$, $\Filt_D\prod\Filt_C \Isom \Filt_C$ and
$[\Filt_C\To\Filt_R]\Isom \Filt_D$,
we have:
 \begin{enumerate}
 \item \label{lem:filtSemTypes-1}
	$\TypeToComp_{\Filt_A}(`s) = \filt_A `s$,
 \item \label{lem:filtSemTypes-2}
	$\SemF {`s} = \Set{a \ele \Filt_A \mid `s \ele a}$.
 \end{enumerate}
 \end{lem}

 \begin{Proof} Recall that the isomorphism $\Filt_R\Isom R$ is established by Prop.\skp\ref{prop:filter-iso} and described in detail in Rem.\skp\ref{rem:Filt_R}, and that $K:\Filt_D\prod\Filt_C \To \Filt_C$ and
$G:[\Filt_C\To\Filt_R]\To \Filt_D$ the are isomorphisms of Def.\skp\ref{def:iso-maps} and Thm.\skp\ref{iso-theorem}. Now
to prove part (\ref{lem:filtSemTypes-1}) we proceed by induction over the structure of types. 
 \begin{description} \itemsep 2pt

 \item [$ `s \equiv `w $] 
Then $\TypeToComp_{\Filt_A}(`w) = \bot_{\Filt_A} = \filt_A`w$.

 \item [$ `s \equiv `s\inter`t $] 
%
 $ \begin{array}[t]{rcl@{\quad}l}
\TypeToComp_{\Filt_A}(`s\inter`t) & = & \TypeToComp_{\Filt_A}(`s) \join \TypeToComp_{\Filt_A}(`t) & 
(\textrm{by Def.\skp\ref{def:typesToComp} and Rem.\skp\ref{rem:ThetaRank}}) \\ 
	& = & \filt_A `s \join \filt_A`t & (\textrm{by induction}) \\ 
	& = & \filt_A(`s\inter`t)
 \end{array} $

 \item [$ `s \equiv `y_a $] 
By Rem.\skp\ref{rem:Filt_R}, under the isomorphism $\Filt_R\Isom R$ the compact point $a\ele \Compact(R)$ is the image of $\filt_R`y_a$; hence, up to isomorphism, we have that $\TypeToComp_{\Filt_R}(`y_a)= a =\filt_R`y_a$.

 \item [$ `s \equiv `d\prod`k $] 
%
 $ \begin{array}[t]{rcl@{\quad}l}
 \TypeToComp_{\Filt_C}(`d\prod`k) & = & \Pair{\TypeToComp_{\Filt_D }(`d)}{\TypeToComp_{\Filt_C}(`k)} & 
	(\textrm{by Def.\skp\ref{def:typesToComp} and Rem.\skp\ref{rem:ThetaRank}}) \\ 
	& = & \Pair{\filt_D`d}{\filt_C`k} & (\textrm{by induction}) \\ 
	& = & (\filt_D`d) \cons \ (\filt_C`k) & (\textrm{up to the iso } K ) \\ 
	& = & \filt_C(`d\prod`k) & (\textrm{by Rem.\skp\ref{rem:GK}}) 
 \end{array} $

 \item [$ `s \equiv `k\arrow `r $] 
%
 $ \begin{array}[t]{rcl@{\quad}l}
 \TypeToComp_{\Filt_D }(`k\arrow `r) & = & \StepFun{\TypeToComp_{\Filt_C}(`k)}{\TypeToComp_{\Filt_R}(`r)} &
	(\textrm{by Def.\skp\ref{def:typesToComp} and Rem.\skp\ref{rem:ThetaRank}}) \\ 
	& = & \StepFun{\filt_C`k}{\filt_R`r} & (\textrm{by induction}) \\ 
	& = & \filt_D(`k\arrow `r) & (\textrm{up to the iso $G$ and by Rem.\skp\ref{rem:GK}})
 \end{array} $
	
 \end{description}

To prove part (\ref{lem:filtSemTypes-2}), observe that for any $`s\ele \Lang_A$ and $a \ele \Filt_A$, $\uparrow_A`s = \Set{`t \ele \Lang_A \mid `s\seqA`t} \subseteq a$ if and only $`s\ele a$. 
Now by Lem.\skp\ref{lem:typeInterpCompactCone} we know that $\SemF {`s} = \filt_{\Filt_A} \TypeToComp_{\Filt_A}(`s)$, which, by part (\ref{lem:filtSemTypes-1}) of this
lemma, implies $\SemF {`s} = \filt_{\Filt_A}(\filt_A `s) = \Set{a \ele \Filt_A \mid \filt_A `s \subseteq a}$. 
By the remark above we have that $\Set{a \ele \Filt_A \mid \filt_A `s \subseteq a} = \Set{a \ele \Filt_A \mid `s \ele a}$. 
 \end{Proof}

The next theorem, together with the Completeness Theorem (Thm.\skp\ref{thm:typeAssCompl}) that it implies, is the main result of this section, which states that the set of types that are assigned to terms and commands by the type assignment system coincides with their interpretation in the filter model.
Its proof essentially depends on Lem.\skp\ref{mu gen Lemma} and Lem.\skp\ref{lem:filtSem}.

 \begin{thm} \label{thm:filter-term interpretation}
Let $A = D,C$.
Given an environment $e \ele \EnvFilt$, then
 \[ \begin{array}{rcl}
	\SemF { T}{ e } &=& \Set{`s \ele \Lang_A \mid \Exists`G,`D \Pred [{ \modelsF e |= `G;`D \And \derLmu `G |- T : `s | `D }] }.
 \end{array} \]
 \end{thm}

 \begin{proof} Because of the logical rules, the set $ \Set{`s \ele \Lang_A \mid \Exists`G,`D \Pred [ \modelsF e |= `G;`D \And \derLmu `G |- T : `s | `D ] }$ is a filter in $\Filt_A$, for $A=D,C$. 
To prove that this filter coincides with $ \SemF { T}{ e }$ we proceed by induction over the structure of terms.

 \begin{description} \itemsep 2pt

 \item [$ T\equiv x $] 
Then $ \SemF { x}{ e } = e\,x$.
By definition
, $\modelsF e |= `G;\emptyset \Iff e\,x \ele \SemF {`G(x)}$. 
By Lem.\skp\ref{lem:filtSemTypes}\,(\ref{lem:filtSemTypes-2}) we know that $ \SemF {`G(x)} = \Set{d\ele \Filt_D \mid `G(x)\ele d}$, so that $\modelsF e |= `G;\emptyset$ is equivalent to $`G(x)\ele e\,x$.
On the other hand, by Lem.\skp\ref{mu gen Lemma}
, we have that $\derLmu `G |- x : `d | {} $ if and only if $`G(x)\seqD `d$, hence $\derLmu `G |- x : `d | {} $ if and only if $`d\ele e\,x$.
	
 \item [$ T\equiv `l x.M $] 
For $`d \ele \SemF { `l x.M}{ e }$, by Lem.\skp\ref{lem:filtSem} 
there exist $I$ such and $`d_i$, $`k_i$, and $`r_i$ ($i \ele I$) such that $`k_i\arrow `r_i \ele \SemF { M }{ e[x\to (\filt_D `d_i)] }$. 
By induction there exist $`G_i,`D _i ~ (i \ele I)$ such that $ \modelsF e[x\to (\filt_D `d_i)] |= `G_i;`D_i $ and $\derLmu `G_i |- M : `k_i\arrow `r_i | `D_i $. 
Let $`d'_i = `G_i(x)$: then by rule $(\LAbs)$ we have $\derLmu `G_i\except x |- `l x.M : `d'_i\prod `k_i\arrow `r_i | `D_i $. 

On the other hand, from $ \modelsF e[x\to (\filt_D `d_i)] |= `G_i;`D_i $ we know that $`G_i(x) = `d'_i $; since $`d'_i \ele \filt_D `d_i$, also $`d_i\seqD `d'_i$.
Then $`d_i\prod `k_i\seqC `d_i'\prod `k_i$ follows by the co-variance of $\prod $, and 
$`d'_i\prod `k_i\arrow `r_i \seqD `d_i\prod `k_i\arrow `r_i$ by the contra-variance of the arrow in its first argument. 
Hence, by applying rule $(\seq)$, for all $i \ele I$, we obtain $ \derLmu `G_i\except x |- `l x.M : `d_i\prod `k_i\arrow `r_i | `D_i $. 

Take $`G = \Inter_{i \ele I}`G_i\except x$ and $`D = \Inter_{i \ele I}`D_i$: then $ \derLmu `G |- `l x.M : `d_i\prod `k_i\arrow `r_i | `D $ for all $i \ele I$ by applying rule $(\Weak)$, and therefore $ \derLmu `G |- `l x.M : \Inter_{i \ele I}`d_i\prod `k_i\arrow `r_i | `D $ by applying rule $(\inter)$ and $ \derLmu `G |- `l x.M : `d | `D $ by rule $(\seq)$. 
Observe that, for all $i \ele I$, $x\notele \dom(`G_i\except x)$ and consequently $x\notele \dom(`G)$, so that, for all $i \ele I$, $ \modelsF e[x\to (\filt_D `d_i)] |= `G_i;`D_i $ implies $ \modelsF e |= `G_i\except x;`D_i $ and so $ \modelsF e |= `G;`D $ as required.
	
Vice-versa, if $ \derLmu `G |- `lx.M : `d | `D $ and $ \modelsF e |= `G;`D $ then, by Lem.\skp\ref{mu gen Lemma}
, there exist 
$`d_i$, $`k_i$, and $`r_i$ ($i \ele I$) such that $ \derLmu `G,x{:}`d_i |- M : `k_i\arrow `r_i | `D $ and $\Inter_{i \ele I} `d_i\prod `k_i\arrow `r_i \seqD `d$.
Observe that $\modelsF e |= `G;`D $ implies $\modelsF e[x\to (\filt_D `d_i)] |= `G, x{:}`d_i;`D $ for all $i \ele I$.
By induction, $`k_i\arrow `r_i \ele \SemF { M }{ e[x\to (\filt_D `d_i)] }$, so by Lem.\skp\ref{lem:filtSem} 
 $\Inter_{i \ele I}(`d_i\prod `k_i\arrow `r_i) \ele \SemF { `l x.M}{ e } $. 
Now $\SemF { `l x.M}{ e }$ is a filter and since $\Inter_{i \ele I} `d_i\prod `k_i\arrow `r_i \seqD `d$, we conclude that $`d \ele \SemF { `l x.M}{ e }$. 

 \item [$ T\equiv MN $] 
If $`d \ele \SemF { MN}{ e }$ then, by Lem.\skp\ref{lem:filtSem}
, there exist $I$ and $`k_i$, $`r_i$ and $`d_i \ele \SemF { N }{ e }$ ($i \ele I$) such that $`d_i\prod `k_i\arrow `r_i \ele \SemF { M }{ e }$ and $\Inter_{i \ele I}`k_i\arrow `r_i\seqD `d$. 
By induction, for all $i \ele I$ there exist $`G_{i,j},`D _{i,j}$ for $j=1,2$, such that $\modelsF e |= `G_{i,j},`D _{i,j} $ and $ \derLmu `G_{i,1} |- M : `d_i\prod `k_i\arrow `r_i | `D_{i,1} $ and $ \derLmu `G_{i,2} |- N : `d_i | `D_{i,2} $. 
Take $`G = \Inter_{i,j}`G_{i,j}$ and $`D = \Inter_{i,j}`D_{i,j}$, then $`G\seqD `G_{i,j}$ and $`D\seqC `D_{i,j}$ so that $\modelsF e |= `G;`D $.
By applying rule $(\Strength)$, for all $i \ele I$ we have $ \derLmu `G |- M : `d_i\prod `k_i\arrow `r_i | `D $ and $ \derLmu `G |- N : `d_i | `D $.
By applying rule $(\App)$ we have $ \derLmu `G |- MN : `k_i\arrow `r_i | `D $; we obtain $ \derLmu `G |- MN : \Inter_{i \ele I}`k_i\arrow `r_i | `D $ by applying rule $(\inter)$, and $ \derLmu `G |- MN : `d | `D $ by applying rule $(\seq)$.
	
Vice-versa, assume $ \derLmu `G |- MN : `d | `D $ and $\modelsF e |= `G;`D $. 
By Lem.\skp\ref{mu gen Lemma}
, there exist $I$ and $`d_i$, $`k_i$, and $`r_i$ ($i \ele I$) such that $ \derLmu `G |- M : `d_i\prod k_i\arrow `r_i | `D $ and $ \derLmu `G |- N : `d_i | `D $ and $\Inter_{i \ele I}`k_i\arrow `r_i \seqD `d$.
For all $i \ele I$, by induction we have $`d_i\prod k_i\arrow `r_i \ele \SemF { M }{ e }$ and $`d_i \ele \SemF { N }{ e }$, and so, by Lem.\skp\ref{lem:filtSem}
, $k_i\arrow `r_i \ele \SemF { MN }{ e }$. 
Since $\SemF { MN }{ e }$ is a filter, we have $\Inter_{i \ele I}`k_i\arrow `r_i \ele \SemF { MN }{ e }$ and therefore $`d \ele \SemF { MN }{ e }$.

 \item [$ T\equiv `m`a.\Cmd $] 
If $`d \ele \SemF { `m`a.\Cmd}{ e }$ then, by Lem.\skp\ref{lem:filtSem}
, there exist $I$ and $`k_i$, $`r_i$, and $`k'_i$ ($i \ele I$) such that $(`k'_i\arrow `r_i)\prod `k'_i \ele \SemF { \Cmd}{ e[`a\to (\filt_C `k_i)] }$. 
For all $i \ele I$, by induction there exist $`G_i$ and $`D_i$ 
such that $\modelsF e[`a\to (\filt_C `k_i)] |= `G_i;`D_i$ and $ \derLmu `G_i |- {\Cmd} : (`k'_i\arrow `r_i)\prod `k'_i | `D_i $. 
Let $`G=\bigwedge_{i\ele I}`G_i$ and, for all $i \ele I$, $`k_i = `D_i(`a)$ and $`D'_i = `D_i\except`a$, then $ \derLmu `G |- {\Cmd} : (`k'_i\arrow `r_i)\prod `k'_i | `a{:}`k_i,`D'_i $ so that $ \derLmu `G |- `m`a.{\Cmd} : `k_i\arrow `r_i | `D'_i $. 
Take $`D = \Inter_{i \ele I} `D'_i$; then $`D\seqC `D'_i$ for all $i \ele I$, so by applying rule $(\Strength)$ we obtain $ \derLmu `G |- {`m`a.\Cmd} : `k_i\arrow `r_i | `D $, from which we derive $ \derLmu `G |- {`m`a.\Cmd} : `d | `D $ by applying rules $(\inter)$ and $(\seq)$.
On the other hand, for all $i \ele I$, since $`a\notele \dom(`D'_i)$ we have $\modelsF e[`a\to (\filt_C `k_i)] |= `G;`D_i$ which implies $\modelsF e |= `G;`D'_i$, so that $`D(`a) = `w$. We conclude that $\modelsF e |= `G;`D $, as desired.
	
Vice-versa, assume $ \derLmu `G |- {`m`a.\Cmd} : `d | `D $ and $\modelsF e |= `G;`D $.
Then by Lem.\skp\ref{mu gen Lemma}
, there exists $I$ and $`k_i$, $`r_i$, and $`k'_i$ ($i \ele I$) such that $ \derLmu `G |- {\Cmd} : (`k_i\arrow `r_i)\prod `k_i | `a{:}`k'_i,`D $, and $\Inter_{i \ele I}`k'_i\arrow `r_i \seqD `d$. 
But if $\modelsF e |= `G;`D $ then $\modelsF e[`a\to (\filt_C `k_i)] |= `G;`a{:}`k'_i,`D $; 
then by induction, for all $i \ele I$, $(`k_i\arrow `r_i)\prod `k_i \ele \SemF { \Cmd}{ e[`a\to (\filt_C `k_i)] }$. 
Since $\Inter_{i \ele I}`k'_i\arrow `r_i \seqD `d$, by Lem.\skp\ref{lem:filtSem}
, we conclude that $`d \ele \SemF { `m`a.\Cmd}{ e }$.
	
 \item [$ {T\equiv [`a]M} $] 
If $`k \ele \SemF { [`a]M}{ e }$ then by Lem.\skp\ref{lem:filtSem}
, there exist $I$ and $`d_i$, $`k_i$, and $`d_i \ele \SemF { M }{ e }$ ($i \ele I$) such that $`k_i \ele e(`a)$ and $\Inter_{i \ele I}`d_i\prod `k_i \seqC `k$.
For all $i \ele I$, by induction there exist $`G_i,`D _i$ 
such that $\modelsF e |= `G_i;`D_i$ and $`D_i(`a) \seqC `k_i$ and $ \derLmu `G_i |- M : `d_i | `D_i $. Let $`G=\bigwedge_{i\ele I}`G_i$.
Then, for all $i \ele I$, $ \derLmu `G |- [`a]M : `d_i\prod `k_i | `D_i $; take $`D = \Inter_{i \ele I}`D_i$, then for all $i \ele I$, by applying rule $(\Strength)$, we obtain $ \derLmu `G |- [`a]M : `d_i\prod `k_i | `D $. 
We obtain $ \derLmu `G |- [`a]M : \Inter_{i \ele I}`d_i\prod `k_i | `D $ by applying rule $(\inter)$ and then $ \derLmu `G |- [`a]M : `k | `D $ by applying rule $(\seq)$. 
Since $\modelsF e |= `G;`D_i$ and $`D \seqC `D_i$ for all $i \ele I$, we conclude that $\modelsF e |= `G;`D $.
	
Vice-versa, if $ \derLmu `G |- [`a]M : `k | `D $ and $\modelsF e |= `G;`D $ then, by Lem.\skp\ref{mu gen Lemma}
, there exists $I$ and $`d_i$ and $`k_i$ ($i \ele I$) such that $ \derLmu `G |- M : `d_i | `D $ and $`D(`a) \seqC `k_i $, and $\Inter_{i \ele I} `d_i\prod `k_i \seqC `k$.
By induction $`d_i \ele \SemF { M }{ e }$; from $\modelsF e |= `G;`D $ we have that $`k_i \ele e(`a)$ for all $i \ele I$.
Then $\Inter_{i \ele I}`d_i\prod `k_i \ele \SemF { [`a]M}{ e }$ and therefore that $`k \ele \SemF { [`a]M}{ e }$ since the last set is a filter.
\qed

 \end{description}
 \end{proof}

 \begin{defi} \label{def:e_GD }
Given a basis $`G$ and a context $`D $, we define the environment $e_{`G;`D } \ele \EnvFilt$ by: 
 \[ \begin{array}{rcl}
e_{`G;`D }(x) &\ByDef& \filt_D `G(x) \\
e_{`G;`D }(`a) &\ByDef& \filt_D `D(`a)
 \end{array} \]
 \end{defi}

Because of the definition of $`G(x)$ and $`D(`a)$ for $x\ele \TVar$ and $`a\ele \CVar$, $e_{`G;`D } \ele \EnvFilt$ implies:
 \[ \begin{array}{rcl}
e_{`G;`D }(x) &=& 
 \begin{cases}{ll}
	\filt_D `d & (\textrm{if }x{:}`d \ele `G) \\
	\filt_D `w & (\textrm{otherwise})
 \end{cases}
\\ [4mm]
e_{`G;`D }(`a) &=& 
 \begin{cases}{ll}
	\filt_C `k & (\textrm{if }`a{:}`k \ele `D) \\
	\filt_C `w & (\textrm{otherwise})
 \end{cases}
 \end{array} \]

For this environment, we can show:

 \begin{lem} \label{lem:EnvStrength}
If $ \modelsF e_{`G;`D } |= `G';`D' $, then $ `G\seqD`G' $ and $ `D\seqC`D'$.
 \end{lem}

 \begin{Proof}
By definition, if $\modelsF e_{`G;`D } |= `G',`D' $ then:
 \[ \begin{array}{rclcl}
e_{`G;`D }(x) &=& \filt_D `G(x) &\ele& \SemF { `G'(x) } \\
e_{`G;`D }(`a) &=& \filt_C `D(`a) &\ele& \SemF { `D'(`a) }
 \end{array} \] 
for all $x\ele \TVar$ and $`a\ele \CVar$.
From Lem.\skp\ref{lem:filtSemTypes}\,(\ref{lem:filtSemTypes-2}) we know that
$\filt_D `G(x) \ele \SemF { `G'(x) }$ if and only if $`G'(x) \ele \filt_D `G(x)$, that is $`G(x) \seqD `G'(x)$. 
Similarly $\filt_C `D(`a) \ele \SemF { `D'(`a) }$ if and only if $`D(`a)\seqC`D'(`a)$. Hence $`G\seqD`G'$ and $`D\seqC`D'$.
 \end{Proof}

This last result implies that every type which is an element of the interpretation of a term is derivable for that term:

 \begin{lem} \label{aux-lemma}
If $ `d \ele \SemF { M }{e_{`G;`D }} $, then $ \derLmu `G |- M : `d | `D $.
 \end{lem}
 \begin{proof}
$ \def \arraystretch{1.1} \begin{array}[t]{lcll}
`d \ele \SemF { M }{ e_{`G;`D } } 
	& \Then & 
 \Exists`G',`D' \Pred [ e_{`G;`D } \models `G';`D' \And \derLmu `G' |- M : `d | `D' ] 
	& (\textrm{Thm.\skp\ref{thm:filter-term interpretation}}) \\
	& \Then & 
 \Exists`G',`D' \Pred [ `G\seqD `G' \And `D\seqC `D' \And \derLmu `G' |- M : `d | `D' ] 
	& (\textrm{Lem.\skp\ref{lem:EnvStrength}}) \\
	& \Then & 
 \derLmu `G |- M : `d | `D 
	& \kern-7.5mm (\textrm{rule $(\Strength)$, Lem.\skp\ref{lem:strengthening}}) 
 \end{array} $
 \arrayqed
 \end{proof}

This lemma would not hold without case $(\TCmd_2)$ of rule $(\TCmd)$. As explained in Rem.\skp\ref{rem:TCmd-rule} we should require that $`D\seqC `D'$ implies $\dom(`D) \supseteq \dom(`D')$, which is not always the case.

We can now prove the completeness theorem for our type assignment system.

 \begin{thm}[Completeness] \label{thm:typeAssCompl}
Let ${\ModM} = (R,D,C)$ be a $`l`m$-model. If $ \modelsM `G |= M : `d | `D $, then $ \derLmu `G |- M : `d | `D $.
 \end{thm}
	
 \begin{proof}
Let $ \modelsM `G |= M : `d | `D $: since ${\ModM}$ is isomorphic to the filter model $\Filt = (\Filt_R,\Filt_D,\Filt_C)$,
we have that $\modelsF `G |= M : `d | `D $. By construction, $\modelsF e_{`G;`D } |= `G;`D $, and therefore:
 \[ \def \arraystretch{1.1} \begin{array}{lcll}
 \modelsF `G |= M : `d | `D 
	& \Then & 
 \SemF { M }{ e_{`G;`D } } \ele \SemF {`d} & (\textrm{Def.\skp\ref{def:judegmentSat}}) \\
	& \Then & 
`d \ele \SemF { M }{ e_{`G;`D } } & (\textrm{Lem.\skp\ref{lem:filtSemTypes}\,(\ref{lem:filtSemTypes-2})}) \\
	& \Then &
 \derLmu `G |- M : `d | `D & 
	(\textrm{Lem.\skp\ref{aux-lemma}}) 
 \end{array} \]
\arrayqed[-22pt]

 \end{proof}

 \section{Closure under conversion} \label{sec:closureUnderConv} 

In this section, we will show that our notion of type assignment is closed under conversion, \ie~is closed both under subject reduction and expansion.
We will first show that this follows from the semantical results we have established in the previous section; then we show the same result via a syntactical proof. 
The latter is more informative about the structure of derivations in our system; also we establish the term substitution and, more importantly, the structural substitution lemmas (Lem.\skp\ref{lem:substitution}\skp(\ref{lem:substitution-term}) and Lem.\skp\ref{structural substitution lemma} respectively).

We begin with the abstract proof, which crucially depends on Lem.\skp\ref{aux-lemma} and hence on Thm.\skp\ref{thm:filter-term interpretation}.

 \begin{thm}[Closure under conversion] \label{thm:convInvariance}
Let $ \derLM M = N $. If $ \derLmu `G |- M : `d | `D $, then $ \derLmu `G |- N : `d | `D $.
 \end{thm}

 \begin{proof}
By Thm.\skp\ref{thm:sem_soundness}, if $\derLM M = N $ then $\Sem{ M }^{\ModM}{ e } = \Sem{ N }^{\ModM}{ e }$ for any 
model ${\ModM}$ and environment $e \ele \EnvM$, which holds in particular for $\Filt$ and $ e_{`G;`D }$. So 
 \[ \begin{array}[t]{rcl@{\quad}l}
 \derLmu `G |- M : `d | `D 
	& \Then &
 \modelsF `G |= M : `d | `D & (\textrm{Thm.\skp\ref{thm:typeAssSound}})
	\\ & \Then &
	\SemF { M }{ e_{`G;`D } } \ele \SemF {`d} & (\textrm{since } \modelsF e_{`G;`D } |= `G;`D )
	\\ & \Then &
`d \ele \SemF { M }{ e_{`G;`D } } & (\textrm{Lem.\skp\ref{lem:filtSemTypes}\,(\ref{lem:filtSemTypes-2})})
	\\ & \Then &
`d \ele \SemF {N}{ e_{`G;`D } } & (\textrm{Thm.\skp\ref{thm:sem_soundness}})
	\\ & \Then &
 \derLmu `G |- N : `d | `D & (\textrm{Lem.\skp\ref{aux-lemma}})
 \end{array} \]
\arrayqed[-22pt]
 \end{proof}

To illustrate the type assignment system itself, we will now show a more `operational' proof for the same property, by studying how reductions and expansions of the term in the conclusion (the `subject' of the typing judgement) are reflected by transformations of its typing derivations.
First we show that type assignment is closed for preforming or reversing both name and term substitution.

 \begin{lem}[Term substitution lemma] \label{lem:substitution}
 \begin{enumerate}
\firstitem \label{lem:substitution-name}
	$\derLmu `G |- T[`a/`b] : `s | `D $ with $`b \notele `D $ if and only $\derLmu `G |- T : `s | `b {:}`k,`D $ and $`D(`a) \leq_C `k$.
\item \label{lem:substitution-term}
	$\derLmu `G |- T[L/x] : `s | `D $ with $x\notele `G$ if and only if there exists $`d'$ such that $\derLmu `G,x{:}`d' |- T : `s | `D $ and $\derLmu `G |- L : `d' | `D $.
 \end{enumerate}
 \end{lem}

 \begin{proof} 

 \begin{enumerate} \itemsep 3pt

\firstitem 
By straightforward induction on the structure of derivations. 
Note that the type of the name $`a$ which is substituted for $`b$ in $T[`a/`b]$ is not necessarily the same as that of $`b$ in $T$, but it is weaker in general.

 \item 
By induction on the definition of term substitution; notice that, if $T$ is a term, then $`s = `d$, and if $T$ is a command, then $`s = `k$.

 \begin{description} \itemsep 3pt

 \item [$ T\same x $] 
 \begin {description} 

 \item [$ \Then $] 
If $ \derLmu `G |- x[L/x] : `d | `D $, then $ \derLmu `G |- L : `d | `D $ and we have $ \derLmu `G,x{:}`d |- x : `d | `D $ by $(\Ax)$.

 \item [$ \If $] 
If $ \derLmu `G,x{:}`d' |- x : `d | `D $, then $`d' \seqD `d$ by Lem.\skp\ref{gen lemma}.
From $\derLmu `G |- L : `d' | `D $ and rule $(\seq)$, we have $ \derLmu `G |- L : `d | `D $, so also $ \derLmu `G |- x[L/x] : `d | `D $.

 \end {description}

 \item [$ T\same y\not= x $] 
 \begin {description}
 \item [$ \Then $] By rule $(\Weak)$, $ \derLmu `G,x{:}`w |- y : `d | `D $, and $ \derLmu `G |- L : `w | `D $ by rule $(`w)$.
 \item [$ \If $] By rule $(\Thin)$, since $x \notele \fv{y}$.
 \end {description}

\Comment{
 \item [$ T\same `l y.N $] 
$ \kern-24mm \begin{array}[t]{lcl} \kern24mm 
\derLmu `G,x{:}`d' |- `ly.N : `d | `D \And \derLmu `G |- L : `d' | `D 
	& \Iff & (\textrm{Lem.\skp\ref{gen lemma}}) \\
\multicolumn{3}{l}{ 
\Exists I \, \Forall i \ele I \, \Exists `d_i,`k_i,`r_i \Pred[ \derLmu `G,x{:}`d',y{:}`d_i |- N : `k_i\arrow `r_i | `D \And \Inter_{i \ele I} `d_i\prod `k_i\arrow `r_i \seqD `d \And \derLmu `G |- L : `d' | `D ] 
} \\ 
	& \Iff & (\textrm{by induction}) \\
\multicolumn{3}{l}{ 
\Exists I \, \Forall i \ele I \, \Exists `d_i,`k_i,`r_i \Pred[ \derLmu `G, y{:}`d_i |- N[L/x] : `k_i\arrow `r_i | `D \And \Inter_{i \ele I} `d_i\prod `k_i\arrow `r_i \seqD `d ] 
} \\
	& \Iff & (\textrm{Lem.\skp\ref{gen lemma}}) \\
\derLmu `G |- `ly.(N[L/x]) : `d | `D 
	~ \ByDef ~ 
\derLmu `G |- (`ly.N)[L/x] : `d | `D
 \end{array} $
 \arrayqed
}

 \item [$ T \equiv PQ $] 
$ \kern-20mm \begin{array}[t]{lcl} \kern20mm 
\Exists `d' \Pred[ \derLmu `G,x{:}`d' |- PQ : `d | `D \And \derLmu `G |- L : `d' | `D ]
	& \Iff & \\ 
\multicolumn{3}{r}{(\textrm{Lem.\skp\ref{gen lemma}; for $\Then$, take $`d' = `d'_1 = `d'_2$, for $\Else$, take $`d' = `d'_1 \inter `d'_2$}) } \\
\multicolumn{3}{l}{\Exists `d'_1,`d'_2, I \, \Forall i \ele I \, \Exists `d_i,`k_i,`r_i, \, [\psk \derLmu `G,x{:}`d'_1 |- P : `d_i\prod `k_i\arrow `r_i | `D \And \derLmu `G,x{:}`d'_2 |- Q : `d_i | `D \And }
	\quad \\ ~ \quad \hfill 
\Inter_{i \ele I}`k_i\arrow `r_i \seqD `d \And \derLmu `G |- L : `d_1' \And \derLmu `G |- L : `d_2' | `D \psk] 
	& \Iff & (\textrm{by induction}) \\
\Exists I \, \Forall i \ele I \, \Exists `d_i,`k_i,`r_i \, [\psk \derLmu `G |- P[L/x] : `d_i\prod `k_i\arrow `r_i | `D \And 
	\quad \\ ~ \quad \hfill 
\derLmu `G |- Q[L/x] : `d_i | `D \And \Inter_{i \ele I}`k_i\arrow `r_i \seqD `d \psk] 
	& \Iff & (\textrm{Lem.\skp\ref{gen lemma}}) \\
\derLmu `G |- P[L/x] : `d | `D \And \derLmu `G |- Q[L/x] : `d | `D 
	~ \ByDef ~ 
\derLmu `G |- (PQ)[L/x] : `s | `D
 \end{array} $

\Comment{
 \item [$ T\same `m`a.\Cmd $] 
$ \kern-15mm \begin{array}[t]{lcl} \kern15mm 
\derLmu `G,x{:}`d' |- `m`a.\Cmd : `k\arrow `r | `D \And \derLmu `G |- L : `d' | `D 
	& \Iff & (\textrm{Lem.\skp\ref{gen lemma}}) \\
\Exists I \, \Forall i \ele I \, \Exists `k_i,`r_i, `k'_i \,[\psk \derLmu `G,x{:}`d' |- {\Cmd} : (`k_i\arrow `r_i)\prod `k_i | `a{:}`k'_i,`D \And 
	\quad \\ ~ \quad \hfill 
\Inter_{i \ele I}`k'_i\arrow `r_i \seqD `s \And \derLmu `G |- L : `d' | `D \psk]
	&\Iff & (\textrm{by induction}) \\ 
\Exists I \, \Forall i \ele I \, \Exists `k_i,`r_i, `k'_i \,[\psk \derLmu `G |- {\Cmd}[N/x] : (`k_i\arrow `r_i)\prod `k_i | `a{:}`k'_i,`D \And 
	\quad \\ ~ \quad \hfill 
\Inter_{i \ele I}`k'_i\arrow `r_i \seqD `s \psk]
	&\Iff & (\textrm{Lem.\skp\ref{gen lemma}}) \\
\derLmu `G |- `m`a. \Cmd[L/x] : `d | `D 
	~ \ByDef ~ 
\derLmu `G |- `m`a. \Cmd[L/x] : `d | `D 
 \end{array} $

 \item [$ T\same {[`a]}N $] 
$ \kern-15mm \begin{array}[t]{lcl} \kern15mm 
\derLmu `G,x{:}`d' |- [`a]N : `k | `D \And \derLmu `G |- L : `d' | `D 
	&\Iff & (\textrm{Lem.\skp\ref{gen lemma}}) \\
\Exists I \, \Forall i \ele I \, \Exists `d_i,`k_i \,[\psk \derLmu `G |- N : `d_i | `D \And `D(`a) \seqC `k_i \And 
	\quad \\ ~ \quad \hfill 
\Inter_{i \ele I} `d_i\prod `k_i \seqC `k \And \derLmu `G, x{:}`d' |- L : `d' | `D \psk] 
	&\Iff & (\textrm{by induction}) \\ 
\Exists I \, \Forall i \ele I \, \Exists `d_i,`k_i \,[\psk \derLmu `G |- N[L/x] : `d_i | `D \And `D(`a) \seqC `k_i \And 
	\quad \\ ~ \quad \hfill 
\Inter_{i \ele I} `d_i\prod `k_i \seqC `k \psk] 
	&\Iff & (\textrm{Lem.\skp\ref{gen lemma}}) \\
\derLmu `G |- [`a](N[L/x]) : `k | `D 
	~ \ByDef ~ 
\derLmu `G |- ([`a]N)[L/x] : `k | `D 
 \end{array} $

}

 \item[$ T\same `l y.N $), $(T\same `m`a.\Cmd)$, ($ T\same {[`a]}N $] 
Straightforward. \qed 
 \end{description}

 \end{enumerate}
 \end{proof}

The next lemma states how the structural substitution $T[`a\Becomes L]$ is related to the type of the name $`a$. 
When $T\ele \Terms$ or $T\same [`b]N$ with $`a \neq `b$, the type of $T[`a\Becomes L]$ remains the same as that of $T$, which is similar to the term substitution lemma, but the context $`D $ used to type $T[`a\Becomes L]$ changes to $`D'$ which is equal to $`D $ but for $`D'(`a) = `d'\prod`D(`a)$, where $`d'$ is a type of $L$ (in the same basis and context). 
When $T\same [`a]N$ the effect of the structural substitution is more complex, and it affects also the type of $T$ with respect to that of $T[`a\Becomes L]$. 
The fact that this does not invalidate the type invariance of terms with respect to structural substitution is due to the form of rule $(\MAbs)$ which essentially is a cut rule: indeed, the `cut type' changes in case of $T$ with respect to that of $T[`a\Becomes L]$, but then is hidden in the conclusion.

In the following proofs we let $I$, $J$, $H$, possibly with apices and indices, range over finite and non-empty sets of indices.

 \begin{lem} [Structural substitution lemma] \label{structural substitution lemma}
Let 
$`a \neq `b$, and assume that $`a \notele \fn(L)$; then:
 \begin{enumerate}

 \item \label{structural substitution lemma i}
$\derLmu `G |- {M[`a\Becomes L]} : `d | `a{:}`k,`D $ if and only if there exists $`d'$ such that $ \derLmu `G |- L : `d' | `D $ and $\derLmu `G |- M : `d | `a{:}`d'\prod`k,`D $.

 \item \label{structural substitution lemma ii}
$ \derLmu `G |- {([`b]N)[`a\Becomes L]} : `k' | `a{:}`k,`D $ if and only if there exists $`d'$ such that $ \derLmu `G |- L : `d' | `D $ and $\derLmu `G |- [`b]N : `k' | `a{:}`d'\prod`k,`D $.

 \item \label{structural substitution lemma iii}
 $\derLmu `G |- {([`a]N)[`a\Becomes L]} : `k' | `a{:}`k,`D $ if and only if either $`k' \sim_C `w$, or
there exist $`d'$, 
 $`k_i$ and $`r_i$ for all $i\ele I$ such that:
 \begin{enumerate}
 \item $\derLmu `G |- L : `d' | `D $
 \item $\derLmu `G |- {[`a]N}: (`d'\prod `k_i\arrow `r_i)\prod `d'\prod `k | `a{:}`d'\prod `k,`D $ for all $i\ele I$
 \item $\bigwedge_{i \ele I}(`k_i\arrow `r_i)\prod `k \seq `k'$
 \end{enumerate}

 \end{enumerate}
 \end{lem}

 \begin{proof} 
By simultaneous induction on derivations. 
Observe that whenever $`a \notele \fn(T)$ for $T\same M$ in part (\ref{structural substitution lemma i}) and $T\same ([`b]N)$ in part (\ref{structural substitution lemma ii}), we have that $T[`a\Becomes L] \same T$, so that the lemma is vacuously true by taking $`d' = `w$. 
We only show the relevant cases.

 \begin{enumerate} \itemsep 3pt

 \item 

\leftmarginiii = -5mm 
 \begin{description} 

 \shiftitem [$\Then$]

 \begin{description} 
\itemindent 1cm

 \item [$`m$]
Then $M = `m `g.\Cmd$, $`d = `k_1\arrow `r$ and $ \derLmu `G |- {\Cmd[`a\Becomes L]} : {(`k_2\arrow `r)\prod`k_2} | `g{:}`k_1, `a{:}`k,`D $ in a sub-derivation.
If $`r \sim_R `w$ then $`k_1\arrow `r \sim_D `w$ and the thesis trivially 
follows by rule $(`w)$. Otherwise assume that $`r \not\sim_R `w$.
We now distinguish two cases:

 \begin{description}

 \item [{$ \Cmd \same {[`b]}N $ with $`a \not= `b$}] 
Then \emph{IH}~(\ref{structural substitution lemma ii}) applies, so there exists $`d'$ such that $\derLmu `G |- L : `d' | `g{:}`k_1,`D $ and $\derLmu `G |- [`b]N : (`k_2\arrow `r)\prod`k_2 | `g{:}`k_1,`a{:}`d'\prod`k,`D $.
Then by applying $(\MAbs)$, $\derLmu `G |- `m `g.[`b]N : `k_1\arrow `r | `a{:}`d'\prod`k,`D $ follows. 
For the second part, since $`g$ is bound in $`m `g.{[`b]}N$, by Barendregt's convention $`g\notele \fn(L)$, so applying rule $(\Thin)$ to $\derLmu `G |- L : `d' | `g{:}`k_1,`D $ gives $ \derLmu `G |- L : `d' | `D $.

\item [{$ \Cmd \same {[`a]}N $}] 
Now \emph{IH}\,(\ref{structural substitution lemma iii}) applies to $\Der'$, so that:
 \begin{enumerate}
 \item $\derLmu `G |- L : `d' | `D $ for some $`d'$;
 \item $\derLmu `G |- {[`a]N}: (`d'\prod `k_i\arrow `r_i)\prod `d'\prod `k | `a{:}`d'\prod `k,`D $ for all $i\ele I$ and some $I$;
 \item $\bigwedge_{i \ele I}(`k_i\arrow `r_i)\prod `k \seq_C (`k_2\arrow `r)\prod`k_2$.
 \end{enumerate}
It follows that $\bigwedge_{i \ele I}(`k_i\arrow `r_i) \seq_D `k_2\arrow `r$ and $`k \seq_C `k_2$. 
By Lem.\skp\ref{lem:SubtypeProp}\skp(\ref{lem:SubtypeProp-3}) the 
first in-equation implies that there exists $i_0 \ele I$ such that $`k_2 \seq_C `k_{i_0}$ and $`r_{i_0} \seq_R `r$. 
Hence we have
$(`d'\prod `k_{i_0}\arrow `r_{i_0})\prod `d'\prod `k \seq (`d'\prod `k_2\arrow `r)\prod `d'\prod `k_2$, and therefore
 \[ \begin{array}{c}
\Inf[\MAbs]
	{\Inf[\seq]
{\derLmu `G |- {[`a]N}: (`d'\prod `k_{i_0}\arrow `r_{i_0})\prod `d'\prod `k | `a{:}`d'\prod `k,`D }
{\derLmu `G |- {[`a]N}: (`d'\prod `k_2\arrow `r)\prod `d'\prod `k | `a{:}`d'\prod `k,`D }
	}{\derLmu `G |- {`m `g.[`a]N} : `k_1\arrow `r | `a{:}`d'\prod `k,`D }
 \end{array}\]

 \end{description}

 \end{description}

 \item [$\Else$]
Let $\derLmu `G |- L : `d' | `D $ and assume $\Der::\derLmu `G |- `m`g.\Cmd : `k_1\arrow `r | `a{:}`d'\prod`k,`D $ ends with rule $(\MAbs)$. Again the thesis
is trivial if $`r \sim_R `w$, so that we suppose that $`r \not\sim_R `w$. As before we have to distinguish the cases
$\Cmd \same [`b]N$ with $`b \neq `a$ and $\Cmd \same [`a]N$. In the first case \emph{IH}~(\ref{structural substitution lemma ii}) directly implies the thesis, so let's consider the case $\Cmd \same [`a]N$, in which $\Der$ ends with
 \[ \begin{array}{c}
\Inf[\MAbs]
	{\derLmu `G |- {[`a]N}: (`k_2\arrow `r)\prod `k_2 | `g{:}`k_1, `a{:}`d'\prod `k,`D }
	{\derLmu `G |- {`m `g.[`a]N} : `k_1\arrow `r | `a{:}`d'\prod `k,`D }
 \end{array}\]
By Lem.\skp\ref{gen lemma} we know that there exist sub-derivations $\Der_i$ of $\Der$, with $i\ele I$ for some $I$, such that
 \[ \begin{array}{ccccc}
\Der_i:: \derLmu `G |- {N}: `d_i | `g{:}`k_1, `a{:}`d'\prod `k,`D, & &
`d'\prod `k \leq_C `k_i, & \textrm{ and } & \bigwedge_{i\ele I} `d_i\prod `k_i \leq_C (`k_2\arrow `r)\prod `k_2
 \end{array}\]
It follows that
$(\bigwedge_{i\ele I} `d_i)\prod `d'\prod `k \seq_C (`k_2\arrow `r)\prod `k_2$ which implies
$\bigwedge_{i\ele I} `d_i \leq_D `k_2\arrow `r$ and $`d'\prod `k \leq_C `k_2$.
Now the assumption $`r \not\sim_R `w$ implies $\bigwedge_{i\ele I} `d_i \not\sim_D `w$, and therefore
 \[ \begin{array}{rcl@{\quad}lccc}
\bigwedge_{i\ele I} `d_i &\sim_D& \bigwedge_{i\ele I, j\ele J_i}\, `k_{i,j}\arrow `r_{i,j} &\leq& `k_2\arrow `r &\leq& `d'\prod `k\arrow`r.
 \end{array} \]
By Lem.\skp\ref{lem:SubtypeProp}\skp(\ref{lem:SubtypeProp-3}) this implies that there are $I' \subseteq I$ and $J'_i \subseteq J_i$ such that
$`d'\prod `k \leq_C `k_2 \leq_C \bigwedge_{i\ele I', j\ele J'_i}\, `k_{i,j}$ and $\bigwedge_{i\ele I', j\ele J'_i}\, `r_{i,j} \leq_R `r$.
Then $`k_{i,j} \sim_C `d_{i,j}\prod `k'_{i,j}$ for every $i,j$, where $`d' \leq_D `d_{i,j}$ and 
$`k \leq_C `k'_{i,j}$. Since clearly $`d_i \leq_D \bigwedge_{j\ele J'_i}\, `k_{i,j}\arrow `r_{i,j}$, for all $i\ele I \supseteq I'$,
by $(\leq)$ we have that
$\derLmu `G |- N : `k_{i,j}\arrow `r_{i,j} | `g{:}`k_1, `a{:}`d'\prod `k,`D $; by $(\TCmd)$ we get
 \[ \derLmu `G |- {[`a]N}: (`k_{i,j}\arrow `r_{i,j})\prod `d'\prod `k | `g{:}`k_1, `a{:}`d'\prod `k,`D \]
for every $i\ele I', j\ele J'_i$. 
Now
 \[ \begin{array}{lll@{\quad}l}
(`k_{i,j}\arrow `r_{i,j})\prod `d'\prod `k & \sim_C & (`d_{i,j}\prod `k'_{i,j}\arrow `r_{i,j})\prod `d'\prod `k & (\textrm{as }`k_{i,j} \sim_C `d_{i,j}\prod `k'_{i,j}) \\
& \leq_C & (`d'\prod `k'_{i,j}\arrow `r_{i,j})\prod `d'\prod `k & (\textrm{as }`d' \leq_D `d_{i,j})
 \end{array}\]
Taking $`k' = \bigwedge_{i\ele I', j\ele J'_i} ( `k'_{i,j}\arrow `r_{i,j})$, \emph{IH}\skp(\ref{structural substitution lemma iii}) applies, 
from which 
$\derLmu `G |- {([`a]N)[`a \Becomes L]} : `k' | `g{:}`k_1, `a{:}`d'\prod `k,`D $ follows. 
Using that $`k \leq_C `k'_{i,j}$ for all $i\ele I'$ and $j\ele J'_i$, we deduce
 \[ \begin{array}{rcl@{\quad}lccc}
`k' &\leq_C& \bigwedge_{i\ele I', j\ele J'_i}(`k\arrow `r_{i,j} )\prod `k &\sim_C&
(`k\arrow \bigwedge_{i\ele I', j\ele J'_i}`r_{i,j} )\prod `k &\leq_C& 
(`k\arrow `r )\prod `k
 \end{array} \]
and therefore $\derLmu `G |- ([`a]N)[`a \Becomes L] : (`k\arrow `r )\prod `k | `g{:}`k_1, `a{:}`d'\prod `k,`D $ by
$(\seq)$, from which we conclude
$\derLmu `G |- {(`m`g[`a]N)[`a \Becomes L]} : `k_1\arrow `r | `a{:}`d'\prod `k,`D $ by $(\MAbs)$ as desired.

 \end{description}

 \item 
Note that $([`b]N)[`a\Becomes L] \same [`b] (N[`a\Becomes L])$.

\leftmarginiii = -5mm 
 \begin{description} 

 \shiftitem [$\Then$]

 \begin{description} 
\itemindent 1cm

 \item [$\TCmd$]
Then $M = [`b](N[`a\Becomes L]) : `k' | `a{:}k,`D $, $`k' = `d''\prod `k''$, $`D = `b{:}`k'',`D'$, and $ \derLmu `G |- {N[`a\Becomes L]} : `d'' | `a{:}`k,`b{:}`k'',`D $.
Then by \emph{IH}\,(\ref{structural substitution lemma i}), there exists $`d'$ such that $ \derLmu `G |- L : `d' | `D $ and $\derLmu `G |- N : `d'' | `a{:}`d'\prod`k,`b{:}`k'',`D' $.
By rule $(\TCmd)$ we get $\derLmu `G |- {[`b]N} : `k' | `a{:}`d'\prod`k,`D $.
 \end{description}

 \item[$\Else$] 
The reasoning for this part is the reverse of part $(\Then)$.

 \end{description}

 \item 
We prove the two implications; note that $([`a]N)[`a\Becomes L] \same [`a](N[`a\Becomes L])L$.

 \begin{description}
 \item [$\Then$] 
Let $\Der :: \derLmu `G |- {[`a](N[`a\Becomes L])L} : `k' | `a{:}`k,`D $ be a derivation, and assume that $`k' \not\sim_C `w$. By Lem.\skp\ref{gen lemma} there exists the sub-derivations of $\Der$:
 \[ \begin{array}{rcl@{\quad}l}
\Der_i:: \derLmu `G |- {(N[`a \Becomes L])L} : `d_i | `a{:}`k,`D 
 \end{array} \]
where $`k \seq_C `k_i$ and $\bigwedge_{i\ele I} `d_i\prod`k_i \seq_C `k'$.
If all the $`d_i \sim_D `w$ then we have
$`w\prod `k \leq_C `w\prod `k_i$ for all $i$, hence
 $ \begin{array}{ccccc}
`w\prod `k &\leq_C& \bigwedge_{i\ele I}`w\prod `k_i &\leq_C& `k'.
 \end{array} $
Taking $`d' = `w$ we trivially have that $\derLmu `G |- L: `d' | `D $ and
 \[ \begin{array}{c}
\Inf	[\TCmd]
	{\Inf	[\omega]
{\derLmu `G |- N : `w | `a{:}`w\prod `k,`D }
	}{\derLmu `G |- {[`a]N} : `w\prod `w\prod `k | `a{:}`w\prod `k,`D }
 \end{array}\]
Now taking $`r_i = `w$ we have $`w \sim_D `w\prod `k_i\arrow `w$ so that
$\derLmu `G |- {[`a]N}: (`w\prod `k_i\arrow `w)\prod `w\prod `k | `a{:}`w\prod `k,`D $ for all $i$, and
$(`k_i\arrow `w)\prod `k \sim_C `w\prod `k \leq_C `k'$.

Otherwise (not all the $`d_i \sim_D `w$), assume (without loss of generality) that all the $`d_i \not\sim_D `w$.
By Lem.\skp\ref{gen lemma} every $\Der_i$ has sub-derivations:
 \[ \begin{array}{r@{\quad}c@{\quad}l}
{{\Der^1}}_{i,j}:: \derLmu `G |- {N[`a \Becomes L]} : `d_{i,j}\prod `k_{i,j}\arrow `r_{i,j} | `a{:}`k,`D & \text{and} &{\Der^2}_{i,j}:: \derLmu `G |- L : `d_{i,j} | `a{:}`k,`D 
 \end{array}\]
(where $\bigwedge_{j \ele J_i} `k_{i,j}\arrow `r_{i,j} \leq_D `d_i$)
for $j$ varying over certain $J_i$. Note that since $`a \notele \fn{L}$ we can assume that
${\Der^2}_{i,j} :: \derLmu `G |- L : `d_{i,j} | `D $, without any assumption for the name $`a$.
By \emph{IH}\,(\ref{structural substitution lemma i}), from the existence of ${\Der^1}_{i,j}$ we infer that derivations
 \[ \begin{array}{rcl@{\quad}l}
\widehat{{\Der^1}}_{i,j}:: \derLmu `G |- {N} : `d_{i,j}\prod `k_{i,j}\arrow `r_{i,j} | `a{:}`d_{i,j}'\prod`k,`D &\text{and} &
\widehat{{\Der^2}}_{i,j}:: \derLmu `G |- L : `d_{i,j}' | `D 
 \end{array}\]
exist.
Taking $`d' = \bigwedge_{i\ele I, j\ele J_i} (`d_{i,j} \inter `d'_{i,j})$ we have:
 \[ \begin{array}{c}
\Inf[\inter]
	{\InfBox{{\Der^2}_{i,j}} {\derLmu `G |- L : `d_{i,j}| `D } 
	 \dquad
	 \InfBox{\widehat{\Der}^2_{i,j}}{\derLmu `G |- L : `d_{i,j}' | `D } \quad
	 \forall i\ele I', j\ele J'_i
	}{\derLmu `G |- L : `d' | `D }
 \end{array}\]
Since $`d' \leq_C `d_{i,j}$ implies
$`d_{i,j}\prod `k_{i,j}\arrow `r_{i,j} \seq_D `d'\prod `k_{i,j}\arrow `r_{i,j}$ for all $i\ele I, j\ele J_i$,
from the $\widehat{\Der}^1_{i,j}$ we obtain the derivations:
 \[ \begin{array}{c}
\Inf	[\seq] 
	{\Inf	[\Strength]
{\InfBox{\widehat{\Der}^1_{i,j}}
	{\derLmu `G |- {N} : `d_{i,j}\prod `k_{i,j}\arrow `r_{i,j} | `a{:}`d_{i,j}'\prod`k,`D }
	 \quad `d' \leq_C `d'_{i,j}
 }{\derLmu `G |- {N} : `d_{i,j}\prod `k_{i,j}\arrow `r_{i,j} | `a{:}`d'\prod`k,`D }
	}{\derLmu `G |- {N} : `d'\prod `k_{i,j}\arrow `r_{i,j} | `a{:}`d'\prod`k,`D }
 \end{array}\]	
From this and by $(\TCmd)$ we get
$\derLmu `G |- {[`a]N} : (`d'\prod `k_{i,j}\arrow `r_{i,j})\prod `d'\prod`k | `a{:}`d'\prod`k,`D $
for all $i\ele I$ and $j\ele J_i$. On the other hand from the above, for all $i\ele I$ we have:
 \[ \begin{array}{rcl@{\quad}lc}
\bigwedge_{j \ele J_i} (`k_{i,j}\arrow `r_{i,j})\prod `k &\leq_C& `d_i\prod `k_i &\leq_C& `k'
 \end{array} \]
and hence
$\bigwedge_{i,\ele I, j \ele J_i} (`k_{i,j}\arrow `r_{i,j})\prod `k \leq_C `k'$
as desired.

\item [$\Else$] 
 Let $\Der_i::\derLmu `G |- {[`a]N}:(`d'\prod `k_i\arrow `r_i)\prod `d'\prod `k | `a{:}`d'\prod `k,`D $
 be a derivation for all $i\ele I$. By Lem.\skp\ref{gen lemma} there exists a $J$
 such that there exist the sub-derivations
 $\Der_{i,j}::\derLmu `G |- {N}: `d_{i,j} | `a{:}`d'\prod `k,`D $ and types $`k_{i,j}$ such that for all $j\ele J$ 
 \[ \begin{array}{rcl@{\quad}c@{\quad}rcl}
`d'\prod `k &\leq_C& `k_{i,j} &\textrm{and}& `d_{i,j}\prod `k_{i,j} 
 	&\leq_C& (`d'\prod `k_i\arrow `r_i)\prod `d'\prod `k 
 \end{array} \]
This implies that 
 \[ \begin{array}{rcl@{\quad}lc}
`d_{i,j}\prod `d'\prod `k &\leq_C& `d_{i,j}\prod `k_{i,j} 
 	&\leq_C& (`d'\prod `k_i\arrow `r_i)\prod `d'\prod `k
 \end{array} \]
hence $`d_{i,j} \leq_D `d'\prod `k_i\arrow `r_i$, for all $i$ and $j$. 
By hypothesis $\derLmu `G |- L : `d' | `D $, so that by \emph{IH}\,(\ref{structural substitution lemma i})
there exist the derivations
$\Der'_{i,j}::\derLmu `G |- {N[`a\Becomes L]}: `d_{i,j} | `a{:}`k,`D $ and therefore:
 \[ \begin{array}{c}
 \Inf	[\TCmd]
 	{\Inf	[\App]
{\Inf	[\seq]
{\InfBox{\Der'_{i,j}}{\derLmu `G |- N[`a\Becomes L]: `d_{i,j} | `a{:} `k,`D }
}{\derLmu `G |- N[`a\Becomes L]: `d'\prod `k_i\arrow `r_i | `a{:} `k,`D }	
\qquad
 \Inf	[\Weak]
 	{\InfBox{\derLmu `G |- L : `d' | `D }
 	}{\derLmu `G |- L : `d' | `a{:} `k,`D }
}{\derLmu `G |- {(N[`a\Becomes L])L}: `k_i\arrow`r_i | `a{:} `k,`D }
	}{\derLmu `G |- {[`a](N[`a\Becomes L])L}: (`k_i\arrow`r_i)\prod `k | `a{:} `k,`D }
 \end{array} \]
is a derivation $\Der'_i$ for all $i \ele I$. Hence by $(\inter)$ we get
 \[\derLmu `G |- {[`a](N[`a\Becomes L])L}: \bigwedge_{i \ele I}(`k_i\arrow `r_i)\prod `k | `a{:} `k,`D\] 
and we conclude by the hypothesis $\bigwedge_{i \ele I}(`k_i\arrow `r_i)\prod `k \seq `k'$
and $(\seq)$. \qed

 \end{description}
 \end{enumerate}

 \end{proof}

We are now in place to prove the subject reduction theorem.

 \begin{thm}[Subject reduction] \label{Subject reduction} 
If $M \redbm N$, and $ \derLmu `G |- M : `d | `D $, then $ \derLmu `G |- N : `d | `D $.
 \end{thm}

 \begin{proof} %
By induction on the definition of the reduction relation; we will only consider the cases for the three basic reduction rules.

 \begin{description} \itemsep 2pt

 \item [{$(`l x . M ) N \red M [ N /x ] $}] 
By Lem.\skp\ref{gen lemma}, if $\derLmu `G |- {(`l x . M ) N} : `d | `D $ then there exists $I$ such that, for all $i\ele I$, $\derLmu `G |- {`l x . M} : `d_i\prod `k_i\arrow `r_i | `D $, $\derLmu `G |- N : `d_i | `D $ and $\bigwedge_{i\ele I} `k_i\arrow `r_i \seq_D `d$. 
By the same lemma, for every $i \ele I$ there exists $J_i$ such that $\derLmu `G, x{:} `d_{i,j} |- M : `k_{i,j}\arrow `r_{i,j} | `D $ for all $j\ele J_i$ and $\bigwedge_{j\ele J_i} `d_{i,j}\prod `k_{i,j}\arrow `r_{i,j} \seq_D `d_i\prod `k_i\arrow `r_i $.
By Lem.\skp\ref{lem:SubtypeProp}\skp(\ref{lem:SubtypeProp-3}) the last in-equation implies that for all $i\ele I$ there is a $J'_i \subseteq J_i$ such that:
 \[ \begin{array}{rcl@{\quad}c@{\quad}rcl}
`d_i\prod `k_i &\leq_C& \bigwedge_{j\ele J'_i}\,`d_{i,j}\prod `k_{i,j} 
	&\textrm{and}&
\bigwedge_{j\ele J'_i}\,`r_{i,j} &\seq_R& `r_i.
 \end{array} \]
It follows that $`d_i \leq_D `d_{i,j}$, and $ `k_i \seq_C `k_{i,j}$ and $`r_{i,j} \seq_R `r_i$, and hence $`k_{i,j}\arrow `r_{i,j} \seq_D `k_i\arrow `r_i$, for all $i\ele I$ and $j\ele J'_i$.
Therefore, by $(\Strength)$ and $(\seq)$ we have $\derLmu `G, x{:} `d_i |- {M} : `k_i\arrow `r_i | `D $, so that by Lem.\skp\ref{lem:substitution}\skp(\ref{lem:substitution-term}) we have $\derLmu `G |- {M[N/x]} : `k_i\arrow `r_i | `D $ for all $i\ele I$.
From this we obtain that $\derLmu `G |- {M[N/x]} : \bigwedge_{i\ele I} `k_i\arrow `r_i | `D $ by $(\inter)$, and conclude by $(\seq)$ using $\bigwedge_{i\ele I} `k_i\arrow `r_i \seq_D `d$.

 \item [{$ (`m `a . \Cmd ) N \red `m `a . \Cmd [ `a\Becomes N ] $}]
By Lem.\skp\ref{gen lemma}, if $\derLmu `G |- {(`m `a . \Cmd ) N} : `d | `D $ then there exists $I$ such that, for all $i\ele I$, $\derLmu `G |- {`m `a . \Cmd} : `d_i\prod `k_i\arrow `r_i | `D $, $\derLmu `G |- N : `d_i | `D $ and $\bigwedge_{i\ele I} `k_i\arrow `r_i \seq_D `d$. 
By the same lemma, for every $i \ele I$ there exists $J_i$ such that $\derLmu `G |- {\Cmd} : (`k_{i,j}\arrow `r_{i,j})\prod `k_{i,j} | `a{:}`k'_{i,j},`D $ for all $j\ele J_i$ and $\bigwedge_{j\ele J_i} `k'_{i,j}\arrow `r_{i,j} \seq_D `d_i\prod `k_i\arrow `r_i $.
By Lem.\skp\ref{lem:SubtypeProp}\skp(\ref{lem:SubtypeProp-3}) the last in-equation implies that for all $i\ele I$ there is a $J'_i \subseteq J_i$ such that:
 \[ \begin{array}{rcl@{\quad}c@{\quad}rcl}
`d_i\prod `k_i &\leq_C& \bigwedge_{j\ele J'_i}\, `k'_{i,j} 
	& \textrm{and} & 
\bigwedge_{j\ele J'_i}\,`r_{i,j} \seq_R `r_i.
 \end{array} \]
Since each $`k'_{i,j} \sim_C `d'_{i,j}\prod `k''_{i,j}$ for some $`d'_{i,j}, `k''_{i,j}$, we have $`d_i \seq_D `d'_{i,j}$ and $`k_i \seq_C `k''_{i,j}$ for all $i\ele I$ and $j\ele J'_i$. Note that also $\derLmu `G |- N : `d'_{i,j} | `D $ by $(\seq)$ for all $i,j$ as above.
	
We distinguish the cases: 

~\kern-10mm $( \Cmd \same {[`b]}L$, with $`b \neq `a ) : $ 
By Lem.\skp\ref{structural substitution lemma}\skp(\ref{structural substitution lemma ii}),
from $\derLmu `G |- {{[`b]}L} : (`k_{i,j}\arrow `r_{i,j})\prod `k_{i,j} | `a{:} `d'_{i,j}\prod `k''_{i,j},`D $ and 
$\derLmu `G |- N : `d'_{i,j} | `D $ we have that
$\derLmu `G |- {{[`b]}L[`a \Becomes N]} : (`k_{i,j}\arrow `r_{i,j})\prod `k_{i,j} | `a{:} `k''_{i,j},`D $ and
hence $\derLmu `G |- {`m `a. {[`b]}L[`a \Becomes N]} : `k''_{i,j}\arrow `r_{i,j} | `D $ for all $i\ele I$ and $j\ele J'_i$.
Since $`k_i \seq_C `k''_{i,j}$ and $`r_{i,j} \seq_R `r_i$ we deduce that $`k''_{i,j}\arrow `r_{i,j} \seq_D `k_i\arrow `r_i$
for all $i\ele I$ so that $\derLmu `G |- {`m `a. {[`b]}L[`a \Becomes N]} : `k_i\arrow `r_i | `D $ by $(\seq)$
and $\derLmu `G |- {`m `a. {[`b]}L[`a \Becomes N]} : \bigwedge_{i\ele I}`k_i\arrow `r_i | `D $ by $(\inter)$,
and we conclude $\derLmu `G |- {`m `a. {[`b]}L[`a \Becomes N]} : `d | `D $ using 
$\bigwedge_{i\ele I} `k_i\arrow `r_i \seq_D `d$ and $(\seq)$.

~\kern-10mm $( \Cmd \same {[`a]}L ) : $ 
we have $\derLmu `G |- {{[`a]}L} : (`k_{i,j}\arrow `r_{i,j})\prod `k_{i,j} | `a{:}`k'_{i,j},`D $ 
and since $`d_i\prod `k_i \leq_C \bigwedge_{j\ele J'_i}\, `k'_{i,j}$ by $(\Strength)$ we know that
$\derLmu `G |- {{[`a]}L} : (`k_{i,j}\arrow `r_{i,j})\prod `k_{i,j} | `a{:} `d_i\prod `k_i,`D $.
By Lem.\skp\ref{gen lemma} in this case we have that
$\derLmu `G |- L : `k_{i,j}\arrow `r_{i,j} | `a{:} `d_i\prod `k_i,`D $ and $`d_i\prod `k_i \seq_C `k_{i,j}$,
from which we have that $`k_{i,j}\arrow `r_{i,j} \seq_D `d_i\prod `k_i\arrow `r_{i,j}$ is a type of $L$
and hence $\derLmu `G |- {{[`a]}L} : (`d_i\prod `k_i\arrow `r_{i,j})\prod `d_i\prod `k_i | `a{:}`d_i\prod `k_i,`D $
by $(\TCmd)$. 
Now since $\derLmu `G |- N : `d_i | `D $ for all $i\ele I$,
by Lem.\skp\ref{structural substitution lemma}\skp(\ref{structural substitution lemma iii}) we obtain
$\derLmu `G |- {({[`a]}L)[`a\Becomes N ]} : \bigwedge_{j\ele J'_i}(`k_i\arrow `r_{i,j})\prod `k_i | `a{:}`k_i,`D $
for all $i$; on the other hand $\bigwedge_{j\ele J'_i} (`k_i\arrow `r_{i,j}) \sim_D 
`k_i\arrow \bigwedge_{j\ele J'_i} `r_{i,j} \seq_D `k_i\arrow `r_i$ by the above, and we conclude:
 \[ \begin{array}{c}
 \Inf[\MAbs]
 	{\Inf[\seq]
 {\InfBox{\derLmu `G |- {({[`a]}L)[ `a\Becomes N ]} : \bigwedge_{j\ele J'_i} (`k_i\arrow `r_{i,j})\prod `k_i | `a{:}`k_i,`D }}
{\derLmu `G |- {({[`a]}L)[ `a\Becomes N ]} : (`k_i\arrow `r_i)\prod `k_i | `a{:}`k_i,`D }
	}
	{\derLmu `G |- {`m `a.({[`a]}L)[ `a\Becomes N ]} : `k_i\arrow `r_i | `a{:}`k_i,`D }
 \end{array}\]
Since this holds for all $i\ele I$ and we know that $\bigwedge_{i\ele I} `k_i\arrow `r_i \seq_D `d$, we eventually derive
$\derLmu `G |- {`m `a.({[`a]}L)[ `a\Becomes N ]} : `d | `a{:}`k_i,`D $ by $(\inter)$ and $(\seq)$.


\item [{$ `m`a.[`b]`m`g.\Cmd \red `m`a.\Cmd[`b/`g] $}]
	Let $ \derLmu `G |- {`m`a.[`a']`m`g.\Cmd} : `d | `D $; then by Lem.\skp\ref{gen lemma}
	there exists $I$ such that for all $i\ele I$
	$\derLmu `G |- {[`a']`m`g.\Cmd} : (`k_i\arrow`r_i)\prod`k_i | `a{:}`k'_i,`D $ and
	$\bigwedge_{i\ele I} `k'_i\arrow`r_i \seq_D `d$. 
	
We distinguish two cases:


~\kern-10mm $(`a \neq `b):$
By Lem.\skp\ref{gen lemma} for all $i\ele I$ there exists $J_i$ such that
for all $j\ele J_i$ it is derivable $\derLmu `G |- {`m`g.\Cmd} : `d_{i,j} | `a{:}`k'_i,`D $,
$`D(`b) \leq_C \bigwedge_{i\ele I,j\ele J_i} `k_{i,j}$ and $\bigwedge_{j\ele J_i} `d_{i,j}\prod`k_{i,j} 
\seq_C (`k_i\arrow`r_i)\prod`k_i$, for all $i\ele I$. The latter implies that 
$\bigwedge_{j\ele J_i} `d_{i,j} \seq_D `k_i\arrow`r_i$ and
$`D(`b) \leq_C \bigwedge_{i\ele I}`k_i$. 
By applying Lem.\skp\ref{gen lemma} to $\derLmu `G |- {`m`g.\Cmd} : `d_{i,j} | `a{:}`k'_i,`D $ we know that for some $H_{i,j}$,
$\derLmu `G |- {\Cmd} : (`k_{i,j,h}\arrow`r_{i,j,h})\prod`k_{i,j,h} | `a{:}`k'_i, `g{:}`k'_{i,j,h},`D $
for all $i\ele I$, $j\ele J_i$ and $h\ele H_{i,j}$, and $\bigwedge_{h} `k'_{i,j,h}\arrow`r_{i,j,h} \seq_D `d_{i,j}$ for all $i,j$.
From the in-equations above we derive that 
$\bigwedge_{j\ele J_i,h\ele H_{i,j}} `k'_{i,j,h}\arrow`r_{i,j,h} \seq_D \bigwedge_{j\ele J_i}`d_{i,j} \seq_D `k_i\arrow`r_i$, 
which
implies by Lem.\skp\ref{lem:SubtypeProp}\skp(\ref{lem:SubtypeProp-3}) 
the existence of certain $J'_I \subseteq J_i$ and $H'_{i,j} \subseteq H_{i,j}$ such that
\[ \begin{array}{rcl@{\quad}c@{\quad}rcl@{\quad}l}
`k_i &\seq_C& \bigwedge_{j\ele J'_i,h\ele H'_{i,j}} `k'_{i,j,h}
&\textrm{ and } &
\bigwedge_{j\ele J'_i,h\ele H'_{i,j}} `r_{i,j,h} &\seq_R& `r_i, & 
\textrm{for all }i\ele I.
 \end{array} \]

It follows that $`D(`b) \leq_C \bigwedge_{i\ele I}`k_i \leq_C \bigwedge_{i\ele I,j\ele J'_i,h\ele H'_{i,j}} `k'_{i,j,h}$, so that
by Lem.\skp\ref{lem:substitution}\skp(\ref{lem:substitution-name}) 
$\derLmu `G |- {\Cmd[`b/`g]} : (`k_{i,j,h}\arrow`r_{i,j,h})\prod`k_{i,j,h} | `a{:}`k'_i,`D $, from which we 
obtain by rule $(\MAbs)$ that 
$\derLmu `G |- {`m`a.\Cmd[`b/`g]} : `k'_i\arrow`r_{i,j,h} | `D $
for all $i\ele I, j\ele J'_i$ and $h\ele H'_{i,j}$.

This implies that, by rule $(\inter)$, we can deduce for $`m`a.\Cmd[`b/`g]$ the type 
 \[ \begin{array}{rcl} 
\bigwedge_{i\ele I, j\ele J'_i,h\ele H'_{i,j}} `k'_i\arrow`r_{i,j,h} &\sim_D& 
\bigwedge_{i\ele I} `k'_i\arrow\bigwedge_{j\ele J'_i,h\ele H'_{i,j}}`r_{i,j,h} ;
 \end{array} \] 
but we know from the above that $\bigwedge_{j\ele J'_i,h\ele H'_{i,j}}`r_{i,j,h} \seq_R `r_i$ for all $i\ele I$ and we conclude that
$\bigwedge_{i\ele I} `k'_i\arrow\bigwedge_{j\ele J'_i,h\ele H'_{i,j}}`r_{i,j,h} \seq_D \bigwedge_{i\ele I} `k'_i\arrow`r_i
\seq_D `d$; hence the thesis follows by rule $(\seq)$.
	
~\kern-10mm $(`a = `b) : $
Reasoning as in the previous case, but now we have that for all $i\ele I$
$(`a{:}`k'_i,`D)(`a) = `k'_i \seq_C \bigwedge_{j\ele J'_i,h\ele H'_{i,j}} `k'_{i,j,h}$, so that we immediately
get that $\bigwedge_{i\ele I} `k'_i\arrow\bigwedge_{j\ele J'_i,h\ele H'_{i,j}}`r_{i,j,h}$ is a type of 
$`m`a.\Cmd[`b/`g]$ which is less
than $\bigwedge_{i\ele I} `k'_i\arrow`r_i \seq_D `d$, and we are done.
\qed

 \end{description}
 \end{proof}

We will now show that types are preserved under expansion, the opposite of reduction.

 \begin{thm}[Subject expansion] \label{Subject expansion} 
If $M \redbm N$, and $ \derLmu `G |- N : `d | `D $, then $ \derLmu `G |- M : `d | `D $.
 \end{thm}
 \begin{proof} %
By induction on the definition of reduction, 
where we focus on the rules and assume $`d \not\sim_D `w$ since otherwise the thesis is trivial.

 \begin{description} \itemsep 2pt

 \item [{$ (`l x . M ) N \red M [ N /x ] $}]
If $ \derLmu `G |- M[N/x] : `d | `D $, then by Lem.\skp\ref{lem:substitution}\skp(\ref{lem:substitution-term}) there exists a $`d'$ such that 
$ \derLmu `G,x{:}`d' |- M : `d | `D $ and $ \derLmu `G |- N : `d' | `D $. Since $`d \not\sim_D `w$ 
we have $`d \sim_D \bigwedge_{i\ele I} `k_i\arrow `r_i$, so that
$ \derLmu `G,x{:}`d' |- M : `k_i\arrow `r_i | `D $ by rule $(\seq)$. By rule $(\Abs)$ we get
$ \derLmu `G |- `l x.M : `d'\prod `k_i\arrow `r_i | `D $, and by rule $(\App)$ 
$ \derLmu `G |- (`l x.M)N : `k_i\arrow `r_i | `D $ for all $i\ele I$, so that we conclude by rules $(\inter)$ and $(\seq)$.

 \item [{$ (`m `a . {\Cmd} ) N \red `m `a . {\Cmd} [ `a \Becomes N ] $}]
By Lem.\skp\ref{gen lemma}, from $\derLmu `G |- {`m `a . \Cmd[ `a \Becomes N ]} : `d | `D $ it follows that there is an $I$ 
such that for all $i\ele I$ it is derivable
 \[\derLmu `G |- {{\Cmd}[ `a \Becomes N ]} : (`k_i\arrow `r_i)\prod `k_i | `a{:}`k'_i ,`D \quad \text{with} \quad
 \bigwedge_{i\ele I} `k'_i\arrow `r_i \seq_D `d.\]
We distinguish the following cases of $\Cmd$:

 \begin{description}

 \item [$ \Cmd \same {[`b]}L$, with $`b \neq `a $] 
Then, by Lem.\skp\ref{structural substitution lemma}\,(\ref{structural substitution lemma ii}), for all $i\ele I$
there exists $`d'_i$ such that 
$\derLmu `G |- N : `d'_i | `D $, and $ \derLmu `G |- {{[`b]}L} : (`k_i\arrow `r_i)\prod `k_i | `a{:}`d'_i\prod `k'_i,`D $.
Then, by rule $(`m)$, $ \derLmu `G |- `m`a.{{[`b]}L} : `d'_i\prod `k'_i\arrow `r_i | `D $, and 
$ \derLmu `G |- (`m`a.{{[`b]}L})N : `k'_i\arrow `r_i | `D $ follows by rule $(\App)$. Using that 
$ \bigwedge_{i\ele I} `k'_i\arrow `r_i \seq_D `d$ the thesis follows by rules $(\inter)$ and $(\seq)$.

 \item [$ \Cmd \same {[`a]}L $] 
 By Lem.\skp\ref{structural substitution lemma}\skp(\ref{structural substitution lemma iii}) we know that for all $i\ele I$ there exist
 $J_i$ and types $`d'_i$, $`k_{i,j}$ and $`r_{i,j}$ such that:
 \begin{flatenumerate}
 \item \label{eq:exp i}
 	$\derLmu `G |- N : `d'_i | `D $;
 \item \label{eq:exp ii}
 	$\derLmu `G |- {[`a]L}: (`d'_i\prod `k_{i,j}\arrow `r_{i,j})\prod `d'_i\prod `k'_i | `a{:}`d'_i\prod `k'_i,`D $ for all $i\ele I$; and
 \item \label{eq:exp iii}
 	$\bigwedge_{j \ele J_i}(`k_{i,j}\arrow `r_{i,j})\prod `k'_i \seq_C (`k_i\arrow `r_i)\prod `k_i$.
 \end{flatenumerate}
The in-equation (\ref{eq:exp iii}) implies that $\bigwedge_{j \ele J_i}`k_{i,j}\arrow `r_{i,j} \seq_D `k_i\arrow `r_i$ and 
$`k'_i \seq_C `k_i$ for all $i\ele I$, 
so that by Lem.\skp\ref{lem:SubtypeProp}\skp(\ref{lem:SubtypeProp-3}) 
for all $i\ele I$ there exists $J'_i \subseteq J_i$ such that 
 $ \begin{array}{rcl}
`k_i &\leq_C& \bigwedge_{j\ele J'_i} `k_{i,j} 
 \end{array} $
and
 $ \begin{array}{rcl}
\bigwedge_{j\ele J'_i} `r_{i,j} &\seq_R& `r_i.
 \end{array} $
Then $`d'_i\prod `k_{i,j}\arrow `r_{i,j} \seq_D `d'_i\prod `k_i\arrow `r_i$, and since $`k'_i \seq_C `k_i$
we conclude that
 \[ \begin{array}{rcl}
(`d'_i\prod `k_{i,j}\arrow `r_{i,j})\prod `d'_i\prod `k'_i &\seq_C& (`d'_i\prod `k_i\arrow `r_i)\prod `d'_i\prod `k_i.
 \end{array} \]
From this, (1
), and (2
) above we get the derivations:
 \[ \begin{array}{c}
\Inf[\App]
	{\Inf[\MAbs]
{\Inf[\seq]
	{\InfBox{\derLmu `G |- {[`a]L}: (`d'_i\prod `k_{i,j}\arrow `r_{i,j})\prod `d'_i\prod `k'_i | `a{:}`d'_i\prod `k'_i,`D }
	}{\derLmu `G |- {[`a]L}: (`d'_i\prod `k_i\arrow `r_i)\prod `d'_i\prod `k_i | `a{:}`d'_i\prod `k'_i,`D }
}{\derLmu `G |- `m`a.{[`a]L}: `d'_i\prod `k'_i\arrow `r_i | `D }
\quad 
\InfBox{ \derLmu `G |- N : `d'_i | `D }
	}{\derLmu `G |- (`m`a.{[`a]L})N : `k'_i\arrow `r_i | `D }
 \end{array}\]
for all $i\ele I$. From these derivations by rule $(\inter)$ we get the type $\bigwedge_{i\ele I} `k'_i\arrow `r_i \seq_D `d$
and we conclude $\derLmu `G |- (`m`a.{[`a]L})N : `d | `D $ by $(\seq)$.

 \item [{$ `m`a[`b]`m`g.\Cmd \red `m`a.\Cmd[`b/`g] $}]
 	If $\derLmu `G |- {`m`a.\Cmd[`b/`g]} : `d | `D $ then by Lem.\skp\ref{gen lemma} we have that for some $I$
	and all $i\ele I$:
 \[ \begin{array}{rcl}
\derLmu `G |- {\Cmd[`b/`g]} : (`k_i\arrow`r_i)\prod `k_i | `a{:} `k'_i,`D & \textrm{and}& 	\bigwedge_{i\ele I} `k'_i\arrow `r_i \leq_D `d.
 \end{array} \]
	Assuming without loss of generality that $`g \notele `D $ we distinguish the cases:

~\kern-10mm $(`a \neq `b):$
By Lem.\skp\ref{lem:substitution}\skp(\ref{lem:substitution-name}) we have
that $\derLmu `G |- {\Cmd} : (`k_i\arrow`r_i)\prod `k_i | `g{:}`k'',`a{:} `k'_i,`D $ where $`D(`b) = `k''$. Then
for all $i \ele I$ there is a derivation $\Der_i$:
 \[ \begin{array}{c}
\Inf[\MAbs]
	{\Inf[\TCmd]
{\Inf[\MAbs]
	{\InfBox{\derLmu `G |- {\Cmd} : (`k_i\arrow`r_i)\prod `k_i | `g{:}`k'',`a{:} `k'_i,`D }}
	{\derLmu `G |- {`m`g.\Cmd} : `k''\arrow `r_i | `a{:} `k'_i,`D }
}
{\derLmu `G |- {[`b] `m`g.\Cmd} : (`k''\arrow `r_i)\prod `k'' | `a{:} `k'_i,`D }
	}
	{\derLmu `G |- {`m `a.[`b] `m`g.\Cmd} : `k'_i\arrow `r_i | `D }
 \end{array}\]
so that we conclude by rules $(\inter)$ and $(\seq)$ using $\bigwedge_{i\ele I} `k'_i\arrow `r_i \leq_D `d$.

~\kern-10mm $(`a = `b) : $
In this case we reason as before but we get 
$\derLmu `G |- {\Cmd} : (`k_i\arrow`r_i)\prod `k_i | `g{:}`k'_i, \ `a{:} `k'_i,`D $
since the type of $`g$ has to be the same than that of $`b = `a$ and
$(`a{:} `k'_i,`D)(`a) = `k'_i$, so that the assumption about $`g$ now differs in each derivation $\Der_i$.
But then we have:
 \[ \begin{array}{c}
\Inf[\MAbs]
	{\Inf[\TCmd]
{\Inf[\MAbs]
	{\InfBox{\derLmu `G |- {\Cmd} : (`k_i\arrow`r_i)\prod `k_i | `g{:}`k'_i, `a{:} `k'_i,`D }}
	{\derLmu `G |- {`m`g.\Cmd} : `k'_i\arrow `r_i | `a{:} `k'_i,`D }
}
{\derLmu `G |- {[`a] `m`g.\Cmd} : (`k'_i\arrow `r_i)\prod `k'_i | `a{:} `k'_i,`D }
	}
	{\derLmu `G |- {`m `a.[`a] `m`g.\Cmd} : `k'_i\arrow `r_i | `D }
 \end{array}\]
and we conclude as in the previous case.
\qed
 \end{description}
 \end{description}
 \end{proof}

We end this section with two examples.

 \begin{exa} \label{SK}
As stated by the last results, we now show that we can assign to $(`lxyz.xz(yz))(`lab.a)$ any type that is assignable to $`lba.a$ since $(`lxyz.xz(yz))(`lab.a) \red^* `lba.a$. We first derive a type for $`lba.a$.
 \[ \begin{array}{c}
\Inf	[\Abs]
	{\Inf	[\Abs]
		{\Inf	[\Ax]
			{ \derLmu a{:}`k\arrow `r,b{:}`w |- a : `k\arrow `r | {} }
		}{ \derLmu b{:}`w |- `la.a : (`k\arrow `r)\prod `k\arrow `r | {} }
	}{ \derLmu {} |- `lba.a : `w\prod (`k\arrow `r)\prod `k\arrow `r | {} }
 \end{array} \]
Let $`G = x{:}(`k\arrow `r)\prod `w\prod `k\arrow `r , y{:}`w, z{:}`k\arrow `r$, then we can derive:
 \[ \begin{array}{c}
\Inf	[\App]
	{\Inf	[\Abs]
		{\Inf	[\Abs]
			{\Inf	[\Abs]
				{
\hbox to 40mm{ \kern-30mm 
				 \Inf	[\App]
					{\Inf	[\App]
						{\Inf	[\Ax]
							{ \derLmu `G |- x : (`k\arrow `r)\prod `w\prod `k\arrow `r | {} }
						\Inf	[\Ax]
							{ \derLmu `G |- z : `k\arrow `r | {} }
						}{ \derLmu `G |- xz : `w\prod `k\arrow `r | {} }
					\Inf	[`w]
						{ \derLmu `G |- yz : `w | {} }
					}{ \derLmu `G |- xz(yz) : `k\arrow `r | {} \hspace*{4cm}}
}
				}{ \derLmu `G\except z |- `lz.xz(yz) : (`k\arrow `r)\prod `k\arrow `r | {} }
			}{ \derLmu `G\except z\except y |- `lyz.xz(yz) : `w\prod (`k\arrow `r)\prod `k\arrow `r | {} } 
		}{ \derLmu {} |- `lxyz.xz(yz) : ((`k\arrow `r)\prod `w\prod `k\arrow `r)\prod `w\prod (`k\arrow `r)\prod `k\arrow `r | {} }
	\Inf	[\Abs]
		{\Inf	[\Abs]
			{\Inf	[\Ax]
				{ \derLmu a{:}`k\arrow `r,b{:}`w |- a : `k\arrow `r | {} }
			}{ \derLmu a{:}`k\arrow `r |- `lb.a : `w\prod `k\arrow `r | {} }
		}{ \derLmu {} |- `lab.a : (`k\arrow `r)\prod `w\prod `k\arrow `r | {} }
	}{ \derLmu {} |- (`lxyz.xz(yz))(`lab.a) : `w\prod (`k\arrow `r)\prod `k\arrow `r | {} } 
 \end{array} \]

 \end{exa}

 \begin{exa} \label{ex:mu-self-app}
Consider the reduction $(`m`a.[`a]x)x \red `m `a.([`a]x) [`a \Becomes x] \same `m `a.[`a]xx$.
This last term is not a proof term in the sense of Parigot, but of interest here since typing the self application $xx$ is a characteristic of intersection type systems. 
Let $`d\ele \Lang_D$ be arbitrary. 
Now we have:
 \[ \begin{array}{c} \kern-35mm
\Inf	[\App]
	{\Inf	[\MAbs]
{\Inf	[\TCmd]
	{\Inf	[\seq]
{\Inf	[\Ax]
	{ \derLmu x{:}`d\inter (`d\prod`k\arrow `r) |- x : `d\inter (`d\prod`k\arrow `r) | `a{:}`d\prod`k }
}{ \derLmu x{:}`d\inter (`d\prod`k\arrow `r) |- x : `d\prod`k\arrow `r | `a{:}`d\prod`k }
	}{ \derLmu x{:}`d\inter (`d\prod`k\arrow `r) |- [`a] x : (`d\prod`k\arrow `r)\prod (`d\prod`k) | `a{:}`d\prod`k }
}{ \derLmu x{:}`d\inter (`d\prod`k\arrow `r) |- `m`a.[`a] x : `d\prod`k\arrow `r | {} }
\raise2\RuleH\hbox to 20mm{\kern-17mm
\Inf	[\seq]
	{\Inf	[\Ax]
{ \derLmu x{:}`d\inter (`d\prod`k\arrow `r) |- x : `d\inter (`d\prod`k\arrow `r) | `a{:}`d\prod`k }
	}{ \derLmu x{:}`d\inter (`d\prod`k\arrow `r) |- x : `d | `a{:}`d\prod`k }
}
\multiput(-10,-1)(0,5){6}{.}
	}{ \derLmu x{:}`d\inter (`d\prod`k\arrow `r) |- (`m`a.[`a] x)x : `k\arrow `r | {} }
 \end{array} \]
as well as:
 \[ \begin{array}{c}
\Inf	[\MAbs]
	{\Inf	[\TCmd]
{\Inf	[\App]
	{\Inf	[\seq]
{\Inf	[\Ax]
	{ \derLmu x{:}`d\inter (`d\prod`k\arrow `r) |- x : `d\inter (`d\prod`k\arrow `r) | `a{:}`k }
}{ \derLmu x{:}`d\inter (`d\prod`k\arrow `r) |- x : `d\prod`k\arrow `r | `a{:}`k }
	\quad
	 \Inf	[\seq]
{\Inf	[\Ax]
	{ \derLmu x{:}`d\inter (`d\prod`k\arrow `r) |- x : `d\inter (`d\prod`k\arrow `r) | `a{:}`k }
}{ \derLmu x{:}`d\inter (`d\prod`k\arrow `r) |- x : `d | `a{:}`k }
	}{ \derLmu x{:}`d\inter (`d\prod`k\arrow `r) |- xx : `k\arrow `r | `a{:}`k}
}{ \derLmu x{:}`d\inter (`d\prod`k\arrow `r) |- [`a]xx : (`k\arrow `r)\prod `k | `a{:}`k }
	}{ \derLmu x{:}`d\inter (`d\prod`k\arrow `r) |- `m`a.[`a]xx : `k\arrow `r | {} }
 \end{array} \]
Observe that the `cut type' in the first derivation, $`d\prod`k$ (appearing twice in the type $(`d\prod`k\arrow `r)$ $\prod (`d\prod`k)$ of the premise of rule $(\MAbs)$), differs from the cut type $`k$ in $(`k\arrow `r)\prod `k$ occurring in the premise of $(\MAbs)$ of the second derivation; indeed the latter is of a smaller size than the former.

A similar but simpler derivation can be obtained from the previous one in the case of the reduction $(`m`a.[`a]x)y \red `m `a.([`a]x) [`a \Becomes y] \same `m `a.[`a]xy$, where there is no self-application: indeed we have that
$ \derLmu x{:}`d\prod`k\arrow `r,y{:}`d |- (`m`a.[`a]x)y : `k\arrow `r | {} $ and
$ \derLmu x{:}`d\prod`k\arrow `r,y{:}`d |- `m`a.[`a]xy : `k\arrow `r | {} $ are derivable in a very similar manner.

 \end{exa}

 \section{Characterisation of Strong Normalisation} \label{sec:character}

One of the main results for $`l`m$, proved in \cite{Parigot97}, states that all $`l`m$-terms that correspond to proofs of second-order natural deduction are strongly normalising; the reverse of this property does not hold for Parigot's system, since there, for example, not all terms in normal form are typeable.

The full characterisation of strong normalisation ($M$ is strongly normalising if and only if $M$ is typeable in a given system) for the $`l$-calculus is a property that is shown for various intersection systems (see \cite{Bar2013}, Sect.~17.2 and the references there). So it is a natural question whether there exists a similar characterisation of strongly normalising $`l`m$-terms in the present context, by suitably restricting the system in Sect.\skp\ref{sec:types}.
We answer this question here; the proof is a revised version of that in \cite{BakBdL-ITRS12}, obtained by a simplified type syntax and by just restricting the full system, instead of considering one of its variants. 

To simplify the technical treatment we shall ignore the structural reduction rule
$(\Rename)$; in fact the set of strongly formalisable terms remains the same, no matter whether this rule is considered or not.

We establish the relation between our result to the one for Parigot's system in the next section.

 \subsection{The restricted type system} \label{sec:system}

For the untyped $`l$-calculus, the characterisation of strong normalisation states that a $`l$-term is strongly normalisable if and only if it is typeable in a restricted system of intersection types, where $`w$ is not admitted as a type and consequently the rule $(`w)$ is not part of the system.
Alas a straightforward extension of this result does not hold for the $`l`m$-calculus, at least with respect to the system presented in this paper. 
This is due to the fact that the natural interpretation of a type $`k = `d_1\prod \cdots`d_k\prod `w $ (for $k>0$) is the set of continuations whose leading $k$ elements are in the denotations of $`d_1,\ldots,`d_k$; since continuations are infinite tuples, the ending $`w$ represents the lack of information about the remaining infinite part. 
Therefore these occurrences of $`w$ cannot be simply deleted without substantially changing the semantics of the type system and questioning the soundness of its rules.

We solve the problem of restricting the type assignment system to the extent of typing strongly normalising terms here by defining a particular subset of the type language. 
There the type $`w$ is allowed only in certain harmless positions such that its meaning becomes just the universe of terms we are looking for, \ie~the strongly normalising ones.
This amounts to restricting the sets of types $\Lang_D$ and $\Lang_C$ to those having $`w $ 
only as the final part of product types. 
We shall then suitably modify the standard interpretation of intersection types, adapting Tait's computability argument.

For what concerns the atomic types, a single constant type $`y$ suffices for our purposes. Therefore restricted types are of two sorts instead of three.

 \begin{defi}[Restricted Types and Pre-order] \label{def:restrTypes}

 \begin{enumerate}

 \firstitem
The sets $\LangR_D$ of \emph{(restricted) term types} and $\LangR_C$ of \emph{(restricted) continuation types} are
defined inductively by the following grammars, where $`y$ is a type constant:
 \[ \begin{array}{l@{\quad}rcl}
\LangR_D : & `d & :: =& `k\arrow `y \mid `d\inter `d \\
\LangR_C : & `k & :: = & `w\mid`d\prod `k \mid `k \inter `k
 \end{array} \]

 \item
We define the set $\LangR$ of \emph{(restricted) types} as $\LangR = \LangR_D \union \LangR_C $ and the relations $\seqRestrA$ 
over $\LangR_A$ (for $A=D,C$) as 
the pre-order induced by the least intersection type theories $\TRestr_A$ such that:
 \[ \begin{array}{c@{\dquad}c@{\dquad}c@{\dquad}c}
\Inf {`k \seqRestrC `w}
 &
\Inf	{(`d _1\prod `k _1) \inter (`d _2\prod `k _2) \seqRestrC (`d _1\inter`d _2)\prod ( `k _1\inter `k _2) }
 &
\Inf	{`d_1\seqRestrD `d _2 \quad `k _1\seqRestrC `k _2 }{ `d _1\prod `k _1 \seqRestrC `d _2\prod `k _2 }
 &
\Inf	{`k _2 \seqRestrC `k _1 }{ `k _1\arrow `y \seqRestrD `k _2\arrow `y } 
 \end{array} \]
 \[ \begin{array}{c@{\dquad}c@{\dquad}c@{\dquad}c}
\Inf	{`s \inter `t \seqRestrA `s }
 &
\Inf	{`s \inter `t \seqRestrA `t }
 &
\Inf	{`s \seqRestrA `t _1 \quad `s \seqRestrA `t _2 }{ `s \seqRestrA `t _1\inter `t _2}
 &
(A = D,C)
 \end{array} \]

 \item
We define the \emph{length} of a continuation type, $ \length{`.} : \LangR_C \To \Nat$, as follows:
 \[ \begin{array}{rcl}
\length{`w} &=& 0 \\ 
\length{`d\prod `k} &=& 1+ \length{`k} \\ 
\length{`k_1 \inter`k_2} &=& \max (\length{`k_1} , \length{`k_2})
 \end{array} \]
 \Ugo{Where length is used? I have put this part into comment}
\Steffen{In Theorem {thm:typableAreSN}; a simple find reveals this.}
 \end{enumerate}
 \end{defi}

By definition, we have that $\LangR_D \subseteq \Lang_D$ and $\LangR_C \subseteq \Lang_C$. 
All the rules axiomatising $\seqRestrA$ are instances of the rules axiomatising $\seqD$ and $\seqC$ in Def.\skp\ref{def:typeTHeorFromDomain} and \ref{def:typeTheories}, 
hence ${\seqRestrA} \subseteq {\seqA}$ for $A=D,C$; in other words, the theories $\TypeTheor_A$ can be seen as extensions of the respective theories $\TRestr_A$. 
It is natural to ask whether $\TypeTheor_A$ (for $A=D,C$) is conservative with respect to $\TRestr_A$.
This is not obvious: in a derivation of $`s \seqA `t$ in the formal theory $\TypeTheor_A$, even if $`s,`t \ele \LangR_A$, one could have used a type $`s'\notele \LangR_A$ and the transitivity rule with premises $`s\seqA`s'\seqA`t$, which cannot be derived in the formal presentation of $\seqRestrA$.
For example, consider the inequalities:
 \[ \begin{array}{rcl@{\quad}lc}
`d_1\prod `d_2\prod `w &\seqC& `d_1\prod `w\prod `w &=_C& `d_1\prod `w,
 \end{array} \]
where the axiom $`w\prod `w =_C `w$ is used (see Def.\skp\ref{def:typeTheories}).
If $`d_1,`d_2 \ele \LangR_D$ then both $`d_1\prod `d_2\prod `w$ and $`d_1\prod `w$ are in $\LangR_C$, but $ `d_1\prod `w\prod `w\notele \LangR_C$ because $`w\notele \LangR_D$. 

However this is not a counterexample, since $`d_2\prod `w\seqRestr `w$ which implies that $`d_1\prod `d_2\prod `w\seqRestr `d_1\prod `w$ is derivable in $\TRestr_C$.
As a matter of fact, we can show that $`s \seqRestr `t$ if and only if $`s \seq `t$ for any $`s,`t \ele \LangR$ by a semantic argument. 

 \begin{lem} \label{lem:pre-order-restriction}
Take $R = \Set{\bot \po \top}$ and set $\TypeToComp_R(`y) = \top$.
Then for all $`s,`t \ele \LangR_A$, where $A=D,C$, if $\TypeToComp_A(`t) \po \TypeToComp_A(`s)$ then $`s \seqRestrA `t$.
 \end{lem}

 \begin{proof}
By induction over the structure of types. 
 \begin{description}

 \item [$ `s,`t \ele \LangR_D $] 
Let $ `s = `d$, and $`t=`d'$.
Then $`d = \Inter_{i \ele I}(`k_i\arrow `y)$ and $`d' = \Inter_{j \ele J}(`k'_j\arrow `y)$; hence:	
 \[ \begin{array}{lll}
\TypeToComp_D(`d') = \join_{j \ele J} \StepFun{\TypeToComp_C(`k'_j) }{ \top} \po \join_{i \ele I} \StepFun{\TypeToComp_C(`k_i) }{ \top} = \TypeToComp_D(`d) & \Then & (\textrm{by hypothesis}) \\ 
\Forall j \ele J\ \Exists i_j \ele I \Pred[ \TypeToComp_C(`k_{i_j}) \po \TypeToComp_C(`k'_j) ] & \Then & (\textrm{*}) \\ 
\Forall j \ele J\ \Exists i_j \ele I \Pred[ `k'_j \seqRestrC `k_{i_j} ] & \Then & (\textrm{by induction}) \\ 
\Forall j \ele J\ \Exists i_j \ele I \Pred[ `k_{i_j}\arrow `y \seqRestrD `k'_j\arrow `y ] & \Then & \\
`d \seqRestr \Inter_{j \ele J} `k_{i_j}\arrow `y \seqRestrD \Inter_{j \ele J}(`k'_j\arrow `y) = `d'
 \end{array} \]
where (*) follows by contraposition: if for some $j_0 \ele J$ we had 
$\TypeToComp_C(`k_i) \notpo \TypeToComp_C(`k'_{j_0})$ for all $i \ele I$ then 
$\join_{j \ele J} \StepFun{\TypeToComp_C(`k'_j) }{ \top}(\TypeToComp_C(`k'_{j_0})) = \top$ and 
$\join_{i \ele I} \StepFun{\TypeToComp_C(`k_i) }{ \top} (\TypeToComp_C(`k'_{j_0})) = \bot$.

 \item [$ `s, `t \ele \LangR_C $] 
Let $ `s = `k$, $`t=`k'$. 
Since the ordering over $C = D \times D \times \cdots$ is component-wise, 
 \[ \begin{array}{rcl}
\Tuple <d_1, d_2, \ldots > \join \Tuple <d'_1, d'_2, \ldots > &=& \Tuple <d_1\join d'_1, \, d_2\join d'_2, \, \ldots >,
 \end{array} \] 
and $\TypeToComp_C(`k \inter`k') = \TypeToComp_C(`k) \join \TypeToComp_C(`k')$, we can assume that $`k = `d_1\prod \cdots\prod `d_n\prod `w$ and $`k' = `d'_1\prod \cdots\prod `d'_m\prod `w$. 
Hence the hypothesis $\TypeToComp_C(`k) \po \TypeToComp_C(`k')$ reads as
 \[ \begin{array}{rcl@{\quad}lc}
\Tuple <\TypeToComp_D(`d_1), \ldots, \TypeToComp_D(`d_n), \bot, \ldots >
	 &\po& 
\Tuple <\TypeToComp_D(`d'_1), \ldots, \TypeToComp_D(`d'_m), \bot, \ldots > 
 \end{array} \]
where $\bot, \ldots $ stands for infinitely many $\bot$s; then $\TypeToComp_D(`d_1) \po \TypeToComp_D(`d'_i)$ for all $i \seq \min(m,n)$. 
It is easy to see that $\TypeToComp_D(`d) \neq \bot$ for any $`d \ele \LangR_D$, and therefore $n\seq m$.
Then by induction, $`d_i \seqRestrD `d'_i$ for all $i\seq n$. 
Now $`d_{n+1}\prod \cdots\prod `d_m\prod `w\seqRestrC`w$ by definition, hence 
 \[ \begin{array}{rcl}
`d_1\prod \cdots\prod `d_n\prod `d_{n+1}\prod \cdots\prod `d_m\prod `w
	&\seqRestrC& 
`d'_1\prod \cdots\prod `d'_m\prod `w
 \end{array} \]	
\arrayqed[-30pt]

 \end{description}

 \end{proof}

 \begin{thm} \label{cor:pre-order-restriction}
The pre-orders $\seqRestrD$ and $\seqRestrC$ are the restriction to $\LangR$ of $\seqD$ and $\seqC$, respectively.
 \end{thm}

 \begin{Proof}
Let $`s,`t \ele \LangR_A$ for either $A=D,C$, then:
 \[ \begin{array}{rcl@{\quad}l}
 `s \seqA `t & \Then & \Sem{`s}^A \subseteq \Sem{`t}^A & (\textrm{Cor.\skp\ref{cor:leq}}) \\ 
 & \Then & \filt_A \TypeToComp_A(`s) \subseteq \filt_A \TypeToComp_A(`t) & (\textrm{Lem.\skp\ref{lem:typeInterpCompactCone}}) \\ 
 & \Then & \TypeToComp_A(`t) \po \TypeToComp_A(`s) & (\textrm{since }\TypeToComp_A(`s) \ele \filt_A \TypeToComp_A(`t) ) \\ 
 & \Then & `s \seqRestrA `t & \textrm{(Lem.\skp\ref{lem:pre-order-restriction})} 
 \end{array} \]
Since trivially ${\seqRestrA} \subseteq {\seqA}$, this establishes the thesis. 
 \end{Proof}

 \begin{defi}[Restricted Bases, Contexts, Judgments, and Type Assignment] 
 \label{def:restrictedTypingSystem}

 \begin{enumerate}

 \firstitem
A \emph{restricted} \\ \emph{basis} is a basis $`G$ such that $`d \ele \LangR_D$ for all $x{:}`d \ele `G$. 
Similarly,
a \emph{restricted name context} is a context $`D $ with $`k \ele \LangR_C$ for all $`a{:}`k \ele `D $. 
Finally, for $T \ele \Terms \union \Commands$ we say that $ \derLmu `G |- T : `s | `D $ is a \emph{restricted judgement} if $`s \ele \LangR$ and
$`G$ and $`D $ are a restricted basis and a restricted name context respectively.

 \item
The restricted judgement $ \derLmu `G |- T : `s | `D $ is \emph{derivable in the restricted typing system}, written $ \derLmuR `G |- T: `s | `D $, if it is derivable 
in the system of Def.\skp\ref{def:intersTypeAss} without using rule $(`w)$, and all the judgements in the derivation are restricted.

 \end{enumerate}
 \end{defi}

Since the restricted system is just the intersection type system of Section\,\ref{sec:types}, where types occurring in judgements are restricted, and rule $(`w)$ is disallowed, we can use results from the previous section in proofs. 
However care is necessary, since the lack of rule $(`w)$ invalidates expansion property, as we illustrate in Sect.\skp\ref{sub:NormType}. 

We only observe that, while it is clear that in the restricted system no term can have type $`w$, this is still the case for commands, because $`w \ele \LangR_C$ and we have subsumption rule in the system. However judgements of the shape $\derLmuR `G |- \Cmd : `w | `D $ cannot occur in any derivation deducing a type for a term. 
In fact, if the subject $M$ of the conclusion is a term including the command $\Cmd$ then,
for some $`a$, $`m`a.\Cmd$ must be a subterm of $M$; but rule $(\MAbs)$, which is the only applicable rule, doesn't admit $\derLmuR `G |- \Cmd : `w | `D $ as a premise.

 \subsection{Typability implies Strong Normalisation} \label{subsec:type-SN}

In this subsection we will show that -- as can be expected of a well-defined notion of type assignment that does not type recursion and has no general rule that types all terms -- all typeable terms are strongly normalising.

For the full system of Def.\skp\ref{def:intersTypeAss} this is not the case. In fact, by means of types not allowed in the restricted system, it is possible to type the fixed-point constructor $`lf.(`lx.f(xx))(`lx.f(xx))$ in a non-trivial way, as shown by the following derivation:

 \begin{exa}
The fixed-point combinator is typeable in the system of Def.\skp\ref{def:intersTypeAss}:
 \[ \begin{array}{c}
\Inf	[\Abs]
	{\Inf	[\App] 
{\Inf	[\Abs] 
	{\Inf	[\App]
{\Inf	[\Ax]
	{\derLmu f{:}`w\prod `w\arrow `y,x{:}`w |- f : `w\prod `w\arrow `y | {} } 
 \quad 
 \Inf	[`w]
	{\derLmu f{:}`w\prod `w\arrow `y,x{:}`w |- xx : `w | {} }
}{ \derLmu f{:}`w\prod `w\arrow `y,x{:}`w |- f(xx) : `w\arrow `y | {} }
	}{ \derLmu f{:}`w\prod `w\arrow `y |- `lx.f(xx) : `w\prod `w\arrow `y | {} } \kern-25mm 
 \Inf	[`w]	
	{\derLmu f{:}`w\prod `w\arrow `y |- `lx.f(xx) : `w | {} } 
}{ \derLmu f{:}`w\prod `w\arrow `y |- (`lx.f(xx))(`lx.f(xx)) : `w\arrow `y | {} }
	 }{ \derLmu |- `lf.(`lx.f(xx))(`lx.f(xx)) : (`w\prod `w\arrow `y)\prod `w\arrow `y | {} }
 \end{array} \]
Notice that this term does not have a normal form, so is not strongly normalisable.
 \end{exa}

We start by showing that if a term is typeable in the restricted system then it is strongly normalising. 
We adapt Tait's computability argument and the idea of saturated sets to our system (see \cite{Krivine-book'93}, Ch.\,3).

 \begin{defi}[Term Stacks]
The set $\Stacks$ of (finite) \emph{term stacks}, whose elements we shall denote by $ \vect{L}$, is defined by
the following grammar:
 \[ \begin{array}{l@{\quad}rcl}
\Stacks : & \vect{L} & :: =& \epsilon \mid L \cons \vect{L} 
 \end{array} \]
where $\epsilon$ denotes the empty stack and $L \ele\Terms$.
Moreover, we define \emph{stack application} as follows:
 \[ \begin{array}{l@{\quad}rl}
M\,\epsilon & \ByDef & M \\
M \, (P\cons \vect{L} ) & \ByDef & (MP) \, \vect{L}
 \end{array} \]
So, if $\vect{L} \equiv L_1\cons \dots \cons L_k\cons\epsilon$ we have $M \, \vect{L} \equiv M L_1 \cdots L_k$. 
We extend the notion of structural substitution to stacks as follows:
 \[ \begin{array}{l@{\quad}rl}
T \strSub[ `a <= \epsilon ] & \ByDef & T \\
T \strSub[ `a <= P \cons \vect{L} ] & \ByDef & (T \strSub[ `a <= P ] ) \strSub[ `a <=\vect{L} ] 
 \end{array} \]
for $T \ele \Terms \union \Commands$, when each $L_i$ does not contain $`a$.
 \end{defi}
We normally omit the trailing $\epsilon$ of a stack. 
Notice that
 \[ \begin{array}{lclclcl}
[`a]M \strSubvec[ `a <= L ] 
 & \equiv &
[`a]M \strSub[ `a <= L_1 ] \strSub[ `a <= L_2 ] \dots \strSub[ `a <= L_n ] 
 & \equiv &
[`a](M\strSub[ `a <= {\vect{L}} ]) \vect{L} 
 \end{array} \]

The notion of string normalisation is formally defined as:

 \begin{defi} 
The set $\SN$ of terms that are \emph{strongly normalisable} is the set of all terms $M \ele \Terms$ such that no infinite reduction sequence out of $M$ exists; we write $\SN(M)$ for $M \ele \SN$, and $ \SN^*$ for the set of finite stacks of terms in $ \SN$, and write $\SN(\vect{L})$ if $\vect{L} \ele \SN^*$.
 \end{defi}

The following property of strong normalising terms is straightforward:

 \begin{prop} \label{SN facts}

 \begin{enumerate}

 \firstitem \label{SN fact head application}
If $\SN(x\vect{ M })$ and $\SN(\vect{N})$, then $\SN(x\vect{ M }\vect{N})$.

 \item \label{SN fact redex}
If $\SN(M[N/x]\vect{P})$ and $\SN(N)$, then $\SN({( `lx.M)N\vect{P}})$.

 \item \label{SN fact add mu abstraction}
If $\SN(M)$, then $\SN(`m`a.[`b]M)$.

 \item \label{SN fact mu in redex}
If $\SN( `m`a.[`b]M \strSubvec[ `a <= N ] \vect{L} ) $ and $ \SN(\vect{N}) $, then $ \SN({( `m`a.[`b]M )\vect{N}\vect{L}}) $.

 \item \label{SN fact mu out redex}
If $\SN(`m`a.[`a]M \strSubvec[ `a <= N ] \vect{N} \vect{L} ) $, then $ \SN({( `m`a.[`a]M )\vect{N}\vect{L}}) $.

 \end{enumerate}
 \end{prop}

 \begin{defi}[Type Interpretation] \label{TypeSem definition}
We define a pair of mappings 
 \[ \TypeSem{\cdot} = \Pair{\TypeSem{\cdot}_1}{\TypeSem{\cdot}_2}: (\LangR_D \To \wp(\Terms)) \times (\LangR_C \To \wp(\Stacks)) \] 
interpreting types as sets of terms and stacks. 
Writing $\TypeSem{\delta} = \TypeSem{\delta}_1$ and $\TypeSem{\kappa} = \TypeSem{\kappa}_2$, the definition is 
as follows:
 \[ \begin{array}{rcl@{\quad}l}
 \TypeSem{`k\arrow `y} &=& \Set{T \mid \Forall \vect{L} \ele \TypeSem{`k } \Pred[\SN(T \vect{L}) ] } 
 \\ 
 \TypeSem{`w} &=& \SN^* 
	\\
 \TypeSem{`d\prod `k } &=& \Set{N \cons \vect{L} \mid N \ele \TypeSem{`d} \And \vect{L} \ele \TypeSem{`k }} 
 \\ 
 \TypeSem{`s \inter `t} &=& \TypeSem{`s} \cap \TypeSem{`t} & (`s,`t \ele \LangR_A,~ A = D,C). 
 \end{array} \]

 \end{defi}

We will now show that the interpretation of a type is a set of strongly normalisable terms and that neutral terms (those starting with a variable) are in the interpretation of any type.

 \begin{lem} \label{TypeSem and SN lemma}
For any $ `d \ele\LangR_D$ and $`k \ele\LangR_C$:
 \begin{enumerate}

 \item \label{TypeSem implies SN}
$\TypeSem{`d} \subseteq \SN $ and $\TypeSem{`k } \subseteq \SN^* $.

 \item \label{SN head implies TypeSem}
$x\vect{N} \ele \SN \Then x\vect{N} \ele \TypeSem{`d}$.

 \end{enumerate}
 \end {lem}

 \begin{proof}
By induction on the structure of types.

 \begin{enumerate}

 \item 
Let $`d = `k\arrow `y$ and $M \ele \TypeSem{`d}$: 
then, for any $\vect{L} \ele \TypeSem{`k}$, by definition, $\SN(M \vect{L})$, so in particular $\SN(M)$.
The case $`k=`w$ follows by definition; the case $`k = `d\prod `k'$ follows by induction, since $\TypeSem{`d} \subseteq \SN$ and $\TypeSem{`k' }\subseteq \SN^*$. 
The cases $`d = `d_1\inter`d_2$ and $`k = `k_1\inter`k_2$ follow immediately by induction.

 \item 
Let $x\vect{N} \ele \SN$ and $`d = `k\arrow `y$. 
If $\vect{L} \ele \TypeSem{`k }$ then $\vect{L} \ele \SN^*$ by point (\ref{TypeSem implies SN}), so that $x\vect{N}\vect{L} \ele \SN$ by observing that the only possible reductions are inside the components of $\vect{N}$ and $\vect{L}$, which are in $\SN$ by assumption.
Then by definition $ x\vect{N} \ele \TypeSem{`d}$.
 \qed

 \end{enumerate}
 \end{proof}

Observe that $`y\notele \LangR_D$, and therefore we do not have the clause $\TypeSem{`y} = \SN$ in Def.\skp\ref{TypeSem definition}, as was the case in \cite{BakBdL-ITRS12}. 
This clause would be consistent with the previous definition, but having $`y \ele \Lang_D$ (the unrestricted language of term types) enforces the equation $`y =_D `w\arrow `y$ which is false in the above interpretation.
In fact, $`l x.xx \ele \SN = \TypeSem{`y}$, so that $(`l x.xx)\cons \epsilon \ele \SN^* = \TypeSem{`w}$, but $(`l x.xx)((`l x.xx)\cons \epsilon) \equiv (`l x.xx)(`l x.xx) \notele \SN$, and therefore $\TypeSem{`y}\not\subseteq \TypeSem{`w\arrow `y}$. 
In \cite{BakBdL-ITRS12} we managed to avoid this incoherence by ruling out $`w$ from $\LangR_C$ and by interpreting types $`k\arrow `y$ differently according to the shape of $`k$; in fact, in that paper $\TypeSem{`k\arrow `y}$ is the set of representable functions from $\SN^*$ to $\SN$ only when $`k \neq`w$, while $\TypeSem{`w\arrow `y}$ is just $\SN$.

We will now show that our type interpretation respects type inclusion.

 \begin{lem} \label{closed for leq}
For all $`s, `t \ele \LangR $:~
$`s \seqRestr `t \Implies \TypeSem{`s} \subseteq \TypeSem{`t}$.
 \end{lem}

 \begin{Proof}
By easy induction over the rules in Def.\skp\ref{def:restrTypes}.
For $`k \seqRestrC `w$, notice that $\TypeSem{`k } \subseteq \SN^*= \TypeSem{`w}$ by Lem.\skp\ref{TypeSem and SN lemma}\,(\ref{TypeSem implies SN}). If $`k\seqRestrC `k'$, then $`k = `d_1\prod \cdots\prod `d_n\prod `w$, $`k' = `d'_1\prod \cdots\prod `d'_m\prod `w$, $m \seq n$, and $`d_i\seqRestrD`d'_i$.
By induction, $\TypeSem{`d_i} \subseteq \TypeSem{`d'_i}$ for all $i\seq m$; notice that $\TypeSem{`d_{i+1}\prod \cdots\prod `d_n\prod `w}$ $\subseteq \TypeSem{`w} =\SN^*$ by Lem.\skp\ref{TypeSem and SN lemma}\,(\ref{TypeSem implies SN}), so we can conclude $\TypeSem{`k }\subseteq\TypeSem{`k' }$.
Assume $`k _1\arrow `y \seqRestrD `k _2\arrow `y$ because $`k _2 \seqRestrC `k _1$, then by induction $\TypeSem{`k_2}\subseteq\TypeSem{`k_1}$.
Assume now $M \ele \TypeSem{`k _1\arrow `y}$, then by definition for all $\vect{L} \ele \TypeSem{`k_1}$, we have $M\vect{L} \ele \SN$. 
Since $ \TypeSem{`k_2} \subseteq\TypeSem{`k_1}$, also for all $\vect{L} \ele \TypeSem{`k_2}$ we have $M\vect{L} \ele \SN$, and therefore also $M \ele \TypeSem{`k _2\arrow `y}$.
 \end{Proof}

The next lemma states that our type interpretation is closed under expansion for the logical and for the structural reduction, with the proviso that the term or stack to be substituted is an element of an interpreted type as well.

 \begin{lem} \label{Typesem expansion} 
For any $`d, `d' \ele\LangR_D$ and $`k \ele\LangR_C$:
 \begin{enumerate}

 \item \label{Typesem expansion logical}
If $ M[N/x]\vect{P} \ele \TSem{`d} $ and $ N \ele \TSem{`d'} $, then $ ( `lx.M)N\vect{P} \ele \TSem{`d}$.

 \item \label{Typesem expansion structural in}
If $ (`m`a.[`b]M \strSubvec [ `a <= N ] )\vect{P} \ele \TSem{`d} $ and $\vect{N} \ele \TSem{`k' }$, then $ (`m`a.[`b]M) \vect{N} \vect{P} \ele \TSem{`d} $.

 \item \label{Typesem expansion structural out}
If $ (`m`a.[`a]M \strSubvec [ `a <= N ] \vect{N}) \vect{P} \ele \TSem{`d} $, then $ (`m`a.[`a]M) \vect{N} \vect{P} \ele \TSem{`d} $.

 \end{enumerate}
 \end{lem}

 \begin{Proof} 

By induction on the structure of types.
If $`d = `k\arrow `y$, then from $M[N/x]\vect{P} \ele \TSem{`k\arrow `y}$ by definition we have $ \SN(M[N/x]\vect{P}\vect{Q})$ for all $\vect{Q} \ele \TSem{`k} $ and from $N \ele \TSem{`d'}$ that $\SN(N)$ by Prop.\skp\ref{TypeSem and SN lemma}\,(\ref{TypeSem implies SN}).
Then by Lem.\skp\ref {SN facts}\,(\ref{SN fact redex}) also $ \SN((`lx.M)N\vect{P}\vect{Q}) $, so by definition $(`lx.M)N\vect{P} \ele \TSem{`k\arrow `y}$.
The case $`d = `d_1\inter`d_2$ follows directly by induction.

The second and third case follow similarly, but rather using Prop.\skp\ref {SN facts}\,(\ref{SN fact mu in redex}) and \ref {SN facts}\,(\ref{SN fact mu out redex}), respectively.
 \end{Proof}

In Theorem\,\ref{thm:typableAreSN} we will show that all typeable terms are strongly normalisable.
In order to achieve that, in Lem.\skp\ref{lem:replacement} we will first show that for a term $M$ typeable with $`d$, any substitution instance $M \indexRepl $ (\ie~replacing all free term variables by terms, and feeding stacks to all free names) is an element of the interpretation of $`d$, which by Lem.\skp\ref{TypeSem and SN lemma} implies that $M \indexRepl $ is strongly normalisable.
We need these substitutions to be applied all `in one go', so define a notion of parallel substitution.
The main result is then obtained by taking the substitution that replaces term variables by themselves and names by stacks of term variables.
The reason we first prove the result for \emph{any} substitution is that, in the proof of Lem.\skp\ref{lem:replacement}, in the case for $`lx.M$ and $`m`a.\Cmd$ the substitution is extended, by replacing the bound variable or name with a normal term (or stack).

 \begin{defi} \label{Repl definition}
 \begin{enumerate}

 \firstitem
A pair of partial mappings $ \Repl =\Pair{\Repl_1}{\Repl_2}: (\TVar\To\Terms) \times (\CVar\To\Terms^*)$, where we simply write $\Repl(x) = \Repl_1(x)$ and
$\Repl(`a) = \Repl_2(`a)$,
is a \emph{parallel substitution} if, for every $\textsl{p},\textsl{q} \ele \dom (\Repl) $, if $\textsl{p}\not=\textsl{q}$ then $\textsl{p} \notele \fv(\Repl \textsl{q})$ and $\textsl{p} \notele \fn(\Repl \textsl{q})$.

 \item
Borrowing a notation for valuations, for a parallel substitution $\Repl$ we define the application of $\Repl$ to a term or a command by:
\[ \begin{array}{rcl@{\quad}l}
x \indexRepl & \ByDef & N & (\textrm{if } \Repl x = N) 
	\\
y \indexRepl & \ByDef & y & (\textrm{if } y \notele \dom(\Repl) )
	\\
(`l x.M) \indexRepl & \ByDef & `l x.M \indexRepl 
	\\
(MN) \indexRepl & \ByDef & M \indexRepl N \indexRepl \\
([`a]M) \indexRepl & \ByDef & M \indexRepl \cons (\Repl `a) \\
(`m`a .[`a]M) \indexRepl &\ByDef & `m`a .[`a] M \indexRepl \, (\Repl \,`a) \\
 (`m`a .[`b]M) \indexRepl &\ByDef & `m`a .[`b] M \indexRepl & (\textrm{if } `a \neq `b) \\ 
 \\ 
 \end{array} \]

 \item
We define $ \Repl [ N/x ]$ and $ \Repl \strSubvec[ `a <= L ] $ as, respectively,
 \[ \begin{array}{rcl}
 \Repl [ N/x ] \, y &\ByDef& 
 \begin{cases}{l@{\quad}l}
N & (\textrm{if } y=x) \\
 \Repl \, y & (\textrm{otherwise})
 \end{cases}
 \\ [5mm]
 \Repl \strSubvec[ `a <= L ] \, `b &\ByDef& 
 \begin{cases}{l@{\quad}l}
 \vect{L} & (\textrm{if } `a=`b) \\
 \Repl \, `b & (\textrm{otherwise})
 \end{cases}
 \end{array} \]

 \item
We say that $\Repl$ \emph{extends} $`G$ \emph{and} $`D $, if, for all $x{:}`d \ele `G$
and $`a{:}`k \ele `D $, we have, respectively, $\Repl (x) \ele \TypeSem{`d}$ and $\Repl (`a) \ele \TypeSem{`k}$.

 \end{enumerate}
 \end{defi}


 \begin{lem}[Replacement Lemma] \label{lem:replacement}
Let $ \Repl $ be a parallel substitution that extends $`G$ and $`D $.
Then:
 $ \begin{array}{@{}rcl}
 \textit{if }~ \derLmuR `G |- T : `s | `D &\textit{then}& T \indexRepl \ele \TypeSem{`s}. 
 \end{array} $
 \end{lem}

 \begin{proof}
By induction on the structure of derivations. 
We show some more illustrative cases. 

 \begin{description} \itemsep 4pt

 \item[$\Ax$] then $M = x$.
Since $x{:}`d \in`G$ and $\Repl $ extends $`G$, we have $x_{\Repl} = \Repl (x) \ele \TypeSem{`d}$.

 \item [$\Abs$] Then $M = `lx.M'$, $`d = `d'\prod `k\arrow`n $, and $ \derLmu `G,x{:}`d' |- M' : `k\arrow`n | `D $.
Take $N \ele \TypeSem{`d'}$; since $x$ is bound, by Barendregt's convention we can assume that it does not occur free in the image of $ \Repl $, so $ \Repl [ N/x ] $ is a 
parallel substitution that extends $`G, x{:}`d'$ and $`D$.
Then by induction, we have $M'{\indexRepl [ N/x ]} \ele \TypeSem{`k\arrow`n}$.
Since $x$ does not occur free in the image of $ \Repl $, $M'{\indexRepl [ N/x ] } = M'{\indexRepl}[N/x]$, so also $M'{\indexRepl}[N/x] \ele \TypeSem{`k\arrow`n}$.
By Lem.~\ref{Typesem expansion}\,(\ref{Typesem expansion logical}), also
$(`lx.M'{\indexRepl})N \ele \TypeSem{`k\arrow`n}$.
By definition of $ \TypeSem{`k\arrow`n}$, for any $ \vect{L} \ele \TypeSem{`k}$ we have $\SN({ (`lx.M'{\indexRepl} )N \vect{L} })$; notice that $N \cons \vect{L} \ele \TypeSem{`d \prod `k}$, so
$(`lx.M'){\indexRepl} \ele \TypeSem{`d'\prod `k\arrow`n}$.

 \item[$\App$]
Then $M = PQ$, and there exists $`d$ such that $ \derLmu `G |- P : `d\prod `k\arrow`n | `D $ and $ \derLmu `G |- Q : `d | `D $.
Then by induction, we have $P{\indexRepl} \ele \TypeSem{`d\prod `k\arrow`n}$ and $Q{\indexRepl} \ele \TypeSem{`d}$; notice that $P{\indexRepl}Q{\indexRepl} = (PQ){\indexRepl}$.
Take $ \vect{L} \ele \TypeSem{`k}$, we get $Q{\indexRepl} \cons \vect{L} \ele \TypeSem{`d\prod `k}$, so also $ \SN( P{\indexRepl}Q{\indexRepl} \vect{L} ) $.
But then also $P{\indexRepl}Q{\indexRepl} \ele \TypeSem{`k\arrow `n}$.

\item[$\TCmd$]  Then $T = [`a]M$, $`s = `d\prod `k$, $`D = `a{:}`k,`D'$, and $\derLmuR `G |- M : `d | `D' $ is the premise of the rule.
In this case $([`a]M) \indexRepl = M \indexRepl \cons (\Repl `a)$, and we have that 
$M \indexRepl  \ele \TypeSem{`d}$ by induction, and $\Repl `a \ele \TypeSem{`k}$ by the hypothesis
that $\Repl$ extends $`D = `a{:}`k,`D'$. So $M \indexRepl \cons (\Repl `a) \ele \TypeSem{`d\prod `k}$ by
definition.

 \item[$`m$] 
Then $M = `m`a.\Cmd$, $\Cmd = [`b]N$, $`s = `k \arrow `n $ and the premise of the rule is
$\derLmuR `G |- \Cmd : (`k' \arrow `n) \prod `k' | `a{:}`k, `D$, where  $`a \not\in `D$.
For any $ \vect{L} \ele \TypeSem{`k}$ the parallel substitution $\Repl' = \Repl \strSubvec[ `a <= L ]$
extends $`G$ and $`a{:}`k, `D$ since $\Repl$ extends $`D$, hence by induction:
\begin{equation}\label{eq:par-subst}
\Cmd_{\Repl'} = N_{\Repl'} \cons (\Repl' `b) \ele \TypeSem{(`k' \arrow `n) \prod `k'}
\end{equation}
We then distinguish two different sub-cases.

\begin{description} \itemsep 3pt
 
 \item[$`a = `b$]
Then $M = `m`a. [`a]N$ and $`k = `k'$. If $ \vect{L} \ele \TypeSem{`k}$ then 
$N\indexRepl[ `a \By \vect{L\,} ]  \ele \TypeSem{`k\arrow`n}$ by (\ref{eq:par-subst}) above.
Since $`a$ does not occur free in the image of $\Repl$, $N\indexRepl[ `a \By \vect{L\,} ] = N{\indexRepl} \strSubvec[ `a <= L ]$, so we have $N{\indexRepl}\strSubvec[ `a <= L ]  \ele \TypeSem{`k\arrow`n} $, and therefore $ \SN( N{\indexRepl} \strSubvec[ `a <= L ] \vect{L} ) $.
But then also ${\SN(`m`a.[`a]N{\indexRepl} \strSubvec[ `a <= L ] \vect{L})}$, by Lem.~\ref{SN facts}\,(\ref{SN fact add mu abstraction}), and by Lem.~\ref{SN facts}\,(\ref{SN fact mu out redex}) $\SN ({ (`m`a.[`a]N{\indexRepl}) \vect{L} })$;  
so $ (`m`a.[`a]N){\indexRepl} \ele \TypeSem{`k\arrow`n}$.

 \item[$`a \not= `b$]
Then $`D = `b{:}`k',`D' $.
Assume $\vect{L} \ele \TypeSem{`k}$, then $\Repl\strSubvec[ `a <= L ]  $ extends $`G$ and $`a{:}`k, `b{:}`k',`D' $ 
and, by (\ref{eq:par-subst}), we have $N\indexRepl[ `a \By \vect{L\,} ] \ele \TypeSem{`k' \arrow`n} $.
Now let $\vect{Q} \ele \TypeSem{`k'}$, then $ \SN({ N {\indexRepl[ `a \By \vect{L\,} ]} \vect{Q} }) $ and then also $ \SN({ (N \vect{Q} )\indexRepl[ `a \By \vect{L\,} ] }) $.
Then $\SN(`m`a.[`b](N \vect{Q} )\indexRepl[ `a \By \vect{L\,} ])$  by Lem.~\ref{SN facts}\,(\ref{SN fact add mu abstraction}), so by Def.~\ref{TypeSem definition},  $\SN ({ `m`a.[`b](N \vect{Q} )\indexRepl[ `a \By \vect{L\,} ] }) $; then also 
$ \SN({ `m`a.[`b](N \vect{Q} ){\indexRepl} {\strSubvec[ `a <= L ] } }) $.

Then $ \SN({ (`m`a.[`b](N\vect{Q}){\indexRepl})\vect{L} }) $ by Lem.~\ref{SN facts}\,(\ref{SN fact mu in redex}).
Notice that $[`b]N{\indexRepl} \vect{Q} = [`b]N{\indexRepl} \strSubvec[ `b <= Q ] $; since $\Repl{`b} = \vect{Q}$, we can infer that $[`b]N{\indexRepl} \vect{Q} = [`b]N{\indexRepl} $, so $ \SN({ (`m`a.[`b]N){\indexRepl} \vect{L} }) $.
But then $ (`m`a.[`b]N){\indexRepl} \ele \TypeSem{`k\arrow`n} $.
\end{description}

 \item[$\inter$]
By induction.

 \item[$\leq$]
By induction and Lem.~\ref{closed for leq}.
\qed 

 \end{description}
\end{proof} 

We now come to the main result of this section, that states that all terms typeable in the restricted system are strongly normalisable.

 \begin{thm}[Typeable terms are $ \SN $] \label{thm:typableAreSN}
If $\derLmu `G |- M : `d | `D $ for some $`G$, $`D $ and $`d$, then $M \ele \SN$.
 \end{thm}

 \begin{Proof}
Let $ \Repl $ be a parallel substitution such that $ \Repl (x) = x $ for all $ x \ele \dom(`G)$ and $ \Repl (`a) = \vect{y}_{`a} $ for $ `a \ele \dom(`D) $, where the length of the stack $ \vect{y}_{`a}$ is $\length{`k}$ if $`a{:}`k \ele `D $
(notice that $\Repl$ is well defined).
By Lem.\skp\ref{TypeSem and SN lemma}, $ \Repl $ extends $`G$ and $`D $.
Hence, by Lem.\skp\ref{lem:replacement}, $M \indexRepl \ele \TypeSem{`d}$,
and then $M \indexRepl \ele \SN$ by Lem.\skp\ref{TypeSem and SN lemma}\,(\ref{TypeSem implies SN}).
Now 
 \[ \begin{array}{rcl}
M \indexRepl 
	&\equiv& 
M \,[ x_1/ x_1, \ldots, x_n/x_n, \StrSubNoB{`a_1}{\vect{y}_{`a_1}}, \ldots, \StrSubNoB{`a_m}{\vect{y}_{`a_m}}] \\
	&\equiv& 
M \,[ \StrSubNoB{`a_1}{\vect{y}_{`a_1}}, \ldots, \StrSubNoB{`a_m}{\vect{y}_{`a_m}}] 
 \end{array} \]
Then, by Prop.\skp\ref{SN facts}, for any $\vect{`b}$ also $(`m`a_1.[`b_1]\dots `m`a_m.[`b_m]M) \vect{y}_{`a_1} \dots \vect{y}_{`a_m} \ele \SN$, and therefore also $M \ele \SN$. 
 \end{Proof}

 \subsection{Strongly Normalising Terms are Typeable} \label{sub:NormType}

In this section we will show the counterpart of the previous result, \ie~that all strongly normalisable terms are typeable in our restricted intersection system.
This result has been claimed in many papers \cite{Pottinger'80,Bakel-TCS'92}, but has rarely been proven completely.

First we give the shape of terms and commands in normal forms.

 \begin{defi}[Normal Forms] 
The sets $\NF \subseteq \Terms$ (and, implicitly, $\NFC\subseteq \Commands$) of \emph{normal forms} are defined by the grammar:
 \[ \begin{array}{rcl}
N & :: =& x N_1 \cdots N_k \mid `lx.N \mid `m`a.[`b]N ~ (`b \not= `a \textrm{ or } `a \ele N 
) 
 \end{array} \]
 \end{defi}

\noindent
It is straightforward to verify that $\NF$ and $\NFC$ coincide with the sets of irreducible terms and commands, respectively.

We can now show that all terms and commands in normal form are typeable in the restricted system.

 \begin{lem} \label{prop:NFaretypeable}

 \begin{enumerate}

 \firstitem \label{prop:NFaretypeable-terms}
If $N \ele \NF$ then there exist $`G,`D $ and $`k\arrow `y \ele \LangR_D$ such that $ \derLmuR `G |- N : `k\arrow `y | `D $.

 \item \label{prop:NFaretypeable-commands}
If ${\Cmd} \ele \NFC$ then there exist $`G,`D $ and $`k \ele \LangR_C$ such that $ \derLmuR `G |- {\Cmd} : (`k\arrow `y)\prod `k | `D $.

 \end{enumerate}
 \end{lem}

 \begin{proof}
By simultaneous induction over the definitions of $\NF$ and $\NFC$.

 \begin{description}

 \item [$ N \equiv xN_1 \ldots N_k $] 
Since $N_1, \ldots,N_k \ele \NF$, by induction (\ref{prop:NFaretypeable-terms})
we have that, for all $i \seq k$, there exist $`G_i,`D _i$ and $`d_i$ such that $ \derLmuR `G_i |- N_i : `d_i | `D _i $ (the structure of each $`d_i$ plays no role in this part). 
Take $ `G = `G_1 \inter \dots \inter `G_k \inter \Set{ x{:}`d_1\prod \cdots\prod `d_k\prod `w\arrow `y }$ and $ `D = `D_1 \inter \dots \inter `D_k $.
Then, by weakening, also $ \derLmuR `G |- N_i : `d_i | `D $ for all $i \seq k$, and 
$ \derLmuR `G |- x : {(`d_1\prod \cdots\prod `d_k\prod `w)\arrow `y} | `D $.
By repeated applications of $(\App)$ we get $ \derLmuR `G |- xN_1\dots N_k : `w\arrow `y | `D $.

 \item [$ N \equiv `lx.M $] 
By induction (\ref{prop:NFaretypeable-terms}) there exist $`G$, $`d'$, and $`D $ such that $\derLmu `G,x{:}`d' |- M : `k\arrow `y | `D $ (if $x \notele \fv(M)$, we can add $x{:}`d$ by weakening, for any $`d' \ele \LangR_D$).
Then by $(\Abs)$ we obtain 
$\derLmu `G |- `lx.M : `d\prod`k\arrow `y | `D $.

 \item [$ N \equiv `m`a.{\Cmd} $] 
By induction (\ref{prop:NFaretypeable-commands}) there exist $`G$, $`k$, $`k'$, and $`D $ such that $ \derLmuR `G |- {\Cmd} : (`k\arrow `y)\prod `k | `a{:}`k',`D $ (if $`a \notele \fn(\Cmd)$, we can add $`a{:}`k'$ by weakening, for any $`k' \ele \LangR_C$). 
We get $ \derLmuR `G |- `m`a.{\Cmd} : `k'\arrow `y | `D $ by applying rule $(\MAbs)$.
	
 \item [$ {\Cmd} \equiv {[`b]N} $] 
By induction (\ref{prop:NFaretypeable-terms}) there exist $`G$, $`d = `k\arrow `y$, and $`D $ such that $ \derLmuR `G |- N : `d | `D $. 
If $`b \notele \dom(`D)$, by weakening also $ \derLmuR `G |- N : `d | `b{:}`k,`D $, and by rule $(\TCmd)$ we obtain $ \derLmuR `G |- [`b]N : (`k\arrow `y)\prod `k | `b{:}`k,`D $.
If $`b \ele \dom(`D)$ then $`D = `b{:}`k',`D' $ for some $`k'$, and we can construct:
 \[ \begin{array}{c}
\Inf	{\Inf	[\TCmd]
{\Inf[\Strength]
	{\InfBox{\derLmu `G |- N : `k\arrow `y | `b{:}`k',`D' }
	}{ \derLmu `G |- N : `k\arrow `y | `b{:}`k\inter`k',`D' }
}{ \derLmu `G |- {[`b]}N : (`k\arrow `y)\prod (`k\inter`k') | `b{:}`k\inter`k',`D' }
	 \quad
	 \InfBox{(`k\arrow `y)\prod (`k\inter`k') \seqRestr (`k\arrow `y)\prod `k}
	}{ \derLmu `G |- {[`b]}N : (`k\arrow `y)\prod `k | `b{:}`k\inter`k',`D' }
 \end{array} \]
which shows the result.
\qed

 \end{description}
 \end{proof}

In the last case, we are forced to use $(`k'\arrow `y)\prod (`k\inter`k')$ instead of $(`k'\arrow `y)\prod `k$ (otherwise we could not apply rule $(\MAbs)$ to type $ `m`a.[`b]N$) which comes at the price of weakening the assumption $`b{:}`k$ to $`b{:}`k\inter`k'$. 
However, this is not a disadvantage since we get, for example, $ \derLmuR `G |- `m`b.{[`b]}N : (`k\inter`k')\arrow `y | `D' $ which safely records in the antecedent type $`k\inter`k'$ the functionality of $N$; notice that, in fact $`k'\arrow `y \seqRestr (`k\inter`k')\arrow `y$.

We will now show that typing in the restricted system is closed under expansion, with the proviso that the term that gets substituted is typeable as well. 
We first establish that types are preserved from a contractum to the respective redex.

 \begin{lem}[Contractum expansion] \label{lem:contractum-exp}

 \begin{enumerate}

\firstitem \label{lem:contractum-exp-beta}
If $ \derLmuR `G |- M[N/x] : `d | `D $ and $ \derLmuR `G |- N : `d' | `D $, then $ \derLmuR `G |- (`l x.M)N : `d | `D $.

\item \label{lem:contractum-exp-mu}
If $\derLmuR `G |- {`m `a.\Cmd[`a \Becomes N]} : `d | `D $ and $\derLmuR `G |- N : `d' | `D $, then $ \derLmuR `G |- {(`m `a.\Cmd)N} : `d | `D $.
 \end{enumerate}
 \end{lem}

 \begin{proof}

 \begin{enumerate}

\firstitem 
As in the corresponding case of Thm.\skp\ref{Subject expansion}, observing that term types have to be in $\LangR_D$, we have that $`d$ and $`d'$ cannot be equivalent to $`w$.
Remark that in the proof of Thm.\skp\ref{Subject expansion} the fact that $`d'$ might be equivalent to $`w$ is of use omly in case $x\notele \fv{M}$.
Here we have to assume $ \derLmuR `G |- N : `d' | `D $, since otherwise $ \derLmuR `G |- (`l x.M)N : `d | `D $ is not derivable.

\item 
As in the corresponding case of Thm.\skp\ref{Subject expansion}.
We observe that by Lem.\skp\ref{gen lemma}, $\derLmuR `G |- {`m `a.\Cmd[`a \Becomes N]} : `d | `D $ and $`d \notequivRD `w$ imply that all the sub-derivations $\derLmuR `G |- \Cmd : (`k_i\arrow`n)\prod `k_i | `a{:}`k'_i,`D $ are such that $\bigwedge_{i\ele I}`k'_i\arrow`n \seqRD `d$ which implies that $(`k_i\arrow`n)\prod `k_i \notequivRC `w$.
Indeed, $(`k_i\arrow`n)\prod `k_i \equivRC `w$ only if $`k_i \equivRC `w$ for all $i$ (that is allowed) and $`n \sim_D `w$, which is not the case in either full and restricted type theories.

Note that the hypothesis $\derLmuR `G |- N : `d' | `D $ is again necessary in case $`a \notele \fn{\Cmd}$. \qed

 \end{enumerate}
 \end{proof}

It is tempting to conclude from Lem.\skp\ref{lem:contractum-exp} that if $M \red N$ by contracting a redex $PQ$ such that $ \derLmuR `G |- Q : `d' | `D $ for some $`d'$, then $ \derLmuR `G |- N : `d | `D $ implies $ \derLmuR `G |- M : `d | `D $. 
But, unfortunately, this is false even for the ordinary $`l$-calculus itself. 

The problem is that in case of $(`b)$-reduction the fact that
 \[ \begin{array}{c}
\textrm{if }\derSN `G |- M[N/x] : `s \textrm{ and } \derSN `G |- N : `t \textrm{, then }\derSN `G |- (`l x.M)N : `s 
 \end{array} \]
does not extend to arbitrary contexts $C[M[N/x]]$: \emph{i.e.}, it is false that
 \[ \begin{array}{c}
\textrm{if }\derSN `G |- C[M[N/x]] : `s \textrm{ and } \derSN `G |- N : `t \textrm{, then }\derSN `G |- C[(`l x.M)N] : `s 
 \end{array} \]
In fact, the basis $`G$ for typing $C[M[N/x]] $ changes in the proof, so induction does not apply.
Reformulating this as 
 \[ \begin{array}{c}
\textrm{if }\derSN `G_1 |- C[M[N/x]] : `s \textrm{ and }\derSN `G_2 |- N : `t \textrm{, then }\derSN `G_1\inter `G_2 |- C[(`l x.M)N] : `s 
 \end{array} \]
(which has been used by various authors)
gives no improvement; the problem is that a free variable in $N$ might be bound in the context, which implies that the derived type might change (get bigger in the sense of $\leq$), and that $`G_1\inter `G_2$ is not a correct basis for $C[(`l x.M)N]$, since it contains types for bound variables.
This suggests the property 
 \[ \begin{array}{c}
\textrm{if }\derSN `G_1 |- C[M[N/x]] : `s \textrm{ and } \derSN `G_2 |- N : `t , \\ \textrm{then there exists }`G_3 \leq `G_1, `r \geq `s \textrm{ such that }\derSN `G_3 |- C[(`l x.M)N] : `r 
 \end{array} \]
but this is also not achievable, since the use of $\leq$ in the derivation for $P$ in $MP$ (where $MP \reduces MQ$) forces a $\geq$ step on $M$ which is not always achievable; this throws the proof for the case of application irreparably out of kilter.

To illustrate this in the context of our system, take $`l x. (`l y. x)(xx)$ which reduces to $`l x.x$, which we can type as follows:
 \[ \begin{array}{c}
\Inf[\LAbs]
	{\Inf[\Ax] {} {\derLmu x{:}`w\arrow `y |- x : `w\arrow `y | }
	 }{ \derLmu |- `l x.x : (`w\arrow `y)\prod `w\arrow `y | }
 \end{array} \]
We cannot infer this type for $`l x. (`l y. x)(xx)$ in the restricted system.
To type the sub-term $xx$ when rule $(`w)$ is not available, the best we can do is (setting $`d = `k\arrow `y$):
 \[ \begin{array}{c}
\Inf	[\App]
	{\Inf	[\seqRestr]
{\Inf	[\Ax]
	{ \derLmu x{:}(`d\prod `k\arrow `y)\inter `d |- x : (`d\prod `k\arrow `y)\inter `d | {} }
}{\derLmu x{:}(`d\prod `k\arrow `y)\inter `d |- x : `d\prod `k\arrow `y | {} }
	 \quad 
	 \Inf	[\seqRestr]
{\Inf	[\Ax]
	{ \derLmu x{:}(`d\prod `k\arrow `y)\inter `d |- x : (`d\prod `k\arrow `y)\inter `d | {} }
}{\derLmu x{:}(`d\prod `k\arrow `y)\inter `d |- x : `d | {} }
	}{ \derLmu x{:}(`d\prod `k\arrow `y)\inter `d |- xx : `d | {} }
 \end{array} \]
Using the same type for $x$, for $`ly.x$ we can construct:
 \[ \begin{array}{c}
\Inf	[\LAbs]
	{\Inf	[\seqRestr]
{\Inf	[\Ax]
	{ \derLmu y{:}`d,x{:}(`d\prod `k\arrow `y)\inter `d |- x : (`d\prod `k\arrow `y)\inter `d | {} }
}{ \derLmu y{:}`d,x{:}(`d\prod `k\arrow `y)\inter `d |- x : `k'\arrow `y | {} }
	}{ \derLmu x{:}(`d\prod `k\arrow `y)\inter `d |- `ly.x : `d\prod `k'\arrow `y | {} }
 \end{array} \]
for any $`k'$ such that $ (`d\prod `k\arrow `y)\inter `d \seqRestr `k'\arrow `y $, and therefore
 \[ \begin{array}{c}
\Inf	[\LAbs]
	{\Inf	[\App]
{\InfBox{ \derLmu x{:}(`d\prod `k\arrow `y)\inter `d |- `ly.x : `d\prod `k'\arrow `y | {} }
 \qquad 
 \InfBox{ \derLmu x{:}(`d\prod `k\arrow `y)\inter `d |- xx : `d | {} }
}{ \derLmu x{:}(`d\prod `k\arrow `y)\inter `d |- (`l y. x)(xx) : `k'\arrow `y | {} }
	 }{ \derLmu |- `lx.(`ly.x)(xx) : ((`d\prod `k\arrow `y)\inter `d)\prod `k'\arrow `y | {} }
 \end{array} \]
So we cannot infer the type $(`w\arrow `y)\prod `w\arrow `y$ for $`l x. (`l y. x)(xx)$ and a general subject-expansion result doesn't hold 
in the restricted system, not even by requiring typability of all substituted subterms.

We can however show the weaker but sufficient statement that although assignable types are not preserved, \emph{typability} is: we cannot type $`l x. (`l y. x)(xx)$ with $(`w\arrow `y)\prod `w\arrow `y$, but still we can type it in the restricted system.

We will established this through considering \emph{leftmost-outermost} reduction ($\Lor$), following a suggestion by Betti Venneri; this technique is also the one used in \cite{Krivine-book'93}, and in \cite{Bakel-NDJFL'04} in the context of the strict type assignment system of \cite{Bakel-TCS'92}; \cite{Bakel-ACM'11} presents the proof for the system with strict types and a co-variant type inclusion relation of \cite{Bakel-TCS'95}.

 \begin {defi} 
An occurrence of a redex $\Redex = ( `l x.P)Q$ or $(`m`a.[`b]P)Q$ in a term $M$ is called the \emph{left-most outer-most redex of $M$} ($\Lor(M)$), if and only if:

 \begin {enumerate}
 \item there is no redex $\Redex'$ in $M$ such that $\Redex' = \Cont[ \Redex ] $
with $\Cont[ - ]\not= [-]$ (\emph{outer-most}) ;
 \item there is no redex $\Redex'$ in $M$ such that $M = \Cont_{0}[{ \Cont_{1}[\Redex'] \, \Cont_{2}[ \Redex ]}] $ (\emph{left-most}).
 \end {enumerate}
We write $M \lored N$ when $M$ reduces to $N$ by contracting $\Lor(M)$.
 \end {defi}

The correct subject expansion result (with respect to $\Lor$, Lem.\skp\ref{l.o. lemma}) is now the one used for the proof that all strongly normalisable terms are typeable, which uses induction on the length of the $\Lor$-path.


 \begin {lem} \label {l.o. lemma}
Assume $M \lored N$ with $\Lor(M) = PQ$. 
If $\derLmuR `G |- N : `d | `D $, and 
 $\derLmuR `G |- Q : `d'' | `D $,
then there exist $`G'$, $`D'$ and $`d'$ such that $\derLmuR `G' |- M : `d' | `D' $.
 \end{lem}

 \begin{proof}
First we observe that, since $`d \ele \LangR_D$, $`d = \bigwedge_{i\ele I} `k_i\arrow`y$ for some $I$ so that certain $`k_i \ele \LangR_C$, and $`d \not\sim_D`w$ (in the full type theory).
Hence $\derLmuR `G |- N : `k_i\arrow`y | `D $ for all $i\ele I$. 
Hence it suffices to show the thesis when $`d$ is an arrow type.
We reason by induction over the structure of terms.

 \begin{description}

\item [$ M \same VP_1\dots P_n $] 
	We distinguish two sub-cases:	

 	\begin {description}

 	\item [$V \same \Lor(M)$] Then $M \lored V'P_1\dots P_n \same N$ and $V'$ is the contractum of $V \same PQ$.
By Lem.\skp\ref{gen lemma} we know that $\derLmuR `G |- V'P_1\dots P_n : `k\arrow`y | `D $ implies that 
 \[ \derLmuR `G |- V' : `d_{j,1}\prod \cdots\prod `d_{j,n}\prod `k\arrow`y | `D \] 
for all $j\ele J$, and $\derLmuR `G |- P_h: `d_{j,h} | `D $ for all $h \ele \n$. 
Since $V'$ is the contractum of the redex $PQ$ and $\derLmuR `G |- Q : `d'' | `D $, we can apply Lem.\skp\ref{lem:contractum-exp} and get 
 \[ \derLmuR `G |- V : d_{j,1}\prod \cdots\prod `d_{j,n}\prod `k\arrow`y | `D , \]
from which, repeatedly applying rule $(\App)$, we get $\derLmuR `G |- VP_1\dots P_n : `k\arrow`y | `D $. Take $`G' = `G$, $`d'= `k\arrow`y$ and $`D' = `D $.
 
 \item [$V \same z$] Then $\Lor(M) = \Lor(P_h) $ for some $h \ele \n$, and $N \same VP_1\dots P'_h \dots P_n$ with $P_h \lored P'_h$. 
Reasoning as in the previous case we get $\derLmuR `G |- P'_h: `d_{j,h} | `D $ and 
 \[ \begin{array}{rcl}
`G(z) &\seq& `d_{j,1}\prod \dots\prod `d_{j,h}\prod \dots\prod `d_{j,n}\prod `k\arrow`y.
 \end{array} \] 
By induction there exist $`G_1$, $`d'_{j,h}$ and $`D_1$ such that $\derLmuR `G_{j,h} |- P_h: `d'_{j,h} | `D_{j,h} $. 
We then set 
 \[ \begin{array}{rcl}
`G' &=&`G \ \inter \ \bigwedge_{j\ele J, h\ele H_j}`G_{j,h} \ \inter \ 
	\Set{ z{:}`d_{j,1}\prod \dots\prod `d'_{j,h}\prod \dots\prod `d_{j,n}\prod `k\arrow `y} \\
`D' &=& `D \ \inter \ \bigwedge_{j\ele J, h\ele H_j}`D_{j,h} 
 \end{array} \]
Then by applying rules $(\Strength)$ and $(\App)$ we conclude $\derLmuR `G' |- z\,P_1\dots P_n : `k\arrow`y | `D' $ as desired.	 
 \end{description}
	
 \item [$ M \same `ly.M' $] 
If $M \lored N$, then $N = `ly.N'$ and $M' \lored N'$.
Then there exists $`d_j$ and $`k_j$ such that $`d = \bigwedge_{j\ele J} `d_j\prod `k_j\arrow `y$, and $\derLmuR `G,y{:}`d_j |- N' : `k_j\arrow `y | `D $ for all $j\ele J$.
By induction, there exists $`G''$, $`D''$, $`d'_j$, and $`k'_j$ such that $\derLmuR `G'',y{:}`d'_j |- M' : `k'_j\arrow `y | `D'' $.
Then, by rule $(\Abs)$, $\derLmuR `G'' |- `ly.M' : `d'_j\prod `k'_j\arrow `y | `D'' $ for all $j\ele J$ so that $\derLmuR `G'' |- `ly.M' : \bigwedge_{j\ele J}`d'_j\prod `k'_j\arrow `y | `D'' $;  take $`G' = `G''$, $`D' = `D''$, and $`d' = \bigwedge_{j\ele J}`d'_j\prod `k'_j\arrow `y$.


 \item [{$M \same `m`a.[`b]M'$}]
 If $M \lored N$, then $N = `m`a.[`b]N'$ and $M' \lored N'$. By Lem.\skp\ref{gen lemma}, using the fact that $`d$ is non-trivial by assumption, and assuming $`a\neq`b$ we have that 
 \[
\derLmu `G |- N' : \bigwedge_{j\ele J, h\ele H_j}`k_{j,h}\arrow `y | `a{:}`k'_j, `b{:}`k'_{j,h},`D ,
 \]
where $ \bigwedge_{j\ele J, h\ele H} `k_{j,h}\arrow `y \seq_D `k_j\arrow`y$ for all $j\ele J$ and $\bigwedge _{j\ele J}`k_j\arrow`y \seq_D `d$.
 
By induction there exist $`G_{j,h}$, $`d'_{j,h} = \bigwedge_{j\ele J, h\ele H_j, u\ele U_{j,h}} `k'_{j,h,u}\arrow`y$, $`k''_j$, $`k''_{j,h}$ and $`D_{j,h}$ such that 
 \[
\derLmuR `G_{j,h} |- N' : {\bigwedge_{j\ele J, h\ele H_j, u\ele U_{j,h}} `k'_{j,h,u}\arrow`y} | `a{:}`k''_j, `b{:}`k''_{j,h},`D_{j,h} .
 \] 
Taking
 $ \begin{array}[t]{rcl}
`G' &=&`G \ \inter \ \bigwedge_{j\ele J, h\ele H_j}`G_{j,h} \\
`D' &=& `D \ \inter \ \bigwedge_{j\ele J, h\ele H_j}`D_{j,h} 
\ \inter \ \Set{`a{:} \bigwedge_{j\ele J, h\ele H_j}`k''_j, `b{:} \bigwedge_{j\ele J, h\ele H_j, u\ele U_{j,h}}(`k'_{j,h,u} \inter`k''_{j,h})}
 \end{array} $

\noindent
we have that $\derLmuR `G' |- {[`b]N'} : (`k'_{j,h,u}\arrow`y)\prod `k'_{j,h,u} | `D' $ for all $j\ele J, h\ele H_j, u\ele U_{j,h}$, possibly using rule $(\Strength)$; we get the result by applying rules $(\MAbs)$ and $(\inter)$.
	
 \item [{$M \same `m`a.[`a]M'$}]
This case is similar to the previous one and easier. \qed
	
 \end{description}
 \end{proof}

We observe that considering leftmost-outermost reduction in the proof above is crucial to rule out the case $M \same V P_1 \dots P_j \dots P_n \red V P_1 \dots P'_j \dots P_n$ because $P_j\red P'_j$ where $V \same PQ$ is a redex, hence different than a variable. 
In fact, the induction hypothesis now tells us that if $P'_j:`d'_j$ then $P_j:`d_j$ for some $`d_j$ which is in general unrelated to $`d'_j$; now nothing ensures that $`d_j$ is compatible with the type of $V$. 
Instead, we can solve the problem for $V \same z$ by taking a suitable weaker assumption for the typing of $z$, which is the same trick to circumvent the difficulty noted before the last lemma, that arises both with $`l$ and $`m$-abstraction.

\newpage

We can now show that all strongly normalisable $`l`m$-terms are typeable in the restricted system.

 \begin{thm}[Typeability of $ \SN$-Terms] \label{thm:typeableAreSN}
For all $M \ele \SN$ there exist $`G$ and $`D $ and a type $`d$ such that $ \derLmuR `G |- M : `d | `D $.
 \end{thm}

 \begin{proof}
First observe that $\Lor$ is normalising. 
Then we reason by induction on the maximum of the lengths of reduction sequences for a strongly normalisable term $M$ to its normal form (denoted by $\#(M)$).

 \begin {description}
 \item [$ \#(M) = 0 $] 
Then $M$ is in normal form, and by Lem.\skp\ref{prop:NFaretypeable}, there exist $`G$ and $ `d$ such that $\derLmuR `G |- M : `d | `D $.

 \item [$ \#(M)\geq 1 $] 
Let $M\lored N$ by contracting the redex $PQ = \Lor(M)$, then $\#(N) < \#(M)$.
Then $\#(Q) < \#(M)$ (since $Q \ele \SN$ is a proper subterm of a redex in $M$). Then by induction there exist $`G_1$, $`G_2$, $`D_1$, $`D_2$, $`d_1$, and $`d_2$ such that $\derLmuR `G_1 |- N : `d_1 | `D_1 $ and $\derLmuR `G_2 |- Q : `d_2 | `D_2 $.
By Lem.\skp\ref {l.o. lemma} there exist $`G$, $`D $, and $ `d$ such that $\derLmuR `G |- M : `d | `D $.%
\qed

 \end{description}
 \end{proof}

 \section{Simply typed $`l`m$-calculus and Intersection Types} \label{sec:Parigot}

In the previous sections we considered the $`l`m$-calculus as a type-free calculus; in this section we will show that we can establish a connection between our intersection type assignment system and Parigot's logical assignment system. 
We will show that logical formulas translate into types of the appropriate sort, and moreover that this can be done in the restricted system of Sect.\skp\ref{sec:character}. 
Hence, the characterisation result carries over and can be used to establish the strong normalisation property of Parigot's calculus.


We use a version of Parigot's logical system as presented in \cite{Parigot'92}, which is equivalent to the original one if just terms (so not also proper commands, \ie~elements of $\Commands$) are typed. This implies that the rule for $\bot$ does not need to be taken into account.%

We briefly recall Parigot's first-order type assignment system, that we call the Simply-Typed $`l`m$-calculus.

 \begin{defi}[Parigot's Simply Typed $`l`m$-calculus] \label{def:ParigotSimpleTypes}

 \begin{enumerate}

 \firstitem
The set $\Formulas$ of \emph{Logical Formulas} is defined by the following grammar:
 \[ \begin{array}{rcl}
A,B & :: =& `v \mid A\arrow B
 \end{array} \]
where $`v$ ranges over an infinite, denumerable set of \emph{Proposition (Type) Variables}.

 \item
 \emph{Judgements} are of the form $\derLmuP `P |- M : A | `S $, where $M \ele\Terms$; 
$`P$ and $`S$ are finite mappings from $\TVar$ and $\CVar$, respectively, to formulas, and are normally written as finite sets of pairs of term variables and formulas and of names and formulas respectively, as in $`P = \Set{x_1{:}A_1,\ldots,x_n{:}A_n}$ and $`S = \Set{`a_1{:}B_1,\ldots,`a_m{:}B_m}$.

 \item
The \emph{inference rules} of this system are:
 \[ \begin{array}{rl@{\dquad}rl@{\dquad}rl}
(\Ax) : &
\Inf	{ \derLmu `P , x{:}A |- x : A | `S }
&
(`m_1): &
\Inf	{ \derLmu `P |- M : A | `a{:}A,`S
	}{ \derLmu `P |- `m`a.[`a]M : A | `S }
&
(`m_2): &
\Inf	{\derLmu `P |- M : B | `a{:}A,`b{:}B,`S
	}{ \derLmu `P |- `m`a.[`b]M : A | `b{:}B,`S }
 \end{array} \]
 \[ \begin{array}{rl@{\dquad}rl}
(\arrI): &
\Inf	
	{ \derLmu `P,x{:}A |- M : B | `S
	 }{ \derLmu `P |- `l x.M : A\arrow B | `S }
&
(\arrE): &
\Inf	
	{\derLmu `P |- M : A\arrow B | `S
	 \quad
	 \derLmu `P |- N : A | `S
	 }{ \derLmu `P |- MN : B | `S }
 \end{array} \]

We write $ \derLmuP `G |- M : A | `S $ to denote that this judgement is derivable in this system.

 \end{enumerate}
 \end{defi}

Through the Curry-Howard correspondence (\emph{formulas as types} and \emph{proofs as terms}), the underlying logic of this system is the \emph{minimal classical logic} (\cite{Ariola-Herbelin'03}).

Comparing Parigot's system with ours we observe that rules $(\arrI)$ and $(\arrE)$ bear some similarity with $(\LAbs)$ and $(\App)$, and rules $(`m_1)$ and $(`m_2)$ are similar to a combination of $(\MAbs)$ and $(\TCmd)$:

 \begin{lem} \label{lem:derivableMuAbs}
The following rules are derivable in the system of Def.\skp\ref{def:intersTypeAss} (and in the restricted system as well):
 \[ \begin{array}{c@{\hspace{6mm}}c}
\Inf	[\MAbs_1]
	{\derLmu `G |- M: `k\arrow `y | `a{:}`k,`D 
	}{ \derLmu `G |- `m`a.[`a]M: `k\arrow `y | `D }
&
\Inf	[\MAbs_2]
	{\derLmu `G |- M: `k'\arrow `y | `a{:}`k, `b{:}`k',`D 
	}{ \derLmu `G |- `m`a.[`b]M: `k\arrow `y | `b{:}`k',`D }
 \end{array} \]
 \end{lem} 

 \begin{proof} Consider the derivations:
 \[ \begin{array}{c@{\hspace{6mm}}c}
\Inf	[\MAbs]
	{\Inf	[\TCmd] 
{\InfBox{\derLmu `G |- M: `k\arrow `y | `a{:}`k,`D } 
}{ \derLmu `G |- {[`a]M} : (`k\arrow `y)\prod `k | `a{:}`k,`D }
	}{ \derLmu `G |- `m`a.[`a]M: `k\arrow `y | `D }
&
\Inf	[\MAbs]
	{\Inf	[\TCmd] 
{\InfBox{ \derLmu `G |- M: `k'\arrow `y | `a{:}`k, `b{:}`k',`D } 
}{ \derLmu `G |- {[`b]M} : (`k'\arrow `y)\prod `k' | `a{:}`k, `b{:}`k',`D }
	}{ \derLmu `G |- `m`a.[`b]M: `k\arrow `y | `b{:}`k',`D }
 \end{array} \]
\arrayqed[-20pt]
 \end{proof}


As an example illustrating the fact that the system for $`l`m$ is more expressive than the simply typed $`l$-calculus, 
we consider the following proof of Peirce's Law, which we take from \cite{Ong-Stewart'97}:
 \[ \begin{array}{c}
\Inf	[\arrowI]
	{\Inf	[`m_1]
		{\Inf	[\arrowE]
			{\Inf	[ \Ax]
				{\derLmu x{:}(A\arrow B)\arrow A |- x : (A\arrow B)\arrow A | `a{:}A } 
			\Inf	[\arrowI]	
				{\Inf	[`m_2]
					{\Inf	[ \Ax]
						{\derLmu x{:}(A\arrow B)\arrow A,y{:}A |- y : A | `a{:}A,`b{:}B }
					}{ \derLmu x{:}(A\arrow B)\arrow A,y{:}A |- `m`b.[`a]y : B | `a{:}A }	 
				}{ \derLmu x{:}(A\arrow B)\arrow A |- `ly.`m`b.[`a]y : A\arrow B | `a{:}A } 
			}{ \derLmu x{:}(A\arrow B)\arrow A |- x(`ly.`m`b.[`a]y) : A | `a{:}A }
		}{ \derLmu x{:}(A\arrow B)\arrow A |- `m`a.[`a](x(`ly.`m`b.[`a]y)) : A | {} }
	}{ \derLmu {} |- `lx.`m`a.[`a](x(`ly.`m`b.[`a]y)) : ((A\arrow B)\arrow A)\arrow A | {} }
 \end{array} \]
We observe that the term $`lx.`m`a.[`a](x(`ly.`m`b.[`a]y))$ is typable in the restricted type system (and hence in the full system as well) by a derivation with the very same structure. 
Indeed, let $`d_{A\arrow B} \ByDef `k_{A\arrow B}\arrow `y$, where $`k_{A\arrow B} \ByDef (`k_A\arrow `y)\prod `k_B$ and $`k_A,`k_B \ele \LangR_C$ 
are arbitrary; 
set further $`k_{(A\arrow B)\arrow A} \ByDef `d_{A\arrow B}\prod `k_A$ and $`d_{(A\arrow B)\arrow A} \ByDef `k_{(A\arrow B)\arrow A}\arrow `y$ ($ \ByDef `d_{A\arrow B}\prod `k_A\arrow `y $), then

(with $`G = x{:}((`k_A\arrow `y)\prod `k_B\arrow `y)\prod `k_A\arrow `y $)
 \[ \begin{array}{c} 
\Inf	[\LAbs]
	{\Inf	[\MAbs_1]
		{\Inf	[\App]
			{\Inf	[ \Ax]
				{ \derLmu `G |- x : ((`k_A\arrow `y)\prod `k_B\arrow `y)\prod `k_A\arrow `y | `a{:}`k_A }
	 		\Inf	[\LAbs]
				{\Inf	[\MAbs_2]
					{\Inf	[ \Ax]
						{\derLmu `G,y{: }{ `k_A\arrow `y} |- y : `k_A\arrow `y | `a{:}`k_A,`b{:}`k_B }
					}{ \derLmu `G,y{: }{ `k_A\arrow `y} |- `m`b.[`a]y : `k_B\arrow `y | `a{:}`k_A }
				}{ \derLmu `G |- `ly.`m`b.[`a]y : (`k_A\arrow `y)\prod `k_B\arrow `y | `a{:}`k_A }
			}{ \derLmu `G |- x(`ly.`m`b.[`a]y) : `k_A\arrow `y | `a{:}`k_A } 
		}{ \derLmu `G |- `m`a.[`a](x(`ly.`m`b.[`a]y)) : `k_A\arrow `y | {} }
	 }{ \derLmu {} |- `lx.`m`a.[`a](x(`ly.`m`b.[`a]y)) : (((`k_A\arrow `y)\prod `k_B\arrow `y)\prod `k_A\arrow `y)\prod `k_A\arrow `y | {} }
 \end{array} \]


As suggested in the example above, we can interpret formulas into intersection types of our system. Notably types from the restricted language $\LangR$ do suffice. 

 \begin{defi} 
The translation functions $\cdot^D:\Formulas\To \LangR_D$
and $\cdot^C:\Formulas\To \LangR_C$ are defined by: 
 \[ \begin{array}{rcl}
 `v^C &\ByDef& `y\prod `w 
 \\
(A\arrow B)^C &\ByDef& (A^C\arrow `y)\prod B^C
 \\
A^D &\ByDef& A^C\arrow `y
 \end{array} \]
We extend these mappings to bases and name contexts by: $`P^D = \Set{x{:}A^D \mid x{:}A \ele `P}$ and $`S^C = \Set{`a{:}A^C \mid `a{:}A \ele `S}$. 
 \end{defi}
It is straightforward to show that the above translations are well defined.

 \begin{thm}[Derivability preservation] \label{thm:translation}
If $ \derLmuP `P |- M:A | `S $ then $ \derLmu `P^{D } |- M:A^{D } | `S^{C} $.
 \end{thm}

 \begin{proof} 
Each rule of the simply-typed $`l`m$-calculus has a corresponding one in the restricted intersection type system; hence it suffices to show that rules are preserved when translating formulas into types. 
The case of $(\Ax)$ is straightforward. 

 \begin{description} \itemsep 6pt
\def\InfBox{}
 \item [$ \arrI $] 
 $ \begin{array}[t]{ccc}
\Inf	[\LAbs]
	{\InfBox{ \derLmu `P^D, x{:}A^D |- M : B^D | `S^C }
	}{ \derLmu `P^D |- `lx.M : (A\arrow B)^D | `S^C }
& \ByDef &
\Inf	[\LAbs]
	{\InfBox{ \derLmu `P^D,x{:}A^C\arrow `y |- M : B^C\arrow `y | `S^C }
	}{ \derLmu `P^D |- `lx.M : {((A^C\arrow `y)\prod B^C)\arrow `y} | `S^C }
 \end{array} $

 \item [$ \arrE $] 
 $ \begin{array}[t]{l}
\Inf	[\App]
	{\InfBox{ \derLmu `P^D |- M : (A\arrow B)^D | `S^C }
	 \quad 
	 \InfBox{ \derLmu `P |- N : A^D | `S^C }
	}{ \derLmu `P |- MN : B^D | `S } 
 \quad \ByDef 
 \\ [6mm]
\Inf	[\App]
	{\InfBox{ \derLmu `P^D |- M : {((A^C\arrow `y)\prod B^C)\arrow `y} | `S^C }
	 \quad 
	 \InfBox{ \derLmu `P |- N : {A^C\arrow `y} | `S^C }
	}{ \derLmu `P^D |- MN : {B^C\arrow `y} | `S^C } 
 \end{array} $

 \item [$ `m_1 $] 
 $ \begin{array}[t]{ccc}
\Inf	[\MAbs_1]
	{\InfBox{ \derLmu `P^D |- M : A^D | `a{:}A^C,`S^C }
	}{ \derLmu `P^D |- `m`a.[`a]M : A^D | `S^C }
& \ByDef &
\Inf	[\MAbs_1]
	{\InfBox{ \derLmu `P^D |- M : {A^C\arrow `y} | `a{:}A^C,`S^C }
	}{ \derLmu `P^D |- `m`a.[`a]M : {A^C\arrow `y} | `S^C }
 \end{array} $

 \item [$ `m_2 $] 
 $ \begin{array}[t]{ccc}
\Inf	[\MAbs_2]
	{\InfBox{ \derLmu `P^D |- M : B^D | `a{:}A^C,`b{:}B^C,`S^C }
	}{ \derLmu `P^D |- `m`a.[`b]M : A^D | `b{:}B^C,`S^C }
& \ByDef &
\Inf	[\MAbs_2]
	{\InfBox{ \derLmu `P^D |- M : {B^C\arrow `y} | `a{:}A^C,`b{:}B^C,`S^C }
	}{ \derLmu `P^D |- `m`a.[`b]M : {A^C\arrow `y} | `b{:}B^C,`S^C }
 \end{array} $
\arrayqed
 \end{description}
 \end{proof}

Strong normalisation of typeable terms in Parigot's simply-typed $`l`m$-calculus (first proved in \cite{Parigot'92}) now follows as a corollary of our characterisation result.

 \begin{cor}[Strong Normalisability of Parigot's Simply Typed $`l`m$-calculus] 
If $\derLmuP `P |- M:A | `S $, then $M \ele \SN$.
 \end{cor}

 \begin{Proof} By Thm.\skp\ref{thm:translation} and Lem.\skp\ref{lem:derivableMuAbs}, if $\derLmuP `P |- M : A | `S $ then $\derLmuR`P^D |- M : A^D | `S^C $. 
That $M$ is $\SN$ now follows from Thm.\skp\ref{thm:typeableAreSN}.
 \end{Proof}

We end this section by observing that negation can be added to the syntax of logical formulas without disrupting the correspondence with the (restricted) intersection types. Let $\bot$ be added to the formula syntax as a special atom for falsehood. Then consider the logical system which is obtained from Def.\skp\ref{def:ParigotSimpleTypes} by replacing rules $(`m_1)$ and $(`m_2)$ by:
%
 \[ \begin{array}{c@{\hspace{1cm}}c}
\Inf	[\textit{Activate}]
	{ \derLmu `P |- {\Cmd} : {\bot} | `a{:}A,`S
	}{ \derLmu `P |- {`m`a.\Cmd} : A | `S }
 & 
\Inf	[\textit{Passivate}]
	{\derLmu `P |- M : A | `a{:}A,`S
	}{ \derLmu `P |- [`a]M : {\bot} | `a{:}A,`S }
 \end{array} \]
The resulting system is the same as in \cite{Bierman'98}, and in \cite{Parigot00} Section~3.1 but with two contexts instead of a single one. 
The only difference is that in rule $(\textit{Passivate})$ we require the assumption $`a{:}A$ in the name contexts, which is just added to the conclusion in both \cite{Bierman'98} and \cite{Parigot00} since there contraction and weakening are implicitly assumed.

In this new system, the rules $(`m_1)$ and $(`m_2)$ are admissible. By defining $\neg A \ByDef A\arrow\bot$ and considering all the formulas in the $`S$ context as negated, rule $(\textit{Activate})$ corresponds to the \emph{reductio ad absurdum} rule of classical logic. On the other hand rule $(\textit{Passivate})$ can be read as the $\neg$-elimination rule, saying that from $A$ and $\neg A$ we get falsity.

To interpret $\bot$ into our intersection type system we need to keep track of the contradiction from which it arises; therefore we write the slightly different rule:
 \[ \begin{array}{c}
\Inf	[\textit{Passivate}']
	{\derLmu `P |- M : A | `a{:}A,`S
	}{ \derLmu `P |- [`a]M : \bot_A | `a{:}A,`S }
 \end{array} \]
where $\bot_A$ is a new constant for each formula $A$. Now, by adding
 \[ \begin{array}{rcl}
 \bot_A^C &\ByDef& (A^C\arrow `y)\prod A^C
 \end{array} \]
to the translation, we obtain:
 \[ \begin{array}{ccc}
\Inf	{ \derLmu `P^D |- {\Cmd} : \bot_B^C | `a{:}A^C,`S^C
	}{ \derLmu `P^D |- {`m`a.\Cmd} : A^D | `S^C }
& \ByDef &
\Inf	[\MAbs]
	{\derLmu `P^D |- {\Cmd} : (B^C\arrow `y)\prod B^C | `a{:}A^C,`S^C
	}{ \derLmu `P^D |- {`m`a.\Cmd} : A^C\arrow `y | `S^C } 
 \\ [8mm]
\Inf	{\derLmu `P^D |- M : A^D | `a{:}A^C,`S^C
	}{ \derLmu `P^D |- [`a]M : \bot_A^C | `a{:}A^C,`S^C }
& \ByDef &
\Inf	[\TCmd]
	{\derLmu `P^D |- M : A^C\arrow `y | `a{:}A^C,`S^C
	}{ \derLmu `P^D |- [`a]M : (A^C\arrow `y)\prod A^C | `a{:}A^C,`S^C }
 \end{array} \]

 \section{Related work} \label{sec:related}

The starting point of the present work is \cite{Bakel-ITRS'10}, where a type assignment system for $`l`m$-terms with intersection and union types was proved to be invariant under reduction and expansion. 
That system was without apparent semantical justification, which motivated the present new construction. 
With respect to the system of \cite{Bakel-ITRS'10}, our system does not use union types, and introduces product types for continuations, which is its main characteristic. 
The introduction of product types is inspired to the continuation model of \cite{Streicher-Reus'98}, which is the main source of this paper, together with \cite{BCD'83}, which contains the first construction of a $`l$-model as a filter model. 
However, as explained in the introduction and in the main body of the paper, we have followed the inverse path, from the model to the type system, building over \cite{Coppo-et.al'84} and \cite{Abramsky'91}.

In \cite{Bakel-Barbanera-deLiguoro-TLCA'11} we conjectured that in an appropriate subsystem of the present one, it should be possible to type exactly all strongly normalising $`l`m$-terms.
We established that result in \cite{BakBdL-ITRS12}, though via a less elegant variant of our system than the restricted system in Sect.\skp\ref{sec:character}.
The first to state the characterisation result for the $`l$-calculus was Pottinger \cite{Pottinger'80}, using a notion of type assignment similar to the intersection system of \cite{Coppo-Dezani'78,Coppo-Dezani-Venneri'80}, but extended in that it is also closed for $`h$-reduction, which is equivalent to adding a co-variant type inclusion relation; in particular, this is a system that is defined without the type constant $`w$.
However, to show that all typeable terms are strongly normalisable, \cite{Pottinger'80} only \emph{suggests} a proof using Tait's computability technique \cite{Tait'67}.
A detailed proof, using computability, in the context of the $`w$-free BCD-system \cite{BCD'83} is given in \cite{Bakel-TCS'92}; to establish the same result, saturated sets are used by Krivine in \cite{Krivine-book'93} (Ch.\skp4), in Ghilezan's survey \cite{Ghilezan'96}, and in \cite{Bakel-ACM'11}.

The converse of that result, namely the property that all strongly normalisable terms are typeable has proven to be more elusive: it has been claimed in many papers but not shown in full (we mention \cite{Pottinger'80,Bakel-TCS'92,Ghilezan'96}); in particular, the proof for the property that type assignment is closed for subject expansion (the converse of subject reduction) is dubious.
Subject expansion can only reliably be shown for \emph{left-most outermost} reduction, which is used for the proofs in \cite{Krivine-book'93,Bakel-Dezani-LATIN'02,Bakel-NDJFL'04,Bakel-ACM'11}, and our result follows that approach.

The translation in Sect.\skp\ref{sec:Parigot} that maps simple types into our extension of intersection types is a form of negative translation; in \cite{KikSak14} it is extended to the system in \cite{Bakel-ITRS'10}, thereby relating the original intersection and union type assignment system for $`l`m$ to ours.

The model in \cite{Streicher-Reus'98} is not a model of de Groote and Saurin's $`L`m$-calculus, but a variant of it, dubbed a `stream model' in \cite{NakazawaK-CLC12}; it provides a sound interpretation of the extended calculus. 
Building over stream models, in \cite{deLiguoro:ApproxLM12} it has been proven that the same type theory, but with different rules of the type assignment, gives a finitary description of the model matching the reduction in the stronger sense that the approximation theorem holds.

In \cite{NakazawaN14} an extension of $`L`m$ is considered, called $`L`m_{\sf cons}$. 
A type assignment for $`L`m_{\sf cons}$ based on \cite{deLiguoro:ApproxLM12}'s type system is proposed, and subject reduction and strong normalisation of the reduction on the typed $`L`m_{\sf cons}$ are proven. 
In \cite{Nakazawa-ITRS14} the same system is shown to enjoy Friedman's theorem.

 \section*{Conclusions} \label{conclusions}
We have presented a filter model for the $`l`m$-calculus which is an instance of Streicher and Reus's continuation model, and a type assignment system such that the set of types that can be given to a term coincides with its denotation in the model. 
The type theory and the type assignment system can be viewed as the logic for reasoning about the computational meaning of $`l`m$-terms, much as is the case for $`l$-calculus. 

By restricting the assignment system to a subset of the intersection types we have obtained an assignment system where exactly the strongly normalising $`l`m$-terms are typeable. 

Finally, we have given a translation of intersection types into logical formulas and proved that if a term is typeable in Parigot's type assignment system for $`l`m$, then it is typeable by its translation in the restricted intersection type system. 
As a by-product we have a new proof that proof terms in the $`l`m$-calculus are strongly normalising.


 \end{document}